\documentclass[conference]{IEEEtran}
\pdfoutput=1
\usepackage{amsmath}
\usepackage{amssymb}
\usepackage{graphicx}
\usepackage[tight,footnotesize]{subfigure}
\usepackage{url}
\usepackage{color}
\usepackage{ifthen}
\usepackage{balance}
\usepackage{cite}
\usepackage{balance}
\usepackage{stfloats}
\usepackage{multirow}
\usepackage{mathtools}
\usepackage{cleveref}
\usepackage{multicol}
\usepackage{caption}
\newtheorem{lemma}{Lemma}

\begin{document}
\title{H-Probe: Estimating Traffic Correlations from\\ Sampling and Active Network Probing}
\author{\authorblockN{Amr Rizk, Zdravko Bozakov, and Markus Fidler}
\IEEEauthorblockA{Institute of Communications Technology, Leibniz Universit\"{a}t Hannover\\ \{amr.rizk, zdravko.bozakov, markus.fidler\}@ikt.uni-hannover.de}}
\maketitle
\begin{abstract}
An extensive body of research deals with estimating the correlation and the Hurst parameter of Internet traffic traces. The significance of these statistics is due to their fundamental impact on network performance. The coverage of Internet traffic traces is, however, limited since acquiring such traces is challenging with respect to, e.g., confidentiality, logging speed, and storage capacity. In this work, we investigate how the correlation of Internet traffic can be reliably estimated from random traffic samples. These samples are observed either by passive monitoring within the network, or otherwise by active packet probes at end systems. We analyze random sampling processes with different inter-sample distributions and show how to obtain asymptotically unbiased estimates from these samples. We quantify the inherent limitations that are due to limited observations and explore the influence of various parameters, such as sampling intensity, network utilization, or Hurst parameter on the estimation accuracy. We design an active probing method which enables simple and lightweight traffic sampling without support from the network. We verify our approach in a controlled network environment and present comprehensive Internet measurements. We find that the correlation exhibits properties such as long range dependence as well as periodicities and that it differs significantly across Internet paths and observation times.
\end{abstract}
\section{Introduction}
\label{sec:introduction}
Traffic characteristics play a key role in planning and operation of packet data networks. As a consequence, in recent years network measurements have attracted considerable attention as a practical method for inferring traffic properties. The scope of such measurements varies from access networks to backbone networks or even across the Internet.

Numerous comprehensive measurement studies, based on recorded network traces, have revealed that aggregate Internet traffic possesses long memory correlations, so-called long range dependence (LRD)\cite{leland94,paxon95,crovella97,feldmann:iptraffic}.
The impact of LRD on network performance was investigated in several works, e.g., \cite{norros95,duffield94,erramilli96,massoulie99,ribeiro:multiscalequeueing,rizk12:e2efbm,liebeherr12:heavytail,mandjes07}.
Networks fed with LRD traffic exhibit a fundamentally different behavior compared to systems fed with memoryless or Markovian traffic.

In practice continuous logging and evaluation of all relevant network events in large networks is typically not feasible due to efficiency, confidentiality, and cost factors. For example, with link speeds of $10$~Gbps and more capturing traffic traces becomes increasingly difficult, as suitably large and fast storage systems are expensive. One main challenge is therefore, to extract the desired information from a subset of events, e.g., using a sampling procedure that yields consistent estimates of the target metric. In addition, ISPs rarely disclose traffic traces because of confidentiality issues such that traffic characteristics can only be inferred from external observations. Further, a fundamental limitation of traffic traces is that these reflect traffic characteristics at only a \emph{single} observation point.

In this work, we investigate the problem of estimating the correlation of Internet traffic given a limited set of random samples. First, we consider passive sampling, i.e., capturing traffic samples at some directly accessible node, e.g., a router. Here, the main focus is on the choice of the sampling process and it's properties. Further, for any practical realization passive sampling yields a finite sample size, which directly influences the accuracy of the results. Secondly, we consider active probing that is a technique, where external measurements of specific probe packets are used. The aim is to avoid any particular network support by exploiting, e.g., timing information that is imprinted on the probes by interaction with network traffic. The additional challenge of active compared to passive methods is to design probes that actually permit inferring the desired traffic characteristics, which in certain cases may even be impossible~\cite{machiraju07}.

The ultimate result of this work is to enable the online estimation of traffic correlations along network paths without network support. To this end, we present methods for extracting LRD characteristics from sampled traffic. We derive the impact of sampling on the observed traffic correlations for different sampling strategies and show that sampling may distort observations. We develop methods that reverse these effects for a set of sampling processes. We quantify the accuracy of the observations under finite sampling durations, showing that the estimation error increases as $\tau^{2-2H}$ with the autocovariance lag $\tau$ and the LRD Hurst parameter $H \in (0.5,1)$. We derive the impact of different sampling parameters on estimation accuracy and show a non-linear trade-off between sampling intensity and sampling duration. Finally, we design and evaluate a practical active probing method to estimate traffic correlations from external observations. We present practical testbed and Internet measurement results showing a complex covariance structure of Internet traffic that exhibits LRD as well as periodic behavior.

The paper is structured as follows: In the next section we present the state-of-the-art on LRD network traffic characteristics, sampling and active network probing. In Sect.~\ref{sec:sampling} we derive our main results concerning traffic sampling and the accuracy of the estimated traffic parameters. In Sect.~\ref{sec:network_probing} we present and deploy an active probing method that uses packet probes to infer traffic correlations. Sect.~\ref{sec:conclusions} concludes the paper.
\section{Related Work}
\label{sec:relatedwork}
In the following, we discuss related work on LRD traffic characteristics, sampling and network probing.
\subsection{LRD traffic characteristics}
Comprehensive measurements in the 90s, e.g.,
\cite{leland94,paxon95,crovella97,feldmann:iptraffic} revealed that aggregate Internet traffic exhibits LRD and self-similarity phenomena, that can be described by the so-called Hurst parameter $H$. A self-similar stochastic process possesses the same finite dimensional distributions on different time scales except for a rescaling factor which depends on $H$. The aggregation of multiple traffic sources offers a possible explanation of these characteristics. It was shown in \cite{taqqu97} that aggregating many on-off sources with heavy tailed on and off periods yields self-similar LRD traffic. This notion corresponds to file transfers from heavy tailed file size distributions as observed on storage systems \cite{crovella97,willinger97}. An experimental validation of the relation between self-similarity and heavy-tailed distributions is carried out in \cite{loiseau:selfsimilarity10} on a large-scale experimental facility.

Given a stationary process $Y(t)$, LRD manifests itself in the slow decay of the autocovariance\footnote{Throughout this work, we use the definition of autocovariance in the signal processing sense, i.e., for a stationary process $Y(t)$ the autocovariance is defined as $c_Y(\tau):= \mathsf{E}[Y(t)Y(t+\tau)]-\mathsf{E}[Y(t)] \mathsf{E}[Y(t+\tau)]$. For brevity, we frequently use the term covariance to mean autocovariance.} $c_Y(\tau)$ such that
\begin{eqnarray}
\label{eq:autocovariance_fgn}
c_Y(\tau) \sim \sigma_Y^2 \tau^{2H-2} \quad \text{for } \tau \rightarrow \infty,
\end{eqnarray}
where $\sigma_Y^2$ is the variance of $Y(0)$ and the Hurst parameter $H \in (0.5,1)$. The sum of the autocovariance over all lags $\tau$ diverges, i.e., $\sum_{\tau} c_Y(\tau) \rightarrow \infty$.

In this work, we focus on the autocovariance structure of~\eqref{eq:autocovariance_fgn}. Our goal is to infer~\eqref{eq:autocovariance_fgn} from traffic observations, respectively, to estimate the Hurst parameter $H$ from from the slope of $c_Y(\tau)$ on a log-log scale. Numerous other methods exist for estimating the Hurst parameter from LRD and self-similar time series~\cite{beran94,taqqu:estimators,veitch99}.

In addition to \eqref{eq:autocovariance_fgn}, we consider two established methods for estimating the Hurst parameter $H$ \cite{beran94,taqqu:estimators}. First, we consider the aggregate variance method, that relies on the convergence rate of the sample mean of an LRD process to the true mean. Given samples of $Y(t)$ of size $M$, the variance of the sample mean decays as $\sim M^{2H-2}$ with growing $M$.

The second method denoted power spectral density method relies on the behavior of the spectral density of the LRD process $Y(t)$. The spectral density $\Psi_Y(f)$ of $Y(t)$ is well known \cite{beran94}. It can be approximated as
\begin{eqnarray}
\label{eq:psd_fgn}
\Psi_Y(f) \sim \left|f\right|^{1-2H} \quad \text{for } f \rightarrow 0.
\end{eqnarray}
The Hurst parameter $H$ can be estimated from the slope of $\Psi_Y(f)$ plotted against the frequency $f$ on a log-log scale.
\subsection{Sampling}
Sampling is widely used to reduce the data processing and storage requirements as well as to circumvent problems, such as system inaccessibility and hardware access latency. A fundamental result often employed in the sampling context is known as PASTA, Poisson Arrivals see Time Averages~\cite{wolff:pasta}. PASTA states that the portion of Poisson arrivals that see a system in a certain state corresponds, in average, to the portion of time the system spends in that state.
%An example from the networking context would be packets arriving at a router as a Poisson process that see on average the true mean queue length.

Further, the authors of \cite{Melamed:asta90} establish general conditions, such that Arrivals See Time Averages (ASTA) holds, i.e., bias free estimates are not limited to Poisson sampling. In a recent work the authors of \cite{baccelli:pasta} coined the term NIMASTA, i.e. Non-intrusive Mixing Arrivals See Time Averages, in the context of network measurements. Using an argument on joint ergodicity, the authors prove an almost sure convergence of
\begin{eqnarray}
\label{eq:nimasta}
\lim_{N\rightarrow \infty} \frac{1}{N} \sum_{i=1}^{N} g(Y(\theta_i)) = \mathsf{E}\left[g(Y(0))\right]
\end{eqnarray}
where $Y(\theta_i)$ is a sample of the process $Y(t)$ at time $\theta_i$ and $g$ is a general positive function of $Y$. The sampling times $\theta_i$ for $i \in \mathbb{N}$ are chosen according to a sampling process. The target metric is specified depending on the chosen function~$g$. Eq. \eqref{eq:nimasta} is satisfied when the process $Y(t)$ is ergodic and the sampling process is mixing~\cite{baccelli:pasta}. The authors in~\cite{baccelli:probing07} show that Poisson sampling, though bias free, does not guarantee minimum variance estimates.

A comparison of Poisson and periodic sampling was carried out in \cite{bintariq05,Roughan:comparison06}. In \cite{bintariq05} the authors show experimentally, that the differences between round trip times (RTT), loss rate and packet pair dispersion estimates, obtained by either Poisson or periodic probing, are in some cases not significant. Depending on the autocovariance of the sampled process, Poisson or periodic sampling can be superior. This is shown in \cite{Roughan:comparison06} using the metric asymptotic variance.

In \cite{viano:random_sampling} it is shown that for correlation lags tending to infinity, random sampling captures the long memory of the original processes, as long as the sampling distribution has a finite mean.
\subsection{Active network probing}
The injection of test packets into a network for inferring network performance, i.e., active probing, has attracted considerable attention in recent years. End-to-end packet delays or inter packet times are metrics commonly used to estimate network characteristics such as the average available bandwidth or even to reconstruct statistics of the cross-traffic \cite{jacobson:97,spruce,Jain:avbw:02,ribeiro:sampling06}. Under the assumption of FIFO scheduling, cross traffic intensity can be estimated from the dispersion of back-to-back probing packets \cite{dovrolis:dispersion01,spruce,liu05,liu07}.

Cross traffic estimation of LRD traffic using active measurements was discussed, in \cite{ribeiro:cross_traffic_est00,he03}. The authors of \cite{he03} carry out a numerical simulation to interpolate cross traffic from probes and predict future traffic from the LRD property. In \cite{ribeiro:cross_traffic_est00} the authors derive and show simulation results for a deterministic probing scheme based on a multi-fractal wavelet traffic model. Essential to their estimation is the assumption that the queue does not empty between the the individual packets of a packet probe. Our work differs from \cite{ribeiro:cross_traffic_est00,he03} as we examine different \emph{random} sampling distributions and show how to extract traffic correlations from distorted observations.

Two important aspects concerning network probing are the measurement intrusiveness and the interaction of probes with the measured system. The first aspect is usually addressed by minimizing the probing rate while controlling the quality of the results. The second aspect is more involved, since the probes perturb the system leading to distorted observations. For example, measuring queueing delays of probes to determine the true queue length distribution is governed by a type of Heisenberg uncertainty \cite{roughan05}, since the probes alter the queue length. The authors describe the impact of the probing intensity on the accuracy of the result using the notion of asymptotic variance. The effect is increased in case of LRD traffic, although not given in closed form, leading to higher uncertainty in the estimated waiting time \cite{roughan05}.
\section{Traffic sampling and parameter estimation}
\label{sec:sampling}
In this section we derive our main results on traffic covariance estimation from sampled observations. Based on sampling properties we present rigorous traffic parameter estimation. Subsequently, we investigate the accuracy of the estimates under the practical constraint of finite sample sizes.
\subsection{Covariance of sampled processes}
\label{subsec:sampling}
We define a sampling model comprising of three stationary discrete time processes: a traffic increment process $Y(t)$, a sampling process $A(t)$, and an observed process $W(t)$ for $t\!~\in~\!\mathbb{N}_0$. We assume statistical independence of $A(t)$ and $Y(t)$. Our focus lies on the estimation of the covariance of $Y(t)$ that is characterized by LRD. While the LRD process may be in continuous time, we regard its increments on a fixed time slot basis, and hence the discretization of $Y(t)$.
\begin{table*}
\begin{center}
\caption{Parametrization of sampling distributions and traffic parameter estimation}
\label{tab:sampling_distributions}
{\small
\begin{tabular}{|c|c c c c c|} \hline
\multirow{2}{*}{} & inter-sample distribution  & autocovariance $c_A(\tau)=$  & reconstructed traffic  & remarks & \\
 & $f(\tau)$ & $\mathsf{E}\left[A(t) A(t+\tau)\right]-\mu_A^2$ for $\tau>0$ & autocovariance $c_Y(\tau)$ & & \\ \hline \hline
Geometric & $p(1-p)^{\tau-1}$ & $0$  & $\frac{c_W(\tau)}{\mu_A^2}$ & $\mu_A = p$ & \\
Periodic & $\delta(\tau - \Delta)$ & $
\begin{array}{l l}
  1/\Delta - 1/\Delta^2 &  \mbox{for $\tau = k\Delta$, $ k \in \mathbb{N}$}\\
  - 1/\Delta^2 &  \mbox{otherwise}\\ \end{array} $ & $\frac{c_W(\tau)-\mu_A\mu_Y^2\left(1-\mu_A\right)}{\mu_A}$ & $\begin{array}{c c}
  c_Y(\tau)   \mbox{ at $\tau = k\Delta$,}\\
  \mu_A = 1/\Delta  \\ \end{array}$ & \\[3pt]
Gamma & $\frac{\beta^{\alpha}}{\Gamma(\alpha)} \tau^{\alpha-1} e^{-\beta \tau}$ & $-\mu_A^2e^{-4\mu_{\!A}\tau}$ & $\frac{c_W(\tau)+\mu_A^2 \mu_Y^2 e^{-4\mu_{\!A}\tau}}{\mu_A^2 \left(1-e^{-4\mu_{\!A} \tau}\right)}$ & for $\alpha=2, \mu_A = \frac{\beta}{\alpha}$ & \\[3pt]
Uniform & $1/b$ for $0 \le \tau \le b$ & $\mu_A^2\left(\frac{1}{2} e^{\frac{1}{2}\mu_{\!A}\tau}-1\right)$  & $\frac{c_W(\tau) - \mu^2_A \mu^2_Y \left(\frac{1}{2}e^{\frac{1}{2}\mu_{\!A} \tau} - 1\right)}{\mu^2_A \frac{1}{2} e^{\frac{1}{2}\mu_{\!A} \tau}}$ & $\begin{array}{c c}  c_Y(\tau)   \mbox{ for $\tau \le b$,}\\
  \mu_A = 2/b  \\ \end{array}$ & \\ \hline
\end{tabular}
}
\end{center}
\end{table*}
%$c_Y(\tau)$ at $\tau = k\Delta$
%$c_Y(\tau)$ for $\tau \le b$

The sampling process $A(t)$ is a point process taking the value of one whenever a sample is taken, and zero otherwise, i.e., $A(t)$ is a Kronecker delta train, where a Kronecker delta is defined as $\delta(n)=1$ for $n=0$ and zero otherwise. The process has independent and identically distributed (iid) inter-sample times drawn from a given probability distribution~$F$. The inter-sample time is the time between two consecutive Kronecker deltas. The sampling intensity, i.e., the mean rate of the sampling process of $A(t)$, is $\mathsf{E}\left[A(t)\right]=\mu_{A}$ for all $t$, with $0 \le \mu_A \le 1$. Throughout this work we use $\mu_{(\cdot)}$ to denote the expected value $\mathsf{E}\left[(\cdot)\right]$.

We base our analysis on the observed stochastic process $W(t)$, generated by random samples $A(t)$ of the increment process $Y(t)$, with
\begin{eqnarray}
\label{eq:process_sampling}
W(t) = A(t) Y(t).
\end{eqnarray}
\begin{figure}[t]
		\includegraphics[width=1.00\columnwidth]{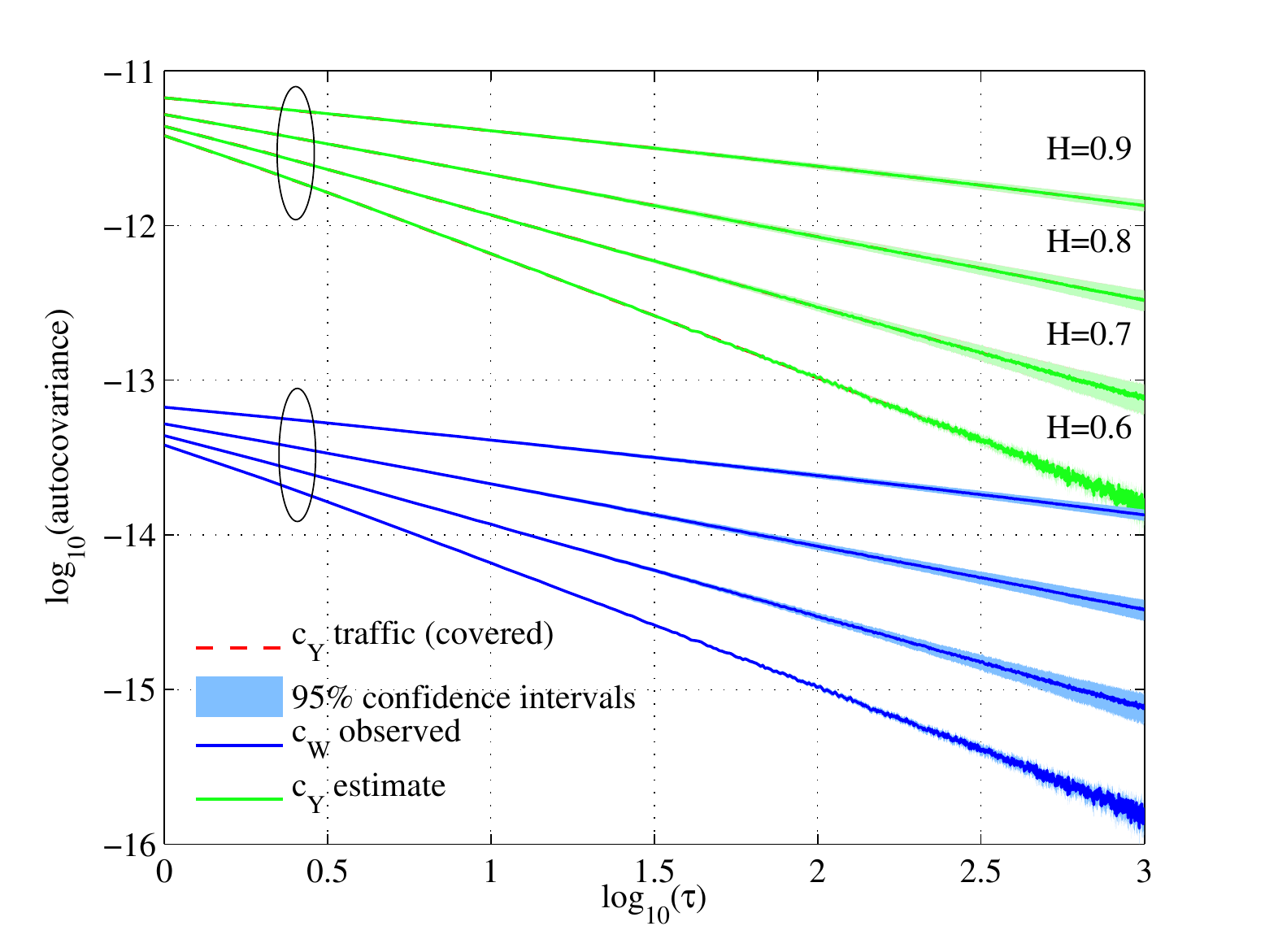}
		\caption{Autocovariance of LRD traffic processes under geometric sampling. The observed ``$c_W(\tau)$" maintains the autocovariance structure of the traffic process. The covariance of the  original process ``$c_Y(\tau)$ (traffic)" is exactly covered by the reconstructed ``$c_Y(\tau)$ (estimate)". }
	\label{fig:xcov_poisson_sim}
\end{figure}
\begin{figure*}[b]
        \centering
            \subfigure[Periodic sampling]{
            \label{fig:cyhat_periodic}
            \includegraphics[scale=0.373]{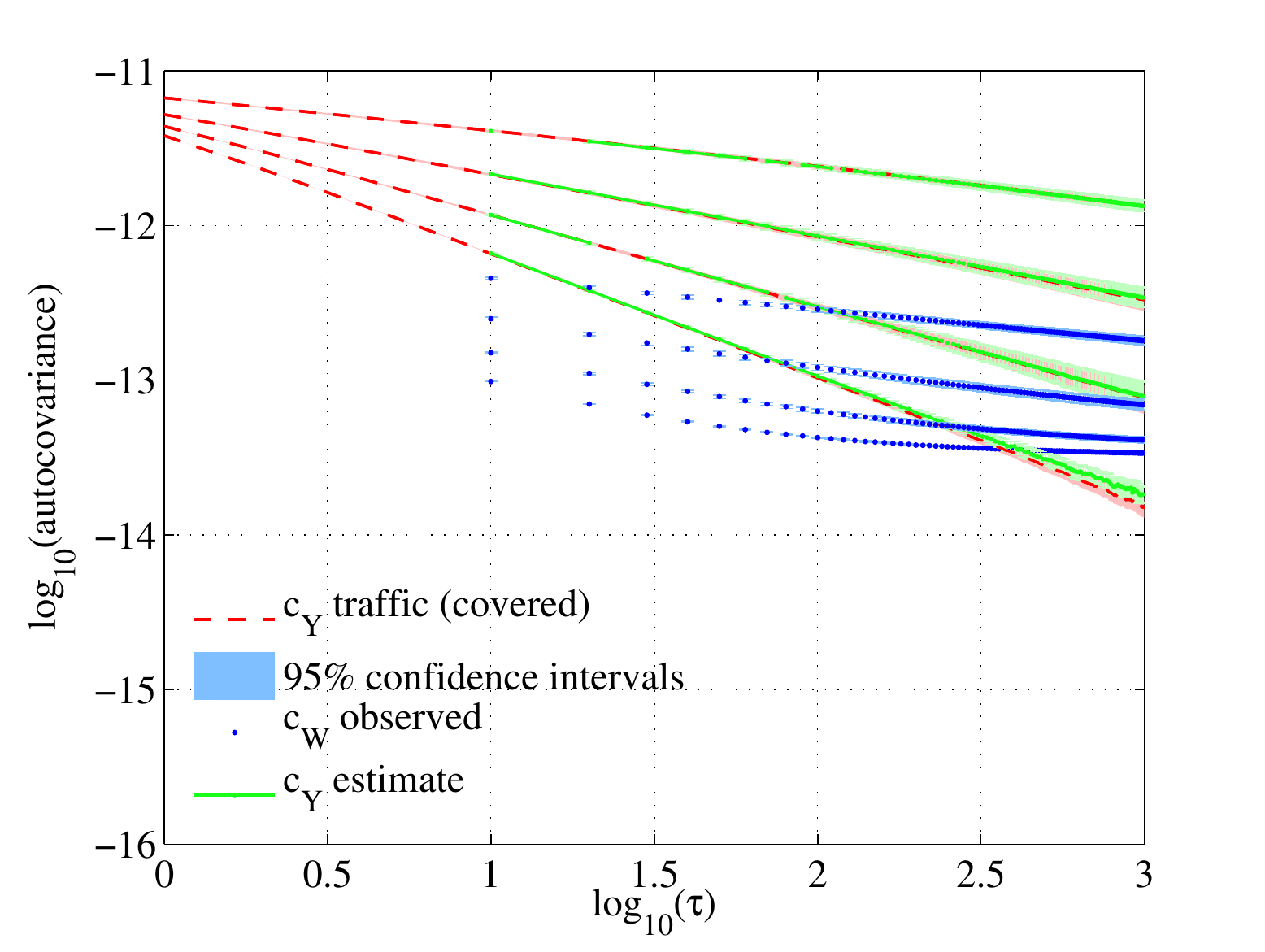}
            }
            \hspace{-20pt}
            \subfigure[Gamma sampling]{
            \label{fig:cyhat_gamma}
            \includegraphics[type=pdf,ext=.pdf,read=.pdf,scale=0.373]{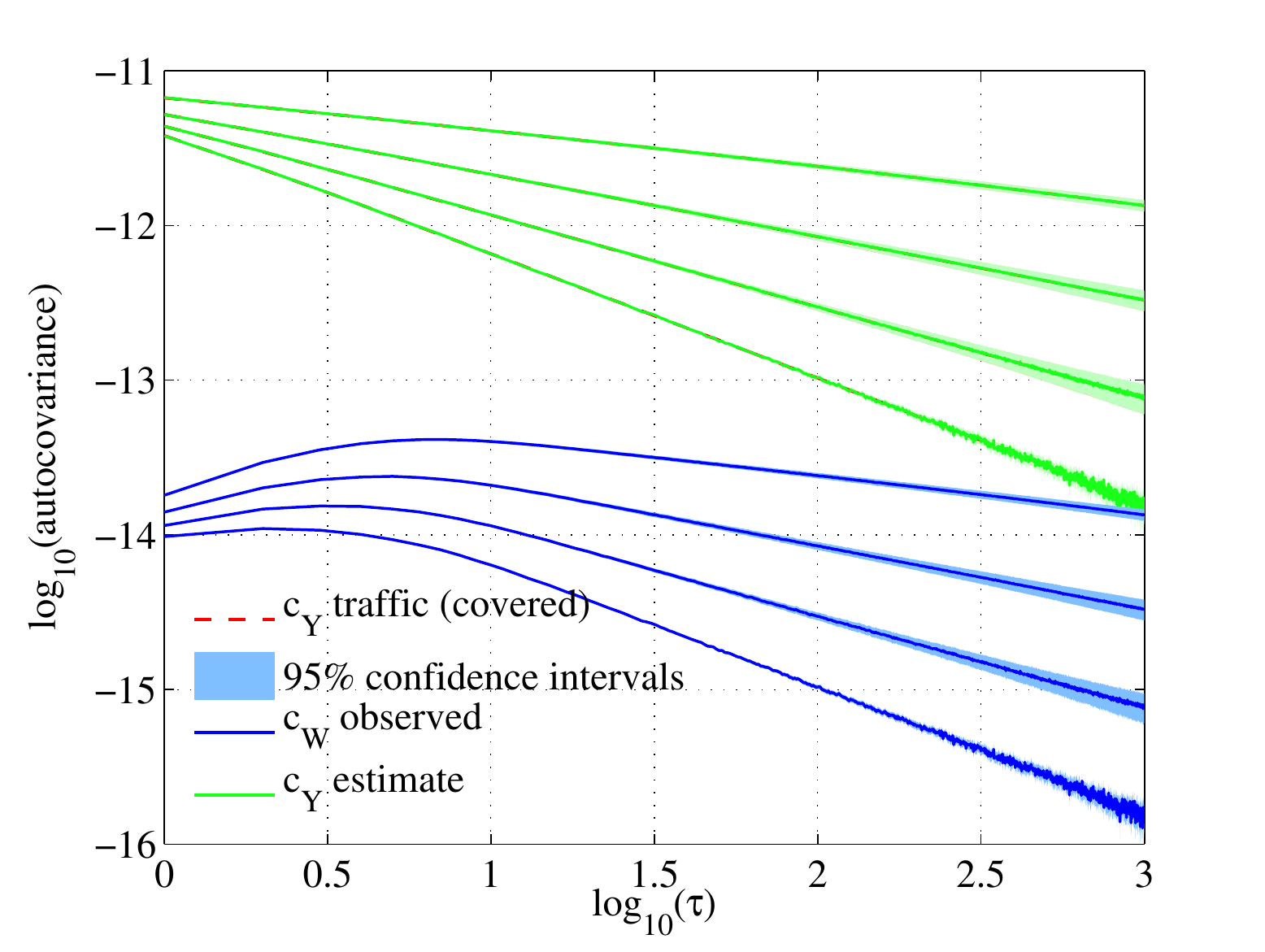}
            }
            \hspace{-20pt}
            \subfigure[Uniform sampling]{
            \label{fig:cyhat_uniform}
            \includegraphics[scale=0.373]{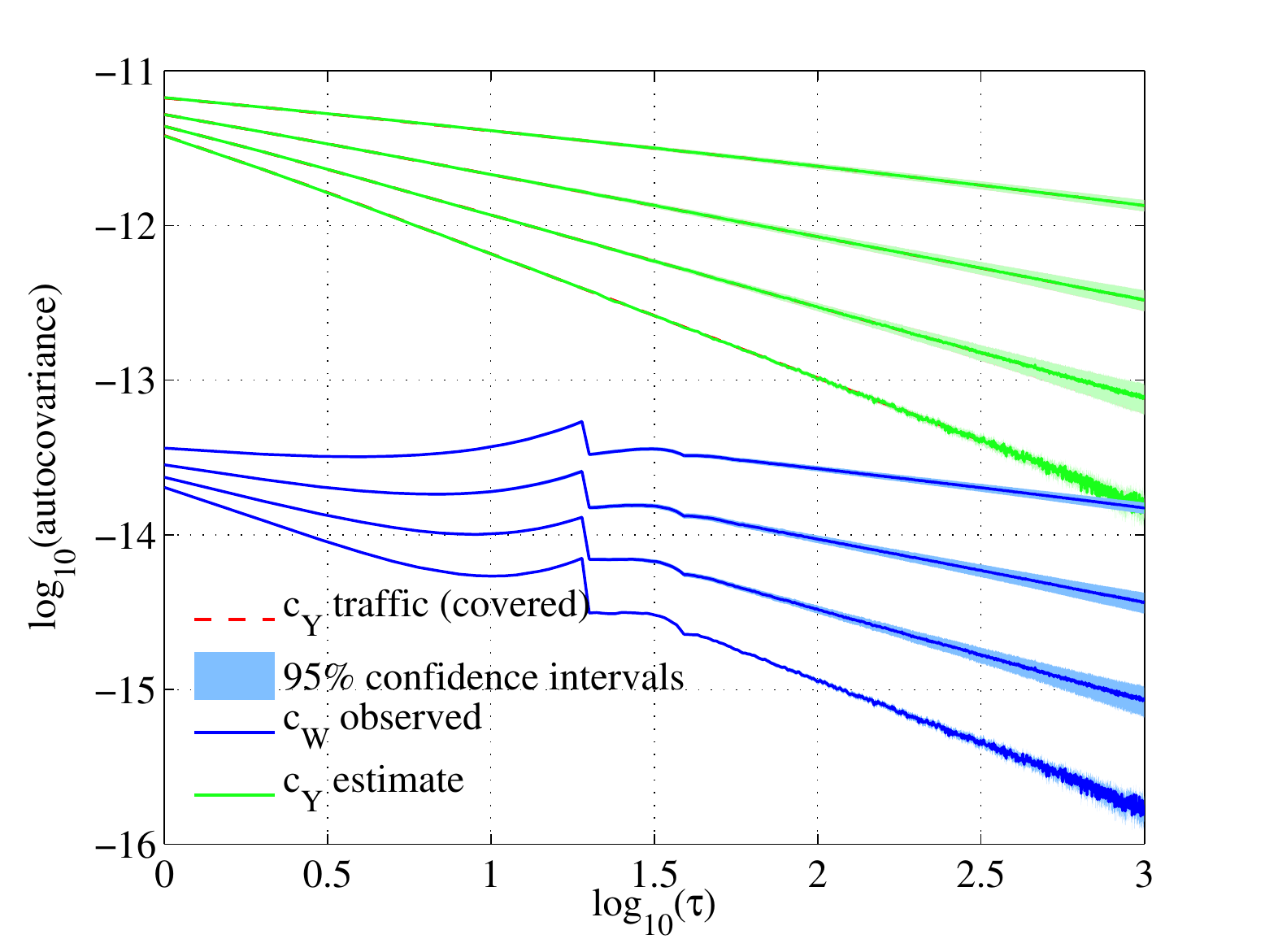}
            }
            \caption{Autocovariance of the LRD process under different sampling strategies. Note that ``$c_Y(\tau)$ (traffic)" is covered by the ``$c_Y(\tau)$ (estimate)".}
            \label{fig:cyhat_diff_probing}
\end{figure*}

We aim to infer properties of the traffic process $Y(t)$ from the observation process $W(t)$. In particular, we are interested in sampling distributions $F$ that deliver accurate estimates of the correlations of the LRD traffic process $Y(t)$ and the associated Hurst parameter $H$. Extracting the autocovariance of the process $Y(t)$, i.e., $c_Y(\tau)$ from the observed $c_W(\tau)$ is generally not a straightforward task. The following lemma reveals the impact of sampling on the autocovariance of the observed process. The proof of Lem.~\ref{lem:sample_covariance} is a variation of standard technique in stochastics.
\begin{lemma}
\label{lem:sample_covariance}
Given the stationary and independent stochastic processes $A(t)$ and $Y(t)$ and let $W(t)=A(t)Y(t)$. The covariance of $W(t)$ can be decomposed into
\begin{eqnarray*}
c_W(\tau)=  \left(c_A(\tau)+\mu_{A}^2\right)c_Y(\tau) + c_A(\tau) \mu_{Y}^2.
\end{eqnarray*}
\end{lemma}
\begin{proof}
Given independent and stationary processes $A(t)$ and $Y(t)$. Let $W(t) = A(t) Y(t)$. It follows that
\begin{eqnarray*}
c_W(\tau) & = & \mathsf{E}\left[A(t) A(t+\tau)\right] \mathsf{E}\left[Y(t) Y(t+\tau)\right] - \mu_{A}^2 \mu^2_{Y}\nonumber\\
& = & \left(c_A(\tau)+\mu_{A}^2\right) \left(c_Y(\tau)+\mu_{Y}^2\right) - \mu_{A}^2 \mu^2_{Y}\nonumber\\
& = & \left(c_A(\tau)+\mu_{A}^2\right)c_Y(\tau) + c_A(\tau) \mu_{Y}^2\nonumber
\end{eqnarray*}
where $c_{(\cdot)}(\tau)$ denotes the covariance of process $(\cdot)$ at lag $\tau$.
\end{proof}

Lem. \ref{lem:sample_covariance} clearly shows the impact of the sampling process on the observed covariance. In particular, the choice of the inter-sample distribution influences $c_W(\tau)$ through $\mu_A$ and $c_A(\tau)$, i.e., both the sampling intensity and the sampling covariance influence the observation.

In this work, we investigate four inter-sample distributions: geometric (memoryless), periodic, Gamma, and uniform. For each distribution we show how to recover the covariance of the LRD process $c_Y(\tau)$ from the observed $c_W(\tau)$ using the covariance $c_A(\tau)$. To this end, we derive the covariance of the sampling process $c_A(\tau) = \mathsf{E}\left[A(t)A(t+\tau)\right] - \mu_A^2$. We use the probability mass function $f(\tau)$ of the inter sample times to calculate the $n$-fold self-convolution $f^{(*n)}(\tau)$. We then calculate the autocorrelation $\mathsf{E}\left[A(t)A(t+\tau)\right] = \mu_A \sum_{n=1}^{\infty}f^{(*n)}(\tau)$ as given in \cite{cox66}, Eq.~(4.6.1). We exploit the property that $f^{(*n)}$ is a power series for the considered distributions and that its sum converges. The derivation of the autocorrelations used in Tab.~\ref{tab:sampling_distributions} is given in appendix~\ref{subsec:Autocorrelation_sampling_processes} to \ref{subsec:Uniform_sampling}. In the last step we insert $c_A(\tau)$ into Lem.~\ref{lem:sample_covariance} and solve for $c_Y(\tau)$.

Tab.~\ref{tab:sampling_distributions} summarizes the expressions used to reconstruct $c_Y(\tau)$ given specific inter-sample distribution parameters and corresponding $c_A(\tau)$. First, we consider the geometric inter-sample distribution, i.e., a Bernoulli sampling process. The independence of the increments implies that $c_A(\tau) = 0$ for $\tau > 0$. From Lem.~\ref{lem:sample_covariance}, the observations $W(t)$ have autocovariance
\begin{eqnarray}
\label{eq:geometric_reconst}
c_W(\tau) = c_Y(\tau) \mu_A^2.
\end{eqnarray}
This shows that sampling processes with uncorrelated increments only shift the autocovariance structure of the sampled process $Y(t)$ by $\mu_A^2$.

Next, we consider periodic sampling, where $A(t)$ is modeled as a comb of Kronecker deltas with sampling period $\Delta$. The mean intensity of the sampling process is $\mu_A = 1/\Delta$. We can recover $c_Y(\tau)$ using Lem.~\ref{lem:sample_covariance}, however, only at $\tau=k\Delta$ where $k \in \mathbb{N}$. To perform this inference, the mean rate $\mu_Y$ of the traffic process $Y(t)$ must be known. Due to the rigid structure of periodic sampling it is, however, shown that the mean rate estimator $\mu_W / \mu_A$ is not unbiased \cite{baccelli:pasta}, e.g., the sampling period may coincide with periodicities in the original process.

Finally, Tab.~\ref{tab:sampling_distributions} provides expressions for reconstructing $c_Y(\tau)$ after Gamma and uniform sampling. For mathematical tractability, here we use continuous time for the derivation of the autocorrelation of $A(t)$. For discretization we use a time slot of unit size. Note that the discretization error diminishes for autocorrelation lags much larger than the discretization time slot. In case of Gamma sampling, the ability to estimate $c_Y(\tau)$ is not limited to the exemplary $\alpha=2$ given in Tab.~\ref{tab:sampling_distributions}. Lem. \ref{lem:sample_covariance} can be used to estimate $c_Y(\tau)$ for Gamma sampling processes with arbitrary parameters as long as the autocovariance $c_A(\tau)$ is computable. We provide results for Gamma sampling with $\alpha=4$ in
appendix~\ref{subsec:Gamma_sampling}. For uniform sampling with support $b$ Tab.~\ref{tab:sampling_distributions} gives a result for lags $\tau \le b$. Due to the finite support, $c_A(\tau)$ quickly approaches zero for $\tau > b$. Like in case of periodic sampling, the reconstruction of $c_Y(\tau)$ from Gamma and uniform sampling, respectively, requires knowledge of $\mu_Y$.

Figures~\ref{fig:xcov_poisson_sim} and \ref{fig:cyhat_diff_probing} illustrate autocovariance estimates derived from observations $W(t)$, that are obtained by sampling LRD traffic with autocovariance $c_Y(\tau) \sim \sigma_Y^2 \tau^{2H-2}$ and $H \in [0.6,0.9]$.\footnote{Synthetic traces of length $2.5\times10^8$ time slots were used for the   simulation which was repeated $25$ times for each considered $H$.}  We use geometric, periodic, Gamma, and uniform inter-sample time distributions and set $\mu_A\!=\!0.1$. In all cases the reconstructed autocovariance denoted ``$c_Y$(estimate)" exactly covers the original traffic autocovariance ``$c_Y$(traffic)".

Geometric sampling in Fig.~\ref{fig:xcov_poisson_sim} preserves the linear covariance structure of $c_Y(\tau)$. The observed $c_W(\tau)$ is vertically shifted by $\log(\mu_A^2)$ w.r.t. the original $c_Y(\tau)$. The Hurst parameter $H$ can be inferred directly from the slope of $c_W(\tau)$.

For the remaining distributions shown in Fig.~\ref{fig:cyhat_diff_probing}, the observations $c_W(\tau)$ are distorted. However, using Lem.~\ref{lem:sample_covariance} we recover the original covariance $c_Y(\tau)$. Using the expressions from Tab.~\ref{tab:sampling_distributions} we reconstruct ``$c_Y$(estimate)" which lies on top of the original autocovariance ``$c_Y$(traffic)".

In the following we discuss advantages and disadvantages of the presented sampling distributions. Periodic and uniform sampling are practically convenient as the inter-sample times cannot become arbitrarily large due to the finite support of the inter-sample distribution. Moreover, periodic sampling is easy to implement.

However, it is important to point out that periodic sampling yields misleading results if the sampling period coincides with periodicities in the target process. In addition, periodic, Gamma as well as uniform sampling require a reconstruction step to estimate the covariance $c_Y(\tau)$ from observations as shown above. To this end, an estimate of $\mu_Y$ is required.

Memoryless sampling is proposed by the IETF as a network probing scheme \cite{RFC2330}. We find that a major advantage of geometric sampling, i.e., memoryless, is that the covariance structure of $c_Y(\tau)$ is preserved in the observations as given in \eqref{eq:geometric_reconst}. This stands in contrast to periodic, Gamma and uniform sampling. In the following we continue the analysis with geometric sampling because of its advantages discussed above.
\subsection{Impact of finite sample sizes}
\label{subsec:impact_finite_sample}
Next, we examine the accuracy of the derived estimates for finite sample sizes as this is important for any practical realization. We determine the impact of sampling parameters, e.g., sampling duration or intensity, on the observations. Moreover, we evaluate the accuracy of the deployed statistical estimators. Finally, we recover the results from Sect.~\ref{subsec:sampling} in the limit for infinite sampling durations.

We investigate sample autocovariances marked by $\tilde{c}_{(\cdot)}$ as estimators of the population autocovariances $c_{(\cdot)}$. In addition, we consider the sample means $\tilde{\mu}_{(\cdot)}$ as estimators of the population means $\mu_{(\cdot)}$. To better understand the impact of finite sample sizes on the observations and the covariance estimates we examine the individual effects of the sample covariances involved in a step by step manner.

While geometric sampling is appealing since it's autocovariance $c_A(\tau) = 0$ for $\tau > 0$, it looses this property for finite sampling duration $T$, where $T$ is the length of the time-slotted sampling process $A(t)$ in slots.

In the following we focus on three aspects. First in subsection \ref{sect:Observation limit}, we derive the impact of finite sample sizes on the observability of the covariance of sampled traffic. The second aspect is the impact of the sample covariance $\tilde{c}_{A}(\tau)$ and its influence on the estimation error. This is handled in subsection \ref{sect:Estimation accuracy}. The third aspect is the impact of finite sample sizes on the bias of the covariance estimators given in subsection \ref{sect:bias_of_covariance_estimators}.
\subsubsection{\textbf{Observation limit}}
\label{sect:Observation limit}
In this section, we consider observations from finite sampling. At first, we do not consider deviations of sample statistics from respective population measures. This assumption is relaxed in the following subsections. We investigate the limit up to which the covariance of the observed LRD process $\tilde{c}_W(\tau)$ can be distinguished from the covariance of iid sequences of the same sample size. Obeying this limit ensures that the variability that is due to the sample size does not mask the covariance that we seek to observe. Exemplary, we depict in Fig. \ref{fig:sim_error_welt_1} the sample autocovariance $\tilde{c}_Y(\tau)$ of an LRD traffic trace $Y(t)$, and the corresponding autocovariance of geometrically sampled observations $\tilde{c}_W(\tau)$, both with a limited sample size $T$. Evidently, $\tilde{c}_W(\tau)$ is not just a shifted version of $\tilde{c}_Y(\tau)$ but distorted for increasing lags $\tau$ by observation ``noise"  that stems from the variability of the limited sample size.

We seek a range of lags $\tau \in [0,\tau^*]$ in which the covariance of the sampled process can be observed without significant distortion. Based on a standard technique \cite{beran94} we compare the covariance of the observed process to the covariance of geometrically sampled iid Gaussian sequences to obtain $\tau^*$ up to which both covariances are significantly different.
\begin{figure}[!t]
            %\center
            \hspace{-9pt}
            \subfigure[Sampling]{
            \label{fig:sim_error_welt_1}
            \includegraphics[width=0.522\columnwidth]{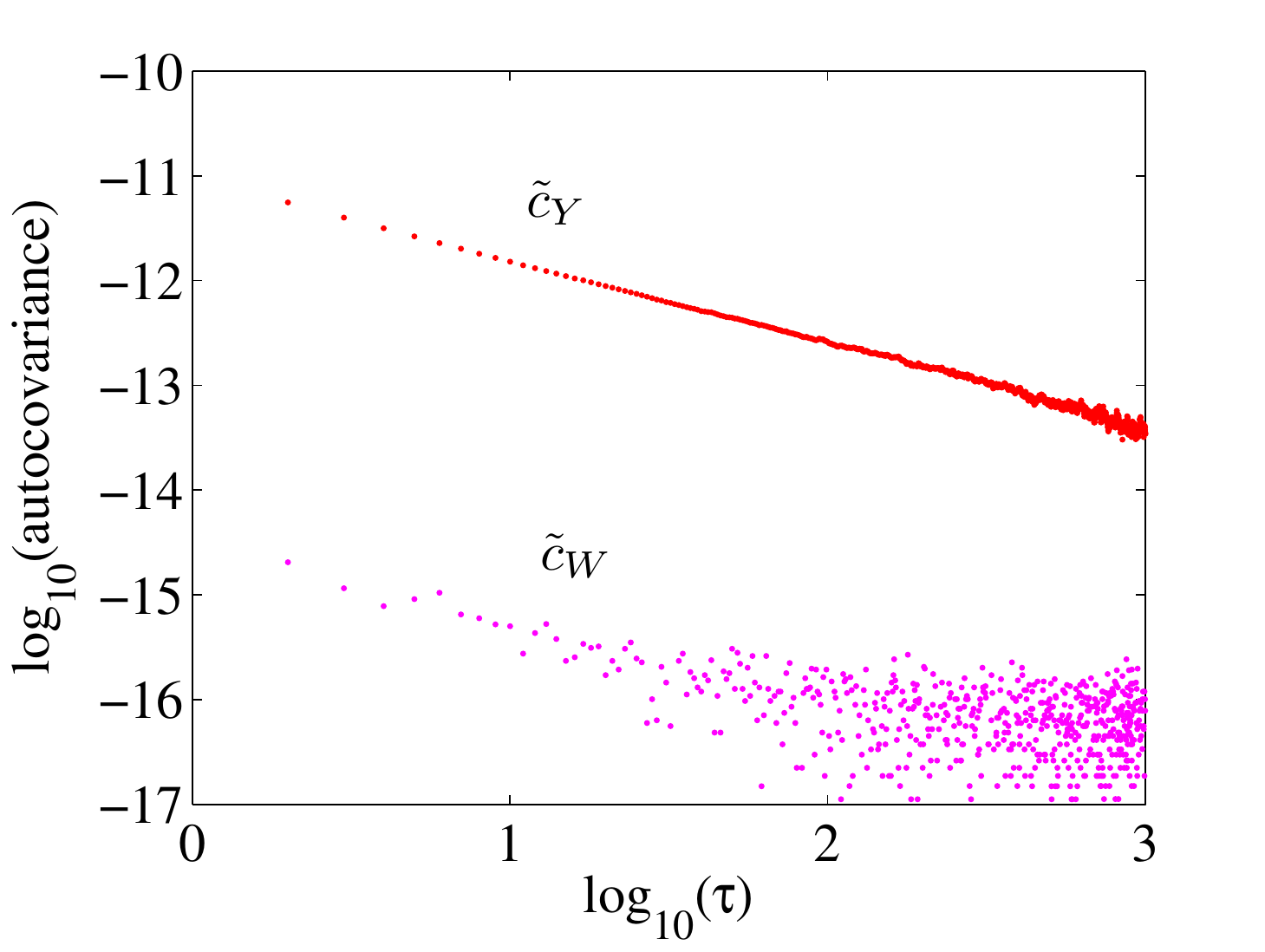}
            }
            \hspace{-23pt}
            \subfigure[Schematic description]{
            \label{fig:sim_error_welt_3}
            \includegraphics[width=0.522\columnwidth]{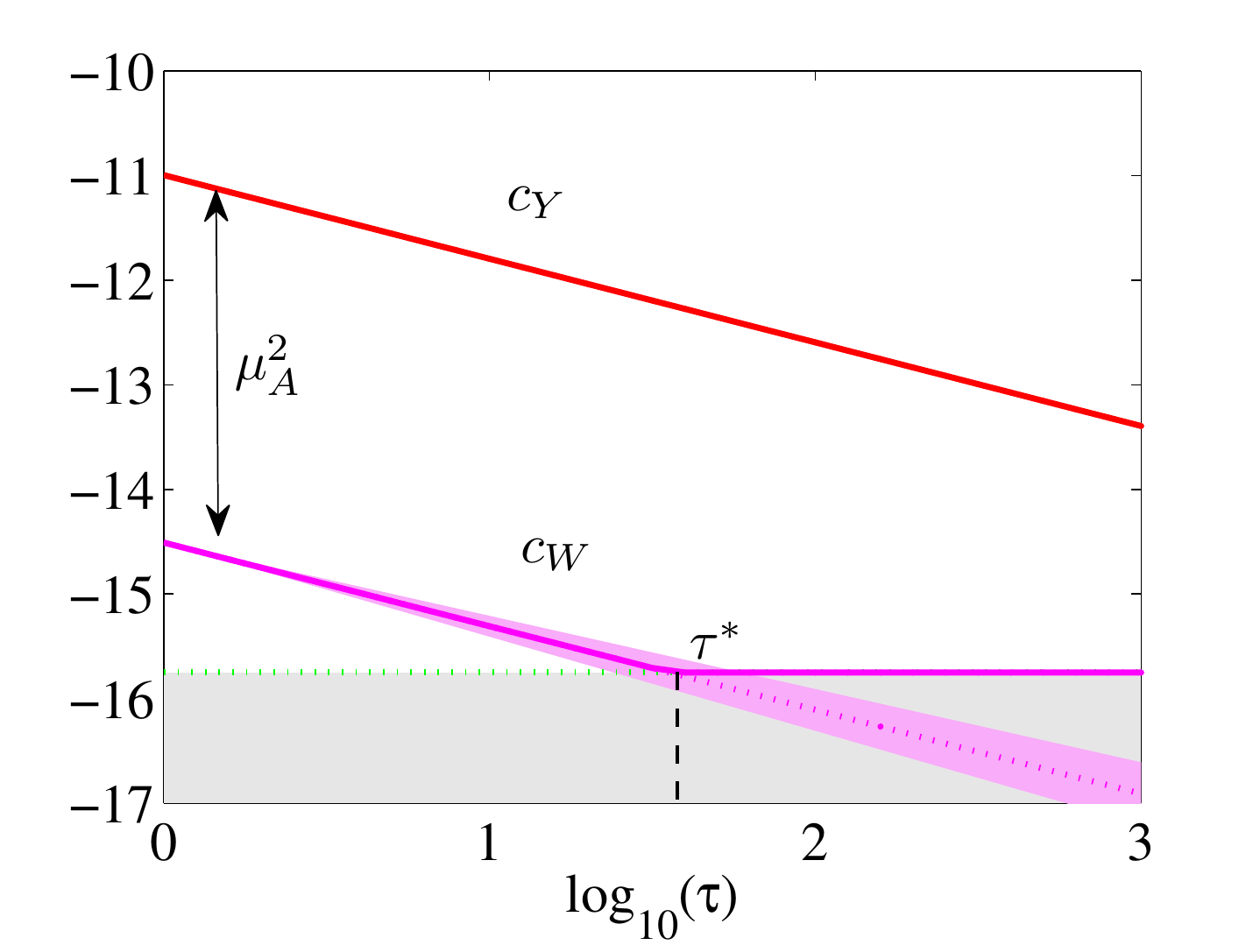}
            }
            \caption{Noisy observations due to finite sampling. Noise floor (shaded area) in Sect. \ref{sect:Observation limit}. Noise cone in Sect. \ref{sect:Estimation accuracy}. }
            \label{fig:observ_noise_poisson}
\end{figure}

We define $\tau^*$ as the intersection of $c_W(\tau)$ from~\eqref{eq:geometric_reconst} and the 0.95 confidence interval for geometrically sampled finite Gaussian iid sequences with mean $\mu_Y$ and variance $\sigma_Y^2$. For $T\gg\tau$ we find that this confidence interval is given by $2\sqrt{(\sigma_A^2\mu_Y^2+\mu_A\sigma_Y^2)^2 + 4\mu_A^2\mu_Y^2(\sigma_A^2\mu_Y^2+\mu_A\sigma_Y^2)}/\sqrt{T}$. The calculation relies on the central limit theorem and is given in detail in appendix~\ref{subsec:distr_iid_gaussian_sampled_CI}. Fig. \ref{fig:sim_error_welt_3} depicts $\tau^*$ as well as the confidence interval, denoted as noise floor, schematically.

To calculate $\tau^*$ for LRD traffic with covariance $c_Y(\tau)\!=\!K \sigma_Y^2 \tau^{2H-2}$, with constant $K$, we equate the above confidence interval with $c_W(\tau) = c_Y(\tau) \mu_A^2$ from~\eqref{eq:geometric_reconst} to obtain
\begin{eqnarray*}
  \tau^* \!\!=\! \left[\!\frac{K \sigma_Y^2 \mu_A^2 \sqrt{T}}{2 \sqrt{\! \left(\sigma_A^2 \mu_Y^2 + \mu_A \sigma_Y^2 \right)^2 + 4\mu_A^2\mu_Y^2 \left(\sigma_A^2 \mu_Y^2 + \mu_A \sigma_Y^2 \right)}} \!\right]^{\!\frac{1}{2-2H}}\!\!\!.
%\label{eq:cy_error_absolute}
\end{eqnarray*}

It is obvious that stronger LRD, i.e., higher $H$, is observed better. Clearly, for an infinite sample size $T \rightarrow \infty$, the observable range goes to infinity $\tau^* \rightarrow \infty$. Fig.~\ref{fig:sim_error_welt_1} shows that in practice it is important to consider this range to ensure that the results are not strongly distorted.
\subsubsection{\textbf{Estimation accuracy}}
\label{sect:Estimation accuracy}

Next, we evaluate the impact of the finite sample size on the sample covariance $\tilde{c}_{A}(\tau)$. We analyze the influence of $\tilde{c}_{A}(\tau)$ on the observation $\tilde{c}_{W}(\tau)$ and of estimates of $c_{Y}(\tau)$ obtained thereof. For ease of exposition, we assume $\tilde{c}_{Y}(\tau)=c_{Y}(\tau)$, $\tilde{\mu}_{A}=\mu_{A}$ and $\tilde{\mu}_{Y}=\mu_{Y}$, i.e., in this subsection we restrict our analysis to the deviation of $\tilde{c}_{A}(\tau)$ from $c_{A}(\tau)$.

We assume $T \gg \tau$ and use the central limit theorem to approximate the distribution of the sample autocovariance $\tilde{c}_A(\tau)$ by a Gaussian distribution with standard deviation $\sigma_A\sqrt{\sigma_A^2+4\mu_A^2}/\sqrt{T - \tau}$. From the Bernoulli sampling process $A(t)$ we know that $\sigma_A^2 = \mu_A -\mu_A^2$. We calculate the $0.95$ confidence interval $ c^{.95}_A \approx \pm 2 \sigma_A\sqrt{\sigma_A^2+4\mu_A^2}/\sqrt{T - \tau}$ for the mean sample autocovariance\footnote{We use the relation $\approx$ to denote the approximation, here due to the central limit theorem.}. The derivation can be found in appendix~\ref{subsec:distr_poisson_sampling_CI}.

With help of $ c^{.95}_A$ we investigate the impact of the variations of $\tilde{c}_A(\tau)$ on the observation $\tilde{c}_W(\tau)$. First, we use $c^{.95}_A$ to calculate a confidence interval for $\tilde{c}_W(\tau)$ from Lem.~\ref{lem:sample_covariance} as $c^{.95}_W(\tau) = \pm c^{.95}_A \left(c_Y(\tau) + \mu_Y^2\right)$. We schematically depict $c^{.95}_W(\tau)$ as noise cone in Fig.~\ref{fig:sim_error_welt_3}.

Next, in reference to \eqref{eq:geometric_reconst} we consider the estimator $\tilde{c}_W(\tau)/ \mu_A^2$ for estimating $c_Y(\tau)$. We analyze the impact of the variations of $\tilde{c}_A(\tau)$ on this estimator. We calculate the confidence interval $c^{.95}_Y(\tau)$ for this estimator as $c^{.95}_Y(\tau) = \pm c^{.95}_A \left(c_Y(\tau) + \mu_Y^2\right) / \mu_A^2$. Finally, we obtain the following relative error
\begin{eqnarray}
\label{eq:cy_error_relative}
\varepsilon_Y^{rel}(\tau) = \frac{|c^{.95}_Y(\tau)|}{c_Y(\tau)} \approx \frac{2 \sigma_A\sqrt{\sigma_A^2+4\mu_A^2}}{\sqrt{T - \tau} \mu_A^2} \left(1+\frac{\mu_Y^2}{c_Y(\tau)}\right).
\end{eqnarray}

From \eqref{eq:cy_error_relative} we observe that the estimation error introduced through $\tilde{c}_A(\tau)$ decays with increasing sampling duration $T$ or with increasing sampling intensity $\mu_A$. For small (practical) sampling intensities, e.g., $\mu_A \le 0.1$, we find a nonlinear tradeoff between sample intensity $\mu_A$ and sampling duration~$T$. Using $\sigma_A^2 = \mu_A-\mu_A^2$ from the Bernoulli sampling process the prefactor in \eqref{eq:cy_error_relative} can be approximated as $1/\sqrt{T}\mu_A$ for $T \gg \tau$.  This result enables the important conclusion that for finite sample sizes sampling intensity has a stronger impact on accuracy than sampling duration.

Next, we examine the influence of the parameter $H$ on (\ref{eq:cy_error_relative}) for large lags $\tau$. For increasing $\tau$, $c_Y(\tau)$ decreases, such that when $c_Y(\tau) \ll \mu_Y^2$, the relative estimation error~\eqref{eq:cy_error_relative} becomes
\begin{eqnarray}
\label{eq:cy_error_relative_large_lag}
\varepsilon_Y^{rel}(\tau) & \approx & \frac{2 \sigma_A\sqrt{\sigma_A^2+4\mu_A^2}}{\sqrt{T - \tau} \mu_A^2} \left(\frac{\mu_Y^2}{c_Y(\tau)}\right) \nonumber\\
& \approx & \frac{2 \sigma_A\sqrt{\sigma_A^2+4\mu_A^2}\mu_Y^2}{\sqrt{T - \tau} \mu_A^2 \sigma_Y^2} \tau^{2-2H},
\end{eqnarray}
where we substituted $c_Y(\tau) = \sigma_Y^2 \tau^{2H-2}$. The relative estimation error $\varepsilon_Y^{rel}(\tau)$ increases with the lag $\tau$ depending on $H \in (0.5,1)$. For LRD traffic which exhibits large $H$, the estimation error increases slower in $\tau$ compared to traffic with a small parameter $H$.

We depict $\varepsilon_Y^{rel}(\tau)$ in Fig.~\ref{fig:ci_simulation_poisson}. To this end, we used $100$ generated LRD traffic traces with $T = 2 \times 10^8$ time slots. The figure includes auxiliary lines with a slope of $2-2H$. It is evident, that the estimation error evolves with $\tau$ as given by (\ref{eq:cy_error_relative_large_lag}).

\begin{figure}[t]
        \centering
            \hspace{-9pt}
            \subfigure[$H = 0.6$]{
            \includegraphics[width=0.5\columnwidth]{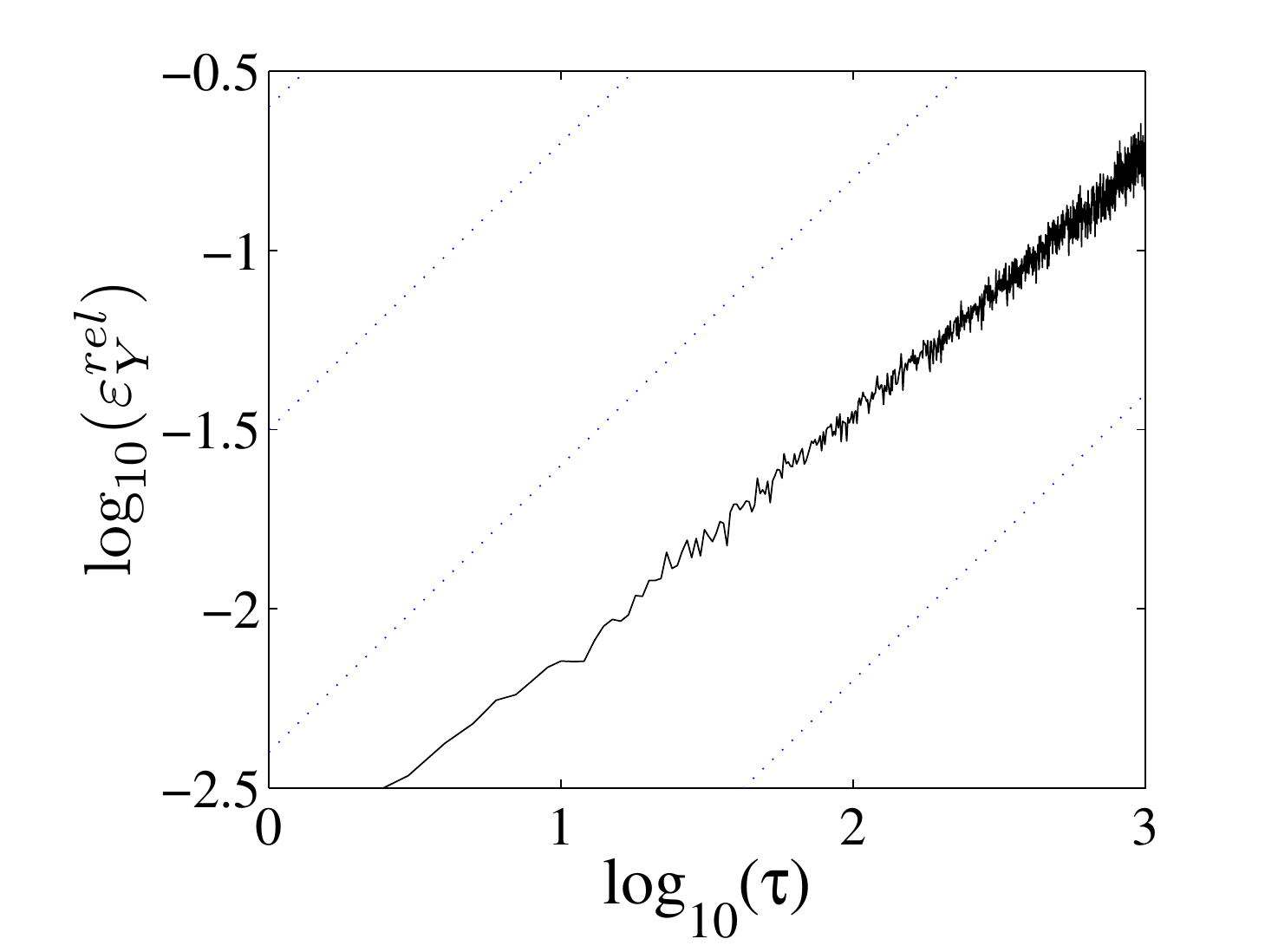}
            }
           % \subfigure[$H = 0.7$]{
%            \includegraphics[scale=0.25]{simulation_ci_cy_r100_H7.pdf}
%            }\\
            \hspace{-23pt}
            \subfigure[$H = 0.8$]{
            \includegraphics[width=0.5\columnwidth]{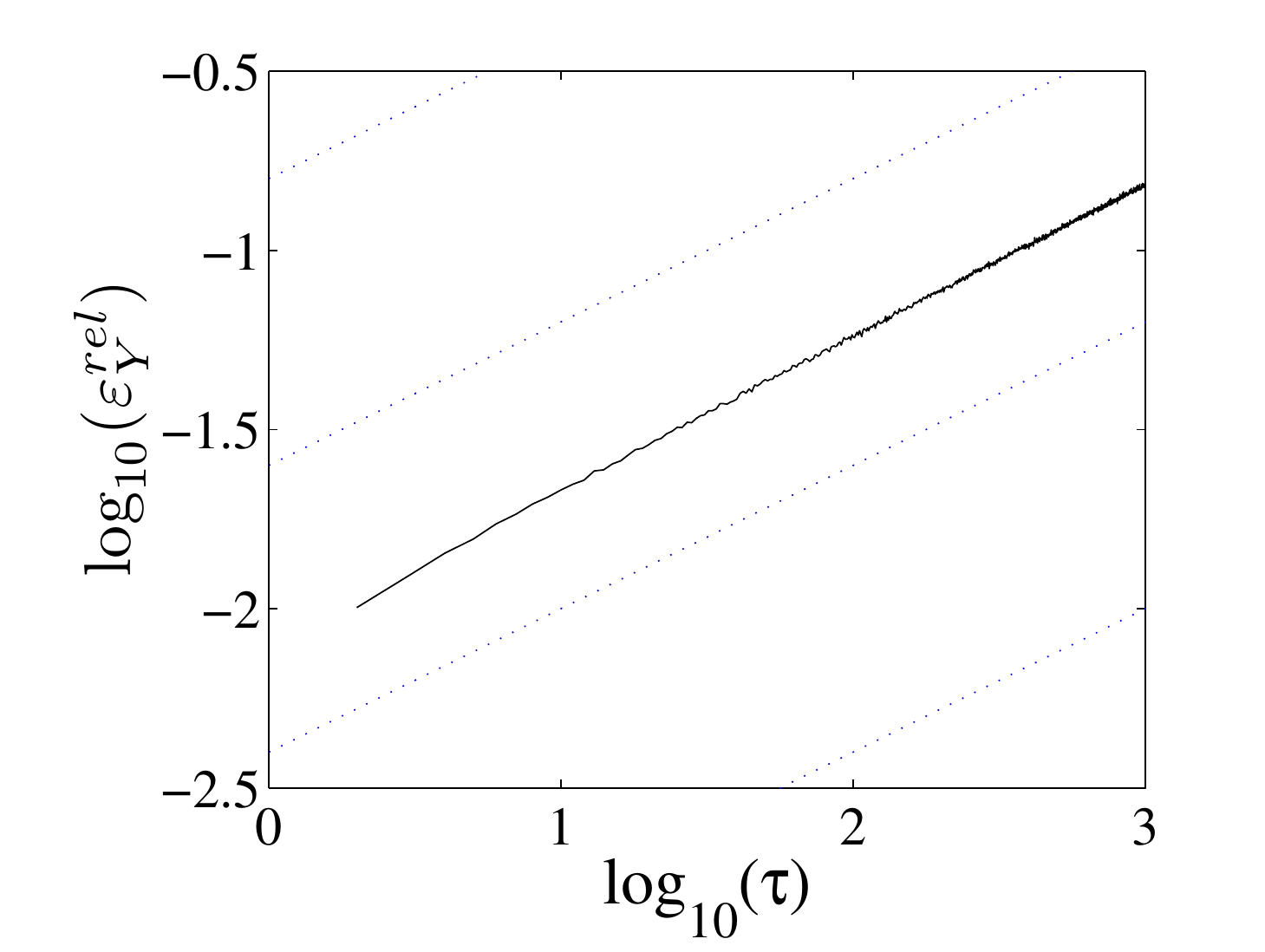}
            }
            %\subfigure[$H = 0.9$]{
%            \includegraphics[scale=0.25]{simulation_ci_cy_r100_H9.pdf}
%            }
            \caption{Estimation error under finite sampling depends on $H$.}
	\label{fig:ci_simulation_poisson}
\end{figure}
In addition, we calculate the needed sampling duration $T$ to achieve constant $\varepsilon_Y^{rel} $ for a given lag~$\tau$, and fixed $\mu_A$, $\mu_Y$ and $\sigma_Y$. We find from \eqref{eq:cy_error_relative_large_lag} that the sampling duration has to increase as $ T \sim \max\{\tau^{4-4H},\tau\}$, which again reveals the impact of $H$. Specifically, for $H < 0.75$ the sampling duration has to increase faster than linearly with $\tau$ to achieve constant $\varepsilon_Y^{rel}$.
%
%\paragraph*{\textbf{Bias of autocovariance estimators}}
\subsubsection{\textbf{Bias of autocovariance estimators}}
\label{sect:bias_of_covariance_estimators}

Next we investigate the accuracy of the deployed statistical estimators. The impact of the finite sample size carries forward to the computation of the autocovariance of $Y(t)$. First we consider the case where we directly observe $Y(t)$ for a finite duration $T$. We consider the autocovariance estimator
$\tilde{c}_Y(\tau) = \frac{1}{T-\tau}\sum_{t=1}^{T-\tau}\left(y(t) - \tilde{\mu}_{Y_0} \right) \left(y(t+\tau) - \tilde{\mu}_{Y_\tau} \right)$ with $\tilde{\mu}_{Y_i} = \tfrac{1}{(T-\tau)} \sum_{t=1}^{T-\tau} y(t+i)$.
An estimator of the autocovariance is unbiased iff $\mathsf{E}\left[\tilde{c}_Y(\tau)\right] = {c}_Y(\tau)$. To inspect the bias of $\tilde{c}_Y(t)$, we calculate its expected value and find
\begin{eqnarray}
\label{eq:expected_value_standard_autocov_estimator}
\mathsf{E}\left[\tilde{c}_Y(\tau)\right] \approx c_Y(\tau) - \frac{\sigma_Y^2}{(T-\tau)^{2 -2H}}.
\end{eqnarray}
The derivation of~(\ref{eq:expected_value_standard_autocov_estimator}) is given in appendix~\ref{subsec:bias_autocov_y}.

From (\ref{eq:expected_value_standard_autocov_estimator}) we conclude that the autocovariance estimator $\tilde{c}_Y(\tau)$  is asymptotically unbiased for $T \rightarrow \infty$ and $T \gg \tau$. The maximum lag, up to which the autocovariance is estimated, must be chosen carefully, such that the bias in (\ref{eq:expected_value_standard_autocov_estimator}) becomes negligible. However, the bias depends on $H$ such that higher $H$ require larger $T$.

After considering the entire process $Y(t)$ we now investigate the bias of the autocovariance estimator when applied to $W(t)$ as observed by sampling with finite duration $T$. We calculate the expected value of the estimated autocovariance
\begin{eqnarray}
\mathsf{E}[\tilde{c}_W(\tau)] & \approx & c_W(\tau) - \frac{c_W(0)}{T-\tau} \nonumber\\ \label{eq:expected_value_autocov_W}
 & & - \frac{2}{(T-\tau)^2}\!\sum_{t=1}^{T-\tau-1} (T-\tau-t) c_W(t).
\end{eqnarray}
The derivation of (\ref{eq:expected_value_autocov_W}) is given in appendix~\ref{subsec:bias_autocov_w}. The bias in (\ref{eq:expected_value_autocov_W}) goes to zero for $T \rightarrow \infty$ and $T \gg \tau$.

In the remainder of this section we provide brief conclusions that highlight our main findings. We presented a framework for extracting the traffic autocovariance from observed samples. From our evaluation of the sampling distributions we conclude, that the covariance observed under geometric sampling does not exhibit any distortions. This property greatly simplifies the reconstruction of the covariance of the original process $Y(t)$, as no additional parameters, such as $\mu_Y$, must be estimated. Hence, for geometric sampling with sufficiently large $T$ we use $\tilde{c}_W(\tau)/\mu_A^2$ as an estimator of the traffic autocovariance. From the evaluation of the estimator we find two major aspects that limit the observability for finite sampling sizes. First, finite sampling sizes yield computable distortions given in Sect.~\ref{sect:Observation limit} and \ref{sect:Estimation accuracy} which may obscure the true covariance structure. Secondly, the bias for covariance estimators depends on the Hurst parameter, such that longer measurements must be conducted for traffic exhibiting strong LRD.

Nevertheless, finite sampling effects disappear in the limit for large sampling durations. Moreover, we found that increasing the probing intensity improves estimation results more quickly than increasing the sampling duration.
\section{Impact of sampling on selected $H$ estimators}
\label{sec:other}
In this section we analyze the impact of sampling on two established Hurst parameter estimation techniques, i.e., the aggregate variance, respectively, the spectral density method.
\subsection{Aggregate variance}
\label{subsec:aggregate_variance}
\begin{figure}[t]
        \centering
            \includegraphics[width=1.0\columnwidth]{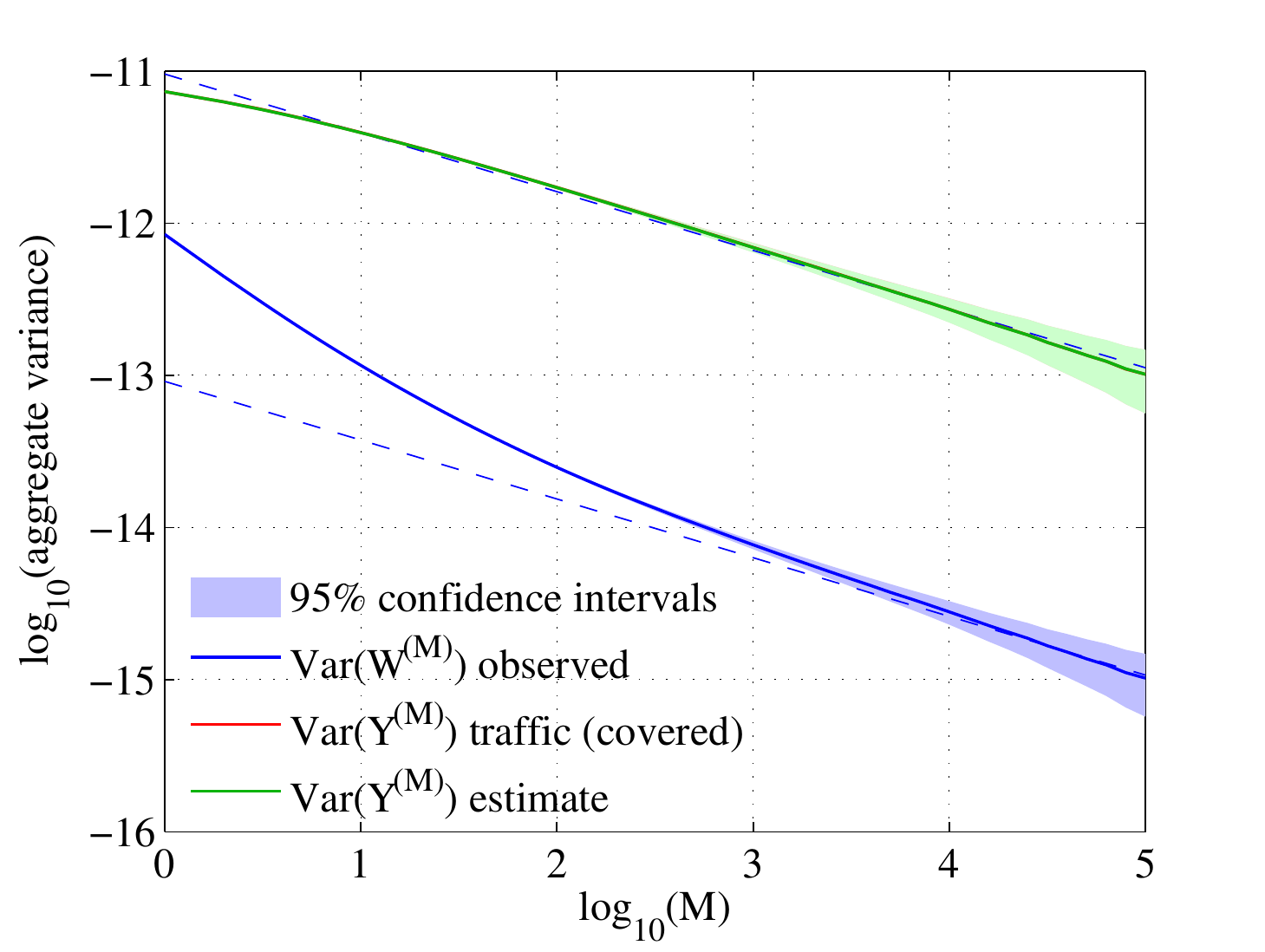}
        \caption{Aggregate variance estimate under geometric sampling. The variance $\text{Var}(Y^{(M)})$ of the traffic process $Y(t)$ is covered by the estimate obtained from \eqref{eq:agg_var_geo} using the variance of the observations $\text{Var}(W^{(M)})$ respectively the sampling process $\text{Var}(A^{(M)})$. }
	\label{fig:agg_var_sampling_periodic_lng_trace_H8}
\end{figure}
Briefly, the method exploits the fact that the variance of an LRD, self-similar process considered at different aggregation time scales decays linearly in $H$ on a log-log scale. As outlined in \cite{taqqu:estimators}, we divide an LRD process $Y(t)$ into blocks of size $M$, denoted aggregation level, and average within each block. The aggregate time series on the aggregation level~$M$ of $Y(t)$ is obtained as
\begin{eqnarray}
\label{eq:aggregation_general}
Y^{(M)}(k) = \frac{1}{M} \sum_{t = 1+(k-1)M}^{kM} Y(t) \quad \text{for} \quad  k \in \mathbb{N}.
\end{eqnarray}
The variance of the sample means is known to decay with the block size as $M^{2H-2}$. The Hurst parameter $H$ is obtained from the corresponding slope on a log-log scale.

The following lemma shows the impact of the sampling process on the aggregate variance of the observed process $W(t)$.
\begin{lemma}
\label{lem:agg_var}
Given \eqref{eq:process_sampling} and the aggregation rule \eqref{eq:aggregation_general}. For the aggregate processes $A^{(M)}$, $W^{(M)}$, and $Y^{(M)}$ it holds that
\begin{eqnarray*}
  \text{Var}(W^{(M)}) \!\!\!\! &=& \!\!\!\! \mu_Y^2 \text{Var}(A^{(M)}) + \mu_A^2 \text{Var}(Y^{(M)})\\
  & & \!\!\!\!+ \frac{1}{M}\sigma_Y^2 \sigma_A^2 + \frac{2}{M^2} \sum_{\tau=1}^{M-1} (M-\tau) c_Y(\tau)c_A(\tau).
\end{eqnarray*}
\end{lemma}
The proof of Lem. \ref{lem:agg_var} is given in the appendix \ref{subsec:aggvar_of_sampled_process}. Lem. \ref{lem:agg_var} illustrates the relation between the variances of the three aggregate processes $A^{(M)}$, $W^{(M)}$ and $Y^{(M)}$. From Lem. \ref{lem:agg_var} we see the impact of sampling on the observed $\text{Var}(W^{(M)})$ through $c_A(\tau)$ and the parameters of $A(t)$, i.e., $\mu_A$ and $\sigma_A^2$.

Again the advantage of geometric sampling is apparent as $c_A(\tau)=0$ allows solving for $\text{Var}(Y^{(M)})$ directly using the variances of the sampling process $\text{Var}(A^{(M)})$, respectively, of the observed process $\text{Var}(W^{(M)})$ as
\begin{eqnarray}
\label{eq:agg_var_geo}
\text{Var}(Y^{(M)})\!\!\!& = &\!\!\!\frac{1}{\mu_A^2} \text{Var}(W^{(M)}) \nonumber \\ & & \!\!\! - \frac{1}{\mu_A^2} \left(\mu_Y^2 \text{Var}(A^{(M)}) + \frac{1}{M}\sigma_Y^2 \sigma_A^2\!\right).
\end{eqnarray}
The estimate $\text{Var}(Y^{(M)})$ in \eqref{eq:agg_var_geo} requires estimates of the mean and variance of the traffic process. These estimates can be obtained from $W(t)=A(t)Y(t)$ as $\mu_Y = \mu_W/\mu_A$ respectively $\sigma_Y^2 = (\sigma_W^2-\sigma_A^2\mu_Y^2)/\mu_A$ where we used that $\sigma_A^2+\mu_A^2=\mu_A$ for the Bernoulli sampling process $A(t)$.

Fig.~\ref{fig:agg_var_sampling_periodic_lng_trace_H8} depicts the decay of $\text{Var}(W^{(M)})$ as a function of~$M$ for geometrically sampled observations. We consider the same scenario and parameters as for Fig.~\ref{fig:xcov_poisson_sim} in Sect.~\ref{subsec:sampling}. Note that $\text{Var}(Y^{(M)})$ of the original process is covered by the estimated aggregate variance using \eqref{eq:agg_var_geo}, denoted ``$\text{Var}(Y^{(M)})$ estimate". A Hurst parameter estimate is deduced from the slope of ``$\text{Var}(Y^{(M)})$ estimate", i.e., $2H-2$. The correct slope for the evaluated $H=0.8$ is indicated in the figure by the dashed auxiliary lines.

For the remaining sampling distributions discussed in Sect.~\ref{subsec:sampling} the inversion of Lem.~\ref{lem:agg_var} for $\text{Var}(Y^{(M)})$ is not easily possible as $c_A(\tau) \neq 0$ such that the term that contains $c_Y(\tau)$ persists.

In general, as the block size $M$ increases, the impact of the terms in the second line in Lem.~\ref{lem:agg_var} diminishes, however. For $M \rightarrow \infty$ the relationship in Lem.~\ref{lem:agg_var} tends to
\begin{eqnarray*}
\text{Var}(W^{(M)}) \approx \mu_Y^2 \text{Var}(A^{(M)}) + \mu_A^2 \text{Var}(Y^{(M)}).
\end{eqnarray*}
In particular, for $M \rightarrow \infty$ the observed $\text{Var}(W^{(M)})$ tends to $\mu_A^2 \text{Var}(Y^{(M)})$. This is due to the fact that sampling processes $A(t)$ considered here are \emph{not} LRD. Hence, $\text{Var}(A^{(M)})$ decays with slope $-1$ with $M$ on a log-log scale, whereas $\text{Var}(Y^{(M)})$ decays with $2H-2$. This effect is visible in Fig.~\ref{fig:agg_var_sampling_periodic_lng_trace_H8} as $\text{Var}(W^{(M)})$ tends for increasing $M$ to the auxiliary dashed line of slope $2H-2$.
\subsection{Spectral density}
Spectrum based Hurst parameter estimators rely on the characteristics of the frequency domain representation of LRD processes. The spectral density of an LRD process $Y(t)$ possesses the behavior given in \eqref{eq:psd_fgn} \cite{taqqu:estimators,beran94}. Hence, $H$ can be estimated from the logarithm of the spectral density of $Y(t)$ plotted vs. $\log(f)$.

We rephrase Lem.~\ref{lem:sample_covariance} for the spectral density to find
\begin{eqnarray}
\label{eq:convolution_spectral_density}
\Psi_W(f) = \Psi_A(f) \ast \Psi_Y(f)
\end{eqnarray}
with $\Psi_{(\cdot)}(f)$ denoting the spectral density of the process $(\cdot)$. We use $\ast$ to denote the convolution defined as~$x\ast~\!\!y(t) := \int_{-\infty}^{\infty}x(\tau)y(t-\tau) d\tau$. To prove~\eqref{eq:convolution_spectral_density} we use the Wiener-Khinchin theorem, which states that the spectral density is given by the Fourier transform of the autocorrelation function. The autocorrelation of $W(t)=A(t)Y(t)$, i.e., $\mathsf{E}\left[W(t)W(t+\tau)\right]$, is obtained as the product of the autocorrelation functions of $A(t)$ and $Y(t)$. Finally, the Fourier transform of the product of two functions $\mathcal{F}\left\{x \cdot y\right\}(f)$ is given by the convolution of the respective Fourier transforms $\mathcal{F}\left\{x\right\} \ast \mathcal{F}\left\{y\right\}(f)$.

For ease of exposition we consider continuous time memoryless sampling, i.e., inter-sample times drawn from an exponential distribution with parameter $\lambda$. The spectral density of $A(t)$ is given, e.g., in \cite{cox66} as $\Psi_A(f)\!\!~=~\!\!\lambda^2 \delta(f)+\lambda$ with the well known Dirac delta function $\delta(f)$ that is defined as $\int_{-\infty}^{\infty} g(f) \delta(f) df = g(0)$ \cite{stirzaker01}. From \eqref{eq:convolution_spectral_density} we calculate the spectral density of the observed process $W(t)$ as
\begin{eqnarray}
\label{eq:psd_w}
\Psi_W(f) \!\!\!& = &\!\!\! \Psi_A(f) \ast \Psi_Y(f) \nonumber\\
\!\!\!& = &\!\!\! (\lambda^2 \delta(f)+ \lambda) \ast \Psi_Y(f) \nonumber\\
\!\!\!& = &\!\!\! \lambda^2 \Psi_Y(f) + \lambda \left(\sigma^2_Y + \mu^2_Y\right)
\end{eqnarray}
The convolution $\lambda\!\ast\!\Psi_Y(f)$  reduces to $\lambda\!\left(\sigma^2_Y\!+\!\mu^2_Y\right)$ using the Wiener-Khinchin theorem as $\lambda\int_{-\infty}^{\infty}\Psi_Y(f)df=\lambda\mathsf{E}\left[Y(t)^2\right]$. Consequently, the spectral density of $Y(t)$ can be estimated by solving (\ref{eq:psd_w}) for $\Psi_Y(f)$. The estimate $\Psi_Y(f) $ requires estimates of the mean and variance of the traffic process, i.e., $\mu_Y$ and $\sigma_Y^2$ respectively, that can be obtained as in Sect.~\ref{subsec:aggregate_variance}. An estimate of $H$ is obtained from the slope of $\Psi_Y(f)$ on a log-log scale.

Next, we consider periodic sampling in conjunction with the spectral density method. For periodic sampling it is known that the sampling frequency has to be twice the highest frequency contained in the sampled process to avoid aliasing. This is known as Nyquist criterion. For periodic sampling $\Psi_A(f)$ is a dirac comb with inter-dirac distances of $1/\Delta$ where $\Delta$ is the sampling period. This leads to a repetition of the spectrum $\Psi_Y(f)$ at distances $1/\Delta$. An \emph{irreversible} spectral overlap occurs if the Nyquist criterion is not met, which is the case here as $\Psi_Y(f)$ is not band limited. As a result the method of spectral estimation of the Hurst parameter cannot be used directly with periodic sampling. Also, the remaining sampling strategies from Sect.~\ref{subsec:sampling} may not yield an expression for $\Psi_W(f)$ in (\ref{eq:psd_w}) that can be solved for $\Psi_Y(f)$.

As in Sect.~\ref{sec:sampling} we find that geometric sampling provides a substantial advantage when deploying established $H$ estimators to sampled observations. For the considered methods, \eqref{eq:agg_var_geo} and \eqref{eq:psd_w} provide the measures needed for $H$ estimation.
\section{Active probing}
\label{sec:network_probing}
So far, we focused on the estimation of traffic correlations using passive sampling. In large multi-provider networks like the Internet, service providers often do not provide such network traces, e.g., for reasons of competition. The estimation of traffic correlations, therefore, must rely on inferring samples of the Internet traffic from network metrics that can be easily observed at end systems, e.g., by active probes. Moreover, passive sampling is a priori limited to single links. In case of network paths, where the correlations of the end-to-end service involve multiple nodes and links, end-to-end measurements may be the only viable option. We present an active probing method that enables users to characterize end-to-end paths, with minimal effort and without administrative support from the  network under observation.

In this section, we address the fundamental problem of inferring the correlation of LRD traffic using active probes. We propose a new active probing method which collects traffic samples by detecting router busy periods. The observations are used to estimate the covariance of the end-to-end service. Subsequently, we estimate the corresponding Hurst parameter.
Furthermore, we show that the well known packet pair dispersions approach, which captures the traffic intensity at the ingress of a router, is also applicable for the derivation of LRD traffic correlations.
In the sequel, we describe our probing methodology and discuss traffic correlation estimation for both the single and multi-node cases. We then show testbed measurements to demonstrate the feasibility of our method. Finally, we present a set of Internet measurement results showing end-to-end correlations of entire network paths.

\subsection{Probing methodology}
\label{sec:packet-delay-probing}
We investigate two probing methods that facilitate the inference of certain characteristics of network traffic, referred to as cross traffic. The two methods differ with respect to the probes, i.e., single packets and packet pairs, and the observed metric, i.e., delay and packet pair dispersion, respectively.
\subsubsection{Single packet probes}
\label{sec:single-packet-probes}
To extract an estimate of the cross traffic autocovariance, we propose an approach which uses the delays of single packet probes to detect busy periods at a router, and hence samples the link utilization at the router egress. For the remainder of this work, cross traffic denotes any traffic sharing resources with the probing traffic.

We make the general assumption that packet scheduling is non-preemptive. Hence, whenever a router is busy transmitting a packet, the delay $d_p$ experienced by an arriving packet will be greater than the minimal delay $d_{min}$ experienced when the router is idle. Consequently, we can sample cross traffic increments at the router egress, by injecting probe packets and analyzing their delays. For each probe, we measure the one way delay $d_p=t_r-t_s$, using the send and receive times $t_s$ and $t_r$, respectively. To determine if the router was busy, we check whether the observed delay $d_p$ is greater than the minimum network delay $d_{min}$. As a result, each probe yields a sample of the egress link state at time $t=t_s$, and the observed process can be constructed as
\begin{eqnarray}
  \label{eq:vector_W_rtt}
  W(t) = \left\{
    \begin{array}{l l}
      1 & \quad \mbox{if $d_p > d_{min}$ and $A(t)=1$}\\
      0 & \quad \mbox{otherwise.}\\
    \end{array} \right.
\end{eqnarray}
It is known \cite{erramilli96,ganesh:04}, that the covariance structure of LRD traffic is preserved at the output of a queue or a traffic shaper, such that $W(t)$ permits observing the covariance of the cross traffic. We assume that the perturbation of the observed traffic due to probe size and probing rate is negligible, since the probing rates used are typically less than one per mill of the capacity.
Furthermore, as we can assume that dropped probes are due to a busy router, we account for lost probing packets by setting $W(t)=1$ for all dropped probes of $A(t)$.
\subsubsection{Packet pair probes}
\label{sec:packet-pair-probes}
In the following we present a second probing technique to estimate the cross traffic autocovariance that uses the dispersion of packet pair probes. Packet pair probing is a popular method for estimating the available bandwidth of a bottleneck link with FIFO queueing, e.g.,~\cite{spruce}. A probe consists of two packets, each of size $L$, that are sent into the network back-to-back. The gap between the packet send times $t_{s,1}$ and $t_{s,2}$ is set to $g_s=t_{s,2}-t_{s,1}=L/C$, where $C$ is the capacity of the bottleneck link. The subscript denotes the packet location and number, e.g., $t_{s,1}$ denotes the time stamp of the first packet of the probe pair at the sender. The packet pair dispersion $g_r=t_{r,2}-t_{r,1}$ observed at a receiver yields a sample of the cross traffic intensity $W(t)$ at $t=t_{s,1}$ as~\cite{spruce}
\begin{eqnarray}
  \label{eq:spruce}
  W(t) = C \left(\frac{g_r - g_s}{g_s}\right).
\end{eqnarray}
We set $W(t)=0$ if no probe is sent at $t$ as well as if probes are lost or incomplete at the receiver.

We apply this method to infer cross traffic intensities by injecting packet pair probes at times $t_s$ and measuring the packet dispersion $g_r$ at the probe receiver. Note that the multiplicative constant $g_s/C$ in~\eqref{eq:spruce} does not alter the covariance structure. Hence, we can drop $g_s/C$ and estimate the covariance of the cross traffic process from $W(t) = g_r-g_s$ if $A(t)=1$ and zero otherwise without knowing the absolute value of the bottleneck link capacity.

The packet pair method requires a sufficiently accurate time stamp resolution at the receiver to correctly capture the variations in $g_r$. As the first packet functions as a time reference, while the traffic intensity is sampled by the second packet, the approach does not require a synchronization between the sender and receiver clocks. However, this results in an overhead associated with each probe. As an example, when equally sized packets are used, $50\%$ of the probing load is ``wasted'', thereby reducing the effective sampling resolution for a given probing rate by a factor of two. In addition, the extension of packet pair probing to the multi-node case is not trivial. Therefore, we proceed using the first method, i.e., by detecting router busy periods.
\subsection{Measuring LRD in single- and multi-node scenarios}
\label{sec:multi-hop}
In this subsection we consider estimating traffic correlations in multi-node scenarios using the busy period detection technique \eqref{eq:vector_W_rtt}. We now show that the observed process $W_{N}(t)$ at the egress of an $N$ node path with LRD cross traffic also exhibits LRD behavior. Moreover, for cross traffic characterized by different Hurst parameters, we show that the largest Hurst parameter dominates the covariance of the observed process $W_{N}(t)$. These results are in agreement with~\cite{ganesh:04} which shows that the largest Hurst parameter dominates end-to-end performance.

\begin{figure}[t]
  \center
  \includegraphics[width=0.95\columnwidth]{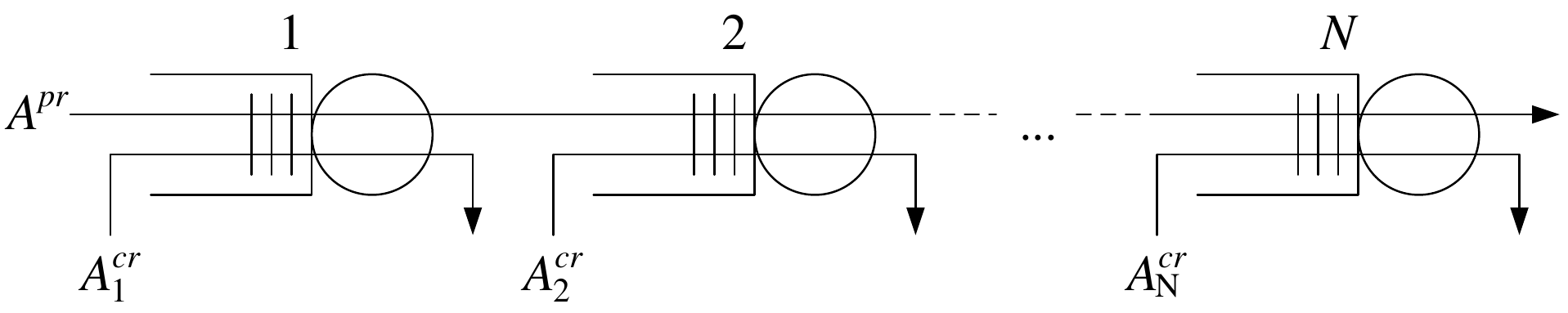}
  \caption{$N$ node topology with probing traffic $A^{pr}$ and LRD cross traffic $A_{i}^{cr}$ for nodes $i \in [1,N]$.}
  \label{fig:scenario}
\end{figure}
Consider an $N$ node topology with independent LRD cross traffic as in Fig.~\ref{fig:scenario}. We describe the busy state of each node using the processes $Y_i(t)$ for node $i \in [1,N]$. Hence, $Y_i(t)=1$ if node $i$ is busy at time $t$ and $Y_i(t)=0$ otherwise. Note that the covariance $c_{Y_i}(\tau)\sim \tau^{2H_i-2}$ measured at the egress of node $i$ has the same LRD property as the cross traffic input at the node \cite{erramilli96,ganesh:04}. Next, consider an active probe that is injected into the path. After subtracting the minimum end-to-end delay $d_{\min}$ the observer at the egress of the path will measure a positive delay only if any of the routers was busy when the probe arrived at the respective router. Otherwise, the probe delay will equal zero. Hence, $W_N(t)$ is the logical OR operation of the individual processes $Y_{i}(t)$ for $i \in [1,N]$. Since $Y_i(t)$ and $W_i(t) \in \{0,1\}$, we straightforwardly find $W_i(t)$ at the egress of node $i$ as
\begin{eqnarray}
\label{eq:Nhopdelayformula}
\!W_i(t) = \left\{
  \begin{array}{l l}
    \!\!Y_1(t)&\!\!\!\text{, } i=1\\
    \!\!W_{i-1}(t)\!+\!Y_i(t)\!-\!W_{i-1}(t)Y_i(t)&\!\!\!\text{, } i \in[2,N].\\
  \end{array} \right.
\end{eqnarray}
For ease of exposition,~\eqref{eq:Nhopdelayformula} assumes that a node that is idle forwards probe packets instantaneously to the next node, such that the probe packet observes $Y_i(t)$ at the same time instance $t$ for all $i \in [1,N]$. Dispensing with this assumption,~\eqref{eq:Nhopdelayformula} can be formulated in the same way requiring, however, additional notation as a probe packet observes $Y_i(t)$ at $t = t_i$ where $t_i \ge t_{i-1}$ for $i \in [2,N]$.

First, we illustrate \eqref{eq:Nhopdelayformula} using a two node example and two independent LRD processes $Y_1(t)$, $Y_2(t)$. The observed process at the egress of node~$2$ is $W_2(t)=1$ if $Y_1(t)=1$ OR $Y_2(t)=1$ and $W_2(t)=0$ otherwise, such that we deduce
\begin{eqnarray*}
%  \label{eq:2hop_delay_formula}
  W_2(t) = Y_1(t) + Y_2(t) - Y_1(t)Y_2(t).
\end{eqnarray*}
We derive the observed covariance~$c_{W_2}(\tau)$ of $W_2(t)$ after some algebra as
\begin{multline*}
%  \label{eq:2hop_cov}
\!\!c_{W_2}(\tau)\!=\!c_{Y_1}(\tau)c_{Y_2}(\tau) + c_{Y_1}(\tau)(1 - \mu_{Y_2})^2 + c_{Y_2}(\tau)(1 - \mu_{Y_1})^2.
\end{multline*}
The equation above directly shows that for large $\tau$ the covariance $c_{W_2}(\tau)$ is dominated by $c_{Y_i}(\tau)$ with the largest Hurst parameter, i.e., slowest decay. The covariance of the $N$-node end-to-end observations $c_{W_N}(\tau)$ is obtained using the recursion formula \eqref{eq:Nhopdelayformula} as
\begin{multline}
\label{eq:Nhop_cov}
\!\!c_{W_i}(\tau)\!=\!c_{W_{i-1}}(\tau)c_{Y_i}(\tau) \\ + c_{W_{i-1}}(\tau)(1 - \mu_{Y_i})^2 + c_{Y_i}(\tau)(1 - \mu_{W_{i-1}})^2.
\end{multline}
Using recursive substitution, it can be shown that the covariance of the end-to-end observations $c_{W_N}\!(\tau)$ is dominated by $c_{Y_i}(\tau)$ with the largest Hurst parameter for $i\!\in\![1,N]$.
\subsection{Probing Software and Experimental Setup}
\label{sec:hurstping}
We developed a probing tool \emph{H-probe}, available at~\cite{hprobe:software}, that performs online measurements to infer the covariance structure of the round trip service of network paths. \emph{H-probe} injects ICMP echo request probes from the sender to the receiver and captures the associated round trip times using \emph{libpcap}. In contrast to the (optional) client/server delay measurements using UDP, the use of ICMP probes is significantly more practical, as it circumvents clock synchronization issues and enables probing the path to any network host without the need for a receiver software. \emph{H-probe} uses the method described in Sect.~\ref{sec:packet-delay-probing} and the statistical analysis discussed in Sect.~\ref{subsec:sampling} and optionally the $H$ estimation technique from Sect.~\ref{subsec:aggregate_variance}. For $H$ estimation using the aggregate variance method a range of scales has to be fixed for the estimation \cite{beran94}. A similar observation is made in the context of Hurst parameter estimation using wavelet decomposition \cite{Veitch03}. The authors in~\cite{Veitch03} describe upper and lower cutoff bounds on the time scales considered for the estimation using wavelets. Similarly, we fix the range of scales $M$ for the $H$ estimation with upper and lower cutoffs $M^{up}$ resp. $M^{low}$. We define the upper range end $M^{up}$ depending on the sample size $T$, i.e., $T/M^{up}=10^{2}$ to ensure enough points for the variance calculation at $M=M^{up}$. We fix $M^{low}=0.1$~s. This range is consistent with the ranges reported in trace driven $H$ estimation literature, e.g., \cite{veitch:05:multifractal,gupta:09:coexistance}.

\begin{figure}[t]
  \center
  \includegraphics[width=0.95\columnwidth]{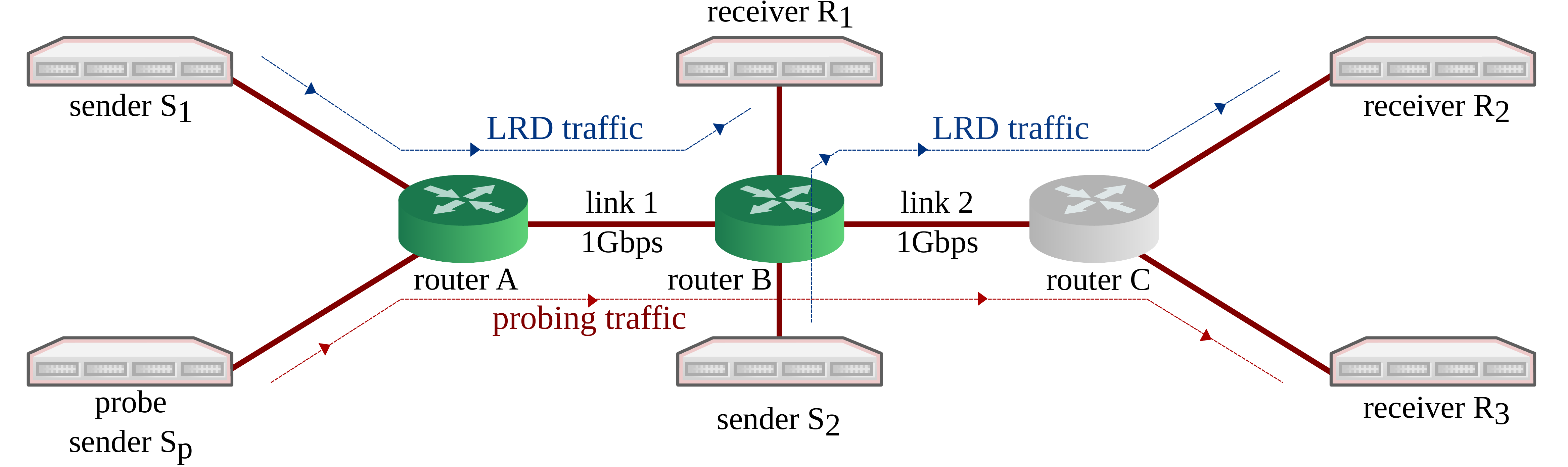}
  \caption{Experimental setup: Emulab testbed}
  \label{fig:topology1}
\end{figure}
\begin{figure}[t]
\includegraphics[width=1.00\columnwidth]{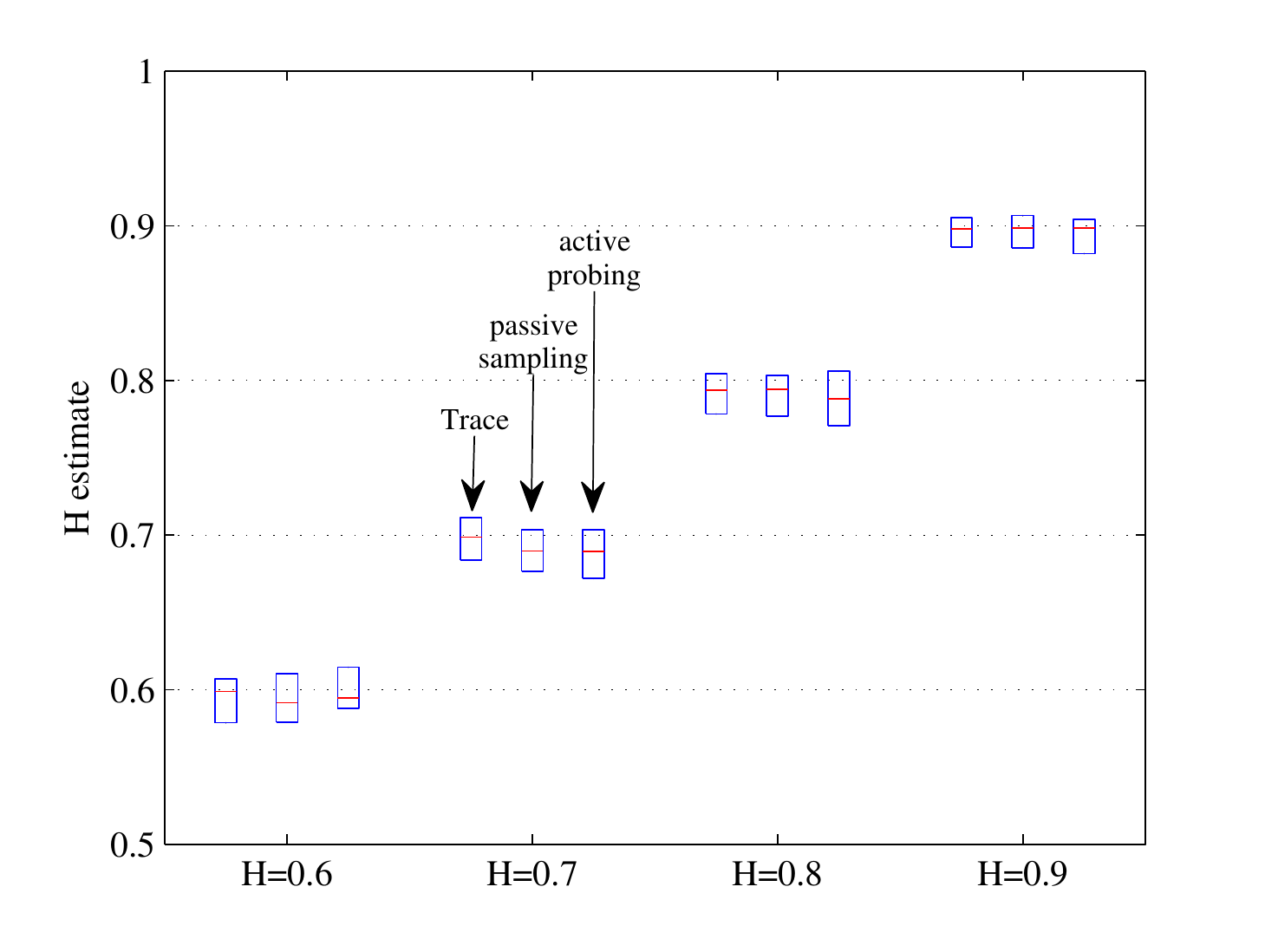}
   \caption{Hurst parameter estimates from (a) offline trace analysis, (b) offline trace sampling and (c) active probing in the Emulab testbed.}
   \label{fig:boxplot_1hop}
\end{figure}
\begin{figure*}[t]
            \centering
            \subfigure[planetlab1.cis.upenn.edu]{
              \label{fig:xcov_upenn}
              \includegraphics[type=pdf,ext=.pdf,read=.pdf,scale=0.41]{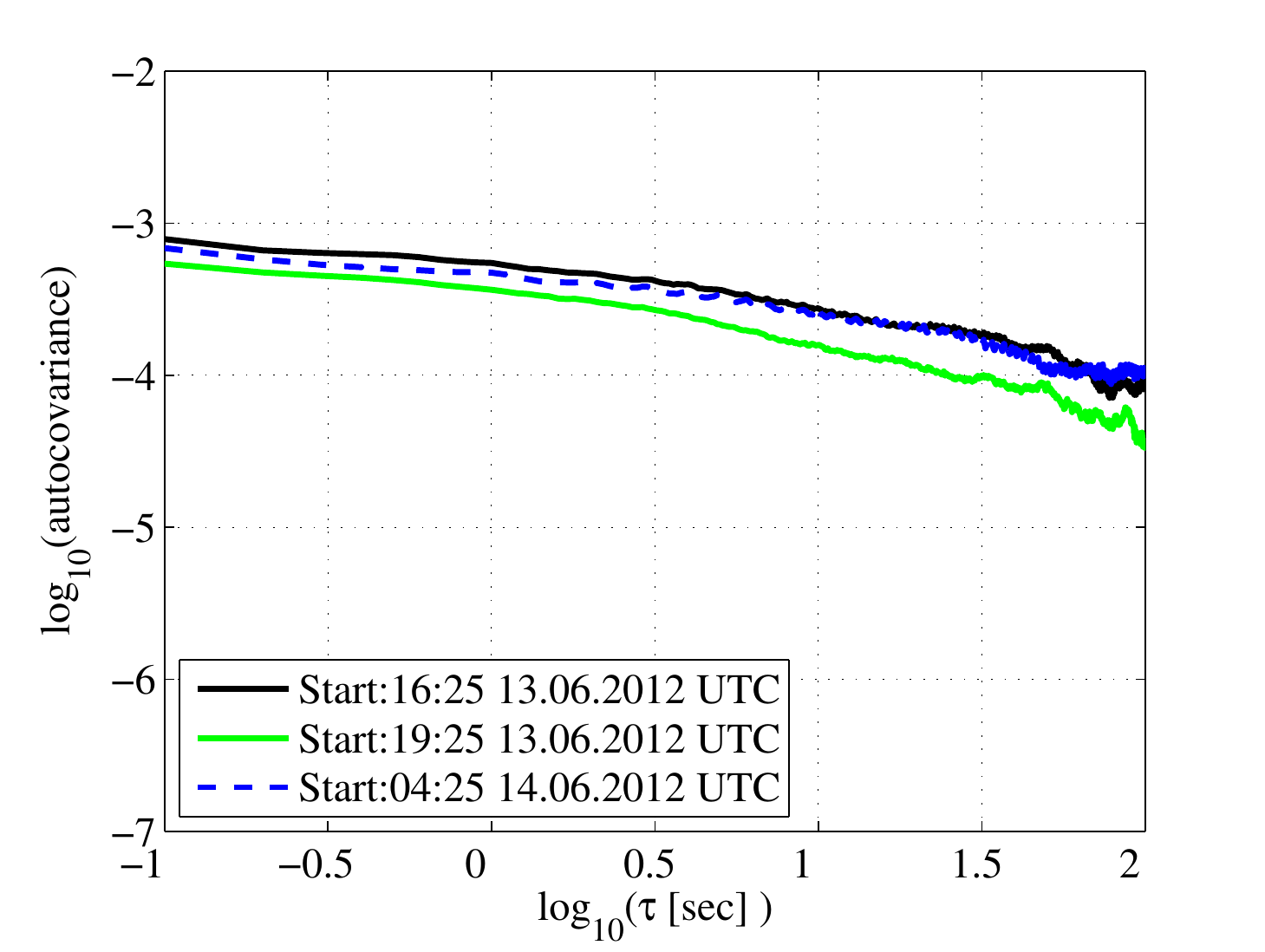}
            }
            \hspace{-20pt}
            \subfigure[planetlab01.sys.virginia.edu]{
              \label{fig:xcov_virginia} \includegraphics[type=pdf,ext=.pdf,read=.pdf,scale=0.41]{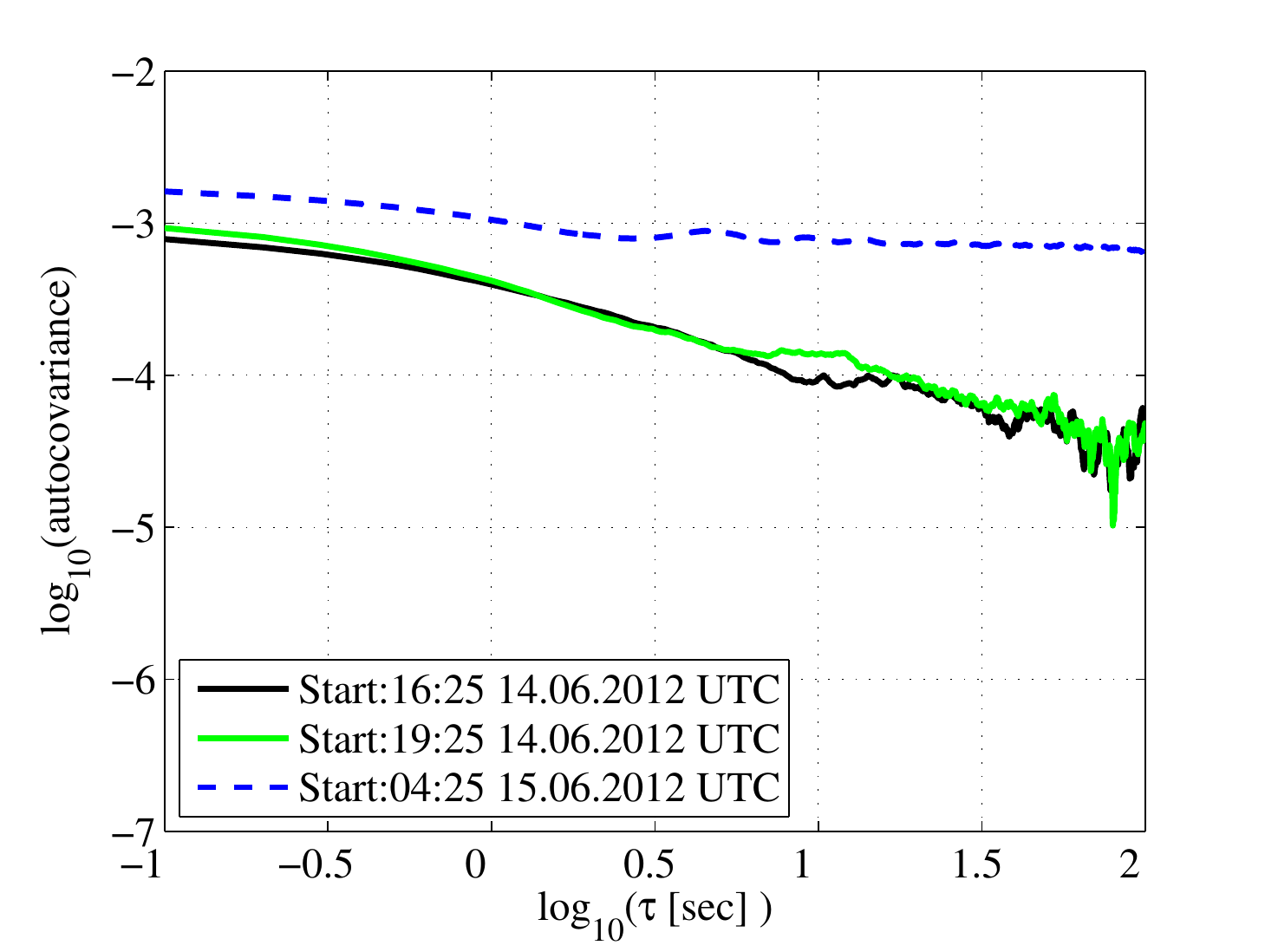}
            }
            \hspace{-20pt}
            \subfigure[planetlab1.cs.umass.edu]{
              \label{fig:xcov_umass}
              \includegraphics[type=pdf,ext=.pdf,read=.pdf,scale=0.41]{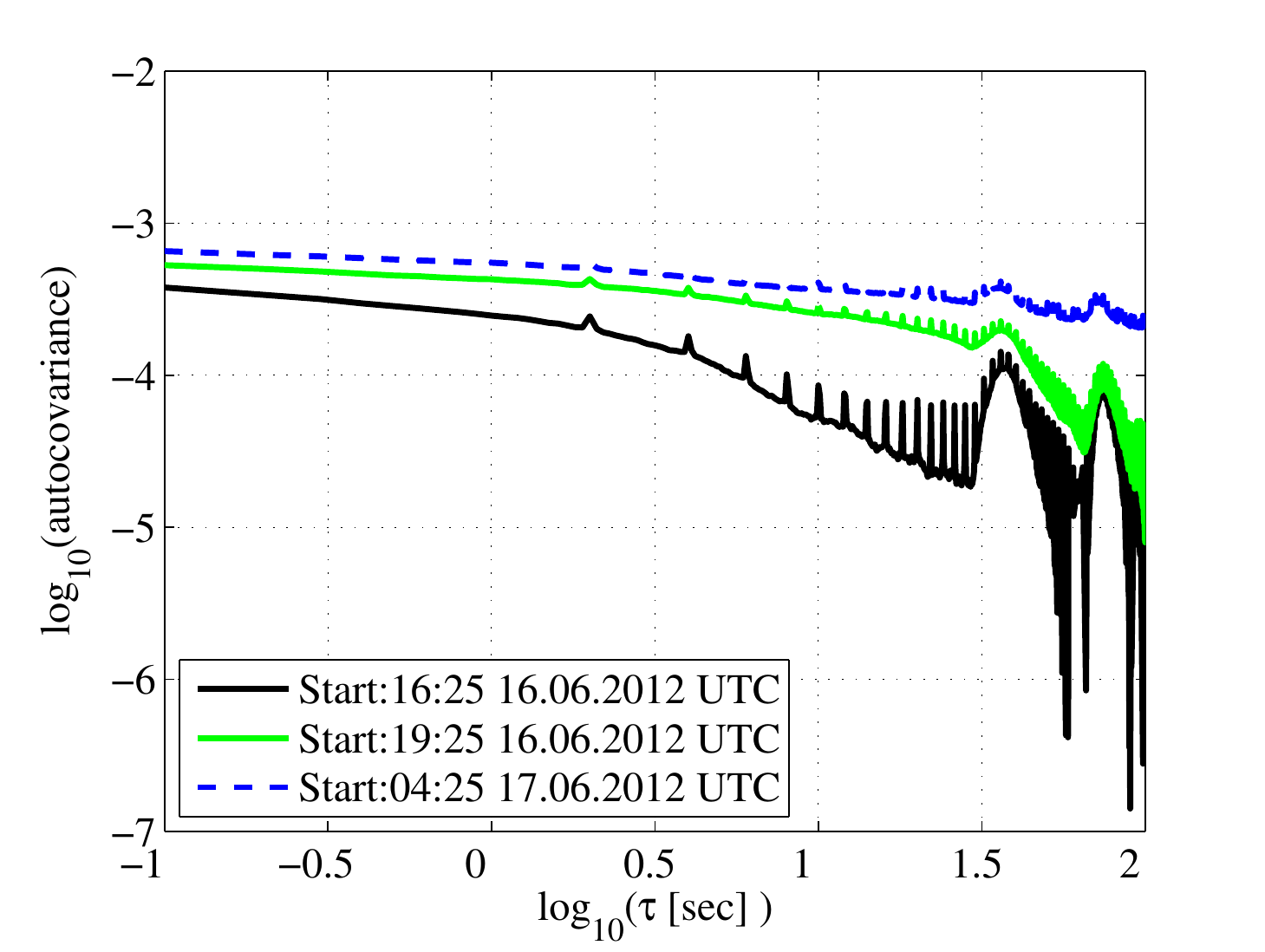}
            }
            % width=0.6\columnwidth
            \caption{End-to-end covariance estimates from Internet measurements. The covariance structure varies across different paths and for different times. For some targets we observe distinct periodicities on different time scales.}
            \label{fig:xcovsInternet}
\end{figure*}
In the following we present results obtained using this software package. Fig.~\ref{fig:topology1} depicts the experimental setup in our Emulab-based testbed\footnote{We use nodes with Supermicro X8DTU server mainboards   with 2.2Ghz Intel E5520 Xeon processors, quad port Intel 82576EB   Gigabit Ethernet Controllers, and Ubuntu 10.04 LTS with kernel   2.6.32-24, FIFO scheduling and buffers for 5000 packets. All links have a capacity of $C=1$Gbps.}. The topology comprises two relevant links, denoted link $1$ and $2$. Two traffic senders $S_i$, $i \in [1,2]$ transmit LRD cross traffic traces with defined Hurst parameter $H$ to the receivers $R_i$. The traces were synthesized by superposition of $10^5$ heavy tailed on-off sources with tail index $\alpha$. The relation between $H$ and the tail index $\alpha$ is given in \cite{willinger97}. We set the mean rate of the traffic at \emph{each} sender to $50$~Mbps, with a constant packet size of $1500$~Byte.

We use geometrically distributed inter-sample times with $p=0.1$ and slot length $\delta=1$~ms. For each measurement we send $10^6$ probes with a mean probing rate of $100$~packets per second (corresponding to $\sim70$~kbps) from the probe sender $S_p$ to the receiver $R_3$. We use the same parameters for the Internet measurements, substituting $R_3$ with PlanetLab nodes. To deal with non-queueing induced jitter in routers, \emph{H-probe} substitutes $d_{min}$ in \eqref{eq:vector_W_rtt} by the average $\mathsf{E}[d]$. This significantly reduces the measurement noise, because we can assume the distribution of this jitter is light tailed. We set the length of each measurement to $3$ hours over which we assume stationarity of the traffic processes.

%We use ICMP packets that have the advantage that the sender and %receiver may be on the same machine.
%We know that traffic engineering based on ICMP filter may distort measurement results. The method can be applied using UDP packets, however it is required to control the two hosts at the both ends of the target path.
%%%%%%%%%%%%%%%%%%%%%%%%

\subsection{Testbed measurements}
\label{sec:Testbed_measurements}
We deploy \emph{H-probe} in our Emulab testbed, in order to verify its functionality in a controlled environment. First, we inject synthetic LRD traffic with $H \in [0.6, 0.9]$ on link $1$ and collect $10^6$ samples using our software. We repeat each experiment $25$ times. We compare the covariance of the full traffic traces calculated offline (denoted trace) to the covariance extracted offline from a sampled process (denoted passive sampling) as well as from probes using \emph{H-probe} (denoted active probing). To this end, we estimate the Hurst parameter using a least square regression of the estimated covariance on lags $\tau \in[10^0,10^3]$. The lag range for the regression as well as the probing process parameters are chosen according to the constraints in Sect.~\ref{subsec:impact_finite_sample}. We show boxplots of the corresponding Hurst parameters in Fig.~\ref{fig:boxplot_1hop}. It is evident that \emph{H-probe} correctly estimates the configured Hurst parameters.

We exemplary deploy the packet pair dispersion method described in Sect.~\ref{sec:packet-pair-probes} in our Emulab testbed using the topology in Fig.\ref{fig:topology1}. We measure the send and receive times of the probes using an Endace DAG packet capture card attached to network taps at the outgoing resp. incoming ports at $S_p$ respectively $R_3$. The recorded time stamps $t_s,t_r \in \mathbb{R}$ have a $7.5$~ns capture precision. Tab.~\ref{tab:ppair_results} includes the mean of estimated $H$ over $25$ runs. The results indicate that capturing cross traffic intensities using packet pairs can be successfully used for Hurst parameter estimation.

\begin{table}[t]\normalsize
\centering
\caption{Mean Hurst parameter estimates using packet pair probes.}
\label{tab:ppair_results}
\begin{tabular}{c|c|c|c|c|}
\cline{2-5}
& \multicolumn{4}{|c|}{configured $H$} \\ \cline{2-5} \cline{2-5}
& 0.6 & 0.7 & 0.8 & 0.9  \\ \cline{1-5}
\multicolumn{1}{|c|}{estimated $H$} & \multirow{1}{*}{0.64}  & \multirow{1}{*}{0.71}  & \multirow{1}{*}{0.79}  & \multirow{1}{*}{0.87}   \\ \cline{1-5}
\end{tabular}
\end{table}

In another experiment we inject LRD traffic with differing $H$ along links $1$ and $2$ denoted $H_1$ and $H_2$ respectively. In Tab.~\ref{tab:H6H9_results} we show exemplary Hurst parameters obtained for all combinations of $H_1 = \{0.6, 0.9\}$ and $H_2 = \{0.6, 0.9\}$ as obtained from single packet probes. We note that our method correctly characterizes the dominant correlations, respectively, H along end-to-end paths from a probing rate of as low as 70~kbps.
\begin{table}[t]\normalsize
\caption{Exemplary Hurst parameter estimates in a $2$ node scenario using single packet probes.}
\label{tab:H6H9_results}
\begin{tabular}{c|c|c|c|c|c|}
\cline{2-6}
& \multicolumn{5}{|c|}{estimated $H$ on run $\#$} \\ \cline{2-6} \cline{2-6}
& 1 & 2 & 3 & 4 & 5  \\ \cline{1-6}
\multicolumn{1}{|c|}{\!\!\!$\{H_1=0.6,H_2=0.9\}$\!\!\!} & 0.87 & 0.89 & 0.89 & 0.90 & 0.90  \\ \cline{1-6}
\multicolumn{1}{|c|}{\!\!\!$\{H_1=0.9,H_2=0.6\}$\!\!\!} & 0.87 & 0.88 & 0.88 & 0.90 & 0.90   \\ \cline{1-6} \cline{1-6} \hline \hline
\multicolumn{1}{|c|}{\!\!\!$\{H_1=0.6,H_2=0.6\}$\!\!\!} & 0.59 & 0.62 & 0.64 & 0.63 & 0.63  \\ \cline{1-6}
\multicolumn{1}{|c|}{\!\!\!$\{H_1=0.9,H_2=0.9\}$\!\!\!} & 0.92 & 0.92 & 0.89 & 0.92  & 0.89  \\ \cline{1-6}
%\multicolumn{1}{|c|}{\!\!\!$\{H_1=0.6\}$\!\!\!} & 0.58 & 0.58 & 0.60 & 0.61 & 0.58  \\ \cline{1-6}
%\multicolumn{1}{|c|}{\!\!\!$\{H_1=0.9\}$\!\!\!} & 0.88 & 0.90 & 0.89 & 0.91  & 0.91  \\ \cline{1-6}
\end{tabular}
\end{table}
\subsection{Internet measurements}
\label{sec:Internet_measurements}
We perform measurements over multiple weeks using \mbox{\emph{H-probe}} from our lab in Germany targeting a number of worldwide PlanetLab nodes, in order to estimate the correlations on end-to-end paths across the Internet. The complex correlation structure along exemplary Internet paths is illustrated by the covariance plots in Fig.~\ref{fig:xcovsInternet}. First, we observe LRD covariance decay depicted in Fig. \ref{fig:xcov_upenn} and \ref{fig:xcov_virginia}. We point out that the correlation and hence the Hurst parameter vary significantly throughout the day. Moreover, we find that the correlation structure varies strongly across different paths. Additionally, for some targets we observed distinct periodicities on different timescales, as exemplified in Fig.~\ref{fig:xcov_umass}. Periodic behavior in offline Internet traces due to various protocol implementations has been previously reported, e.g., in \cite{broido:03}. Fig.~\ref{fig:upenn_H_estimates} depicts estimated $H$ using the aggregate variance method from Sect.~\ref{subsec:aggregate_variance}. The Hurst parameter estimates indicate a diurnal behavior. We provide additional data sets and results in the appendix \ref{sec:add_internet_results}.

\mbox{\emph{H-probe}} provides a new tool enabling researchers to shed light on the complex structure of traffic correlations without requiring the availability of traffic traces from Internet service providers.
\begin{figure}[htpb]
\includegraphics[type=pdf,ext=.pdf,read=.pdf,width=1.00\columnwidth]{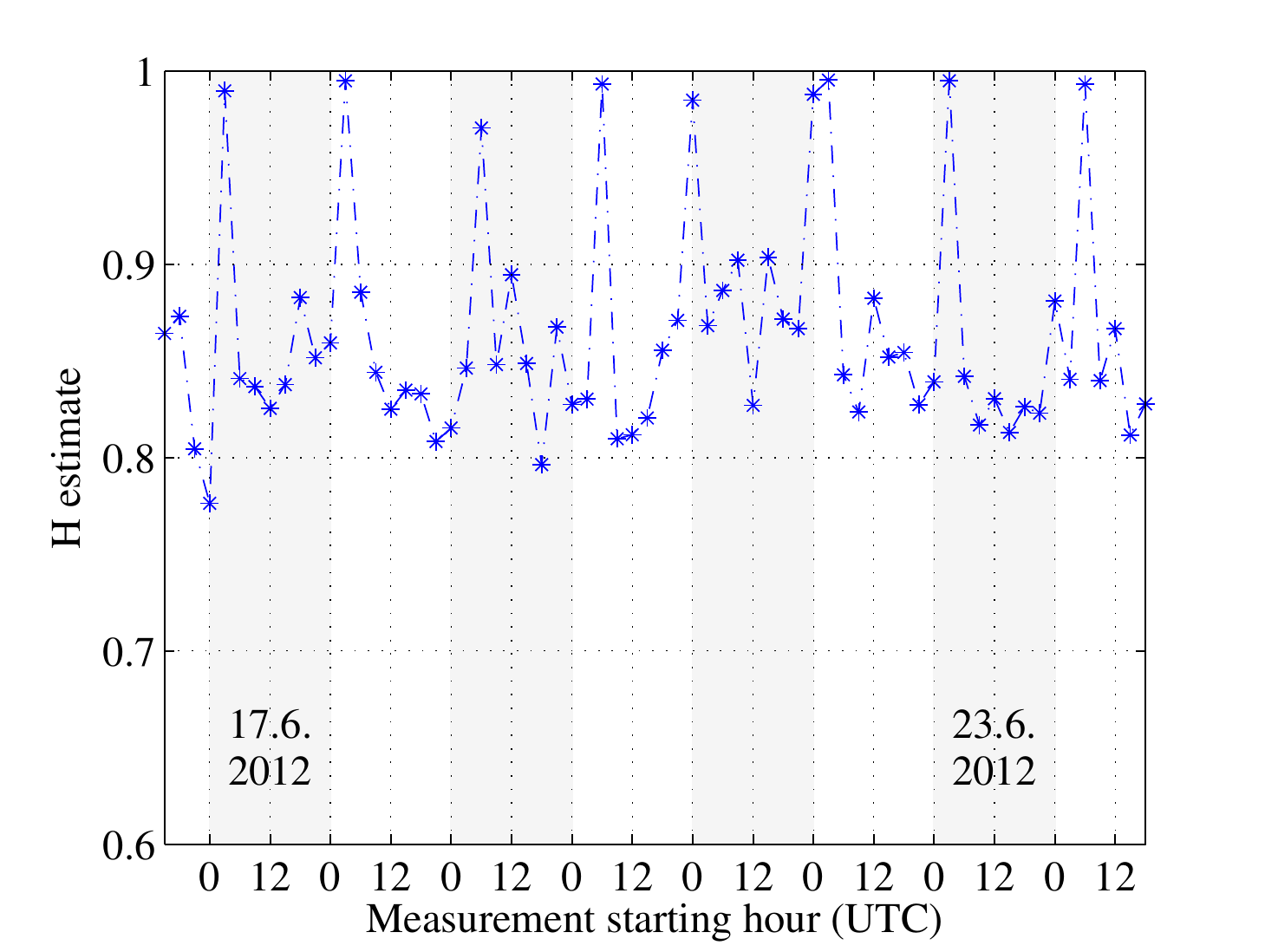}
   \caption{Hurst parameter estimates from continues measurements over one week. Target is planetlab1.cis.upenn.edu. $H$~ estimates obtained from the aggregate variance method.}
   \label{fig:upenn_H_estimates}
\end{figure}
\section{Conclusions}
\label{sec:conclusions}
In this paper, we derived estimators for the correlations of network traffic, given a limited set of traffic samples obtained by passive monitoring or active probing. We explored the impact of different sampling strategies on observed traffic correlations and quantified the impact of sampling on the observations. We showed that for finite sample sizes there are intrinsic limitations on the accuracy of the estimates and showed the influence of different sampling parameters. We found a non-linear tradeoff between sampling duration and sampling intensity. Further, we inferred the Hurst parameter $H$ from covariance estimates to quantify LRD. We developed and deployed an active probing method that estimates traffic correlations from end-to-end measurements without network support. The corresponding software is made publicly available. Finally, we presented measurement results from a controlled testbed environment as well as Internet paths. We observe a complex correlation structure on Internet paths. The correlation structure as well as $H$ significantly vary across time and paths. In addition to LRD we observe periodic behavior at different time scales.
\section{Appendix}
\balance
\subsection{Autocorrelation of sampling point processes:}
\label{subsec:Autocorrelation_sampling_processes}
We re\-phra\-se a basic result from~\cite{cox66} that is essential for the following derivations. Given a stationary stochastic process $A(\tau)$ in continuous time $\tau$ that takes values of either zero or one (Kronecker delta). The times between two Kronecker deltas are independent and identically distributed according to a density function $f(\tau)$. The autocorrelation density of $A(\tau)$ is known for lags $\tau > 0$ as~\cite{cox66} Eq. (4.6.1)
\begin{eqnarray}
\label{eq:autocorrelation_general}
\mathsf{E}\left[A(t)A(t+\tau)\right] = \mu_A \sum_{n=1}^{\infty}f^{(*n)}(\tau)
\end{eqnarray}
where $f^{(*n)}(\tau)$ is the $n$-fold self-convolution of $f(\tau)$. The intuition behind~\eqref{eq:autocorrelation_general} is that starting from one Kronecker delta at $A(t)$, another Kronecker delta at $A(t+\tau)$ can be the first, the second , the third \dots, delta to come after the one at $A(t)$. The derivation in \cite{cox66} uses a small time interval of length $\Delta t \rightarrow 0$, such that $\mu_A \Delta t$ is the probability that a Kronecker delta occurred in $[t,t+\Delta t)$ Eq. (4.5.9). To calculate the correlations of $A(t)$ \cite{cox66} deduces the conditional probability that a Kronecker delta occurred at $[t+\tau,t+\tau+\Delta t)$ given a Kronecker delta at $[t,t+\Delta t)$. This is given by $\sum_{n=1}^{\infty}f^{(*n)}(\tau) \Delta t$ Eq. (4.5.11). It follows that $(\Delta t)^2 \mu_A \sum_{n=1}^{\infty}f^{(*n)}(\tau)$ is the correlation of Kronecker deltas observed in time slots of length $\Delta t$.

A discrete-time extension of the correlation calculation from~\cite{cox66} with probability mass functions is straightforward. To this end we replace the probability density functions by probability mass functions and consider a time slot $\Delta t = 1$ such that we obtain correlation functions instead of densities.

For the continuous time distributions considered in this work we regard the correlations on fixed time slot basis. We use a discretization with a time slot of unit size.

\subsection{Geometric sampling:}
\label{subsec:Geometric_sampling}
Given a discrete time sampling process $A(\tau)$ with inter-sample times drawn from a geometric distribution given in Tab.~\ref{tab:sampling_distributions} with parameter $p$. The $n$-fold self-convolution of $f(\tau)$ is the probability mass function (pmf) of the sum of $n$ geometrically distributed random variables, i.e., negative binomial distributed with parameters $p$ and $n$ \cite{stirzaker01}. We insert the pmf from \cite{stirzaker01} for $f^{(*n)}(\tau)$ into \eqref{eq:autocorrelation_general} to find
\begin{eqnarray*}
\mathsf{E}\left[A(t)A(t+\tau)\right]\!\!\!& = &\!\!\!\mu_A \sum_{n=1}^{\infty}  \binom {\tau-1} {n-1} p^n (1-p)^{\tau-n}\\
& = &\!\!\!\mu_A \sum_{n=1}^{\tau}  \binom {\tau-1} {n-1} p^n (1-p)^{\tau-n}\\
& = &\!\!\!\mu_A \sum_{s=0}^{l} \binom {l} {s} p^{s} (1-p)^{l-s}p \\
& = &\!\!\!\mu_A p\\
& = &\!\!\!\mu_A^2.
%\label{eq:autocorr_sum_gamma}
\end{eqnarray*}
In the second line we used the support of the pmf $f^{(*n)}(\tau)$ to bound $1\le n\le\tau$. In the third line we substituted $\tau\!\!~-~\!\!1=l$ and $n-1=s$. In the fourth line we used the binomial identity $\sum_{s=0}^{l} \binom {l} {s} x^{s} y^{l-s}\!\!~=~\!\!(x+y)^{l}$. Finally, we know from the geometric distribution that $\mu_A=p$.

\subsection{Gamma sampling:}
\label{subsec:Gamma_sampling}
%\noindent \textbf{Gamma sampling:}
Given the time between two Kronecker deltas is Gamma distributed as in Tab.~\ref{tab:sampling_distributions} with parameters $\alpha$,~$\beta$. The sum in~\eqref{eq:autocorrelation_general} leads to the following expression
\begin{eqnarray}
\label{eq:autocorr_sum_gamma}
\sum_{n=1}^{\infty}f^{(*n)}(\tau) = \frac{e^{-\beta \tau}}{\tau} \sum_{n=1}^{\infty}\frac{(\beta \tau)^{n \alpha}}{\Gamma(n \alpha)} .
\end{eqnarray}

We derive analytical expressions for the cases $\alpha = 2$ and $\alpha = 4$ corresponding to Erlang(2) and Erlang(4) distributions of the inter-sample times. The mean rate of the sampling process is $\mu_A = \beta/\alpha$. We substitute $\alpha = 2$ into \eqref{eq:autocorr_sum_gamma} and evaluate the sum in \eqref{eq:autocorr_sum_gamma} as
\begin{eqnarray*}
\sum_{n=1}^{\infty}\frac{(\beta \tau)^{2n}}{(2n-1)!}=\beta \tau \sinh(\beta \tau)
\end{eqnarray*}
using the series expansion for $\sinh$ from \cite{abramo65} Eq.~(4.5.62). We then exploit the identity $\sinh(x)=(e^{x}-e^{-x})/2$ and that $\mu_A = \beta/2$ to evaluate \eqref{eq:autocorrelation_general} as
\begin{eqnarray*}
\mathsf{E}\left[A(t)A(t+\tau)\right] = \mu_A^2 (1-e^{-2\beta \tau})
\end{eqnarray*}

For $\alpha = 4$ we evaluate the sum in \eqref{eq:autocorr_sum_gamma} as
\begin{eqnarray*}
\sum_{n=1}^{\infty}\frac{(\beta \tau)^{4n}}{(4n-1)!}=\beta \tau \left(\frac{\sinh(\beta \tau) - \sin(\beta \tau)}{2}\right)
\end{eqnarray*}
using the series expansion for $\sinh$ and $\sin$ from \cite{abramo65} Eq.~(4.5.62) resp. Eq.~(4.3.65). We insert this result into \eqref{eq:autocorrelation_general} to find
\begin{eqnarray*}
\mathsf{E}\left[A(t)A(t+\tau)\right] = \mu_A^2 (1-e^{-2\beta \tau}-2\sin(\beta \tau) e^{-\beta \tau}).
\end{eqnarray*}
using the identity $\sinh(x)=(e^{x}-e^{-x})/2$ and that $\mu_A = \beta/4$.
\subsection{Uniform sampling:}
\label{subsec:Uniform_sampling}
%\noindent \textbf{Uniform sampling:}
Given the time between two Kronecker deltas is uniformly distributed with $f(\tau) = 1/b$ for $\tau \in [0,b]$ and zero otherwise. The mean rate of the sampling process is $\mu_A = 2/b$. In the following, we calculate the sum from~\eqref{eq:autocorrelation_general} for $\tau \in [0,b]$. We expand $\sum_{n=1}^{\infty}f^{(*n)}(\tau)
$ in~\eqref{eq:autocorrelation_general} as
\begin{eqnarray}
\label{eq:sum_n_convolution_uniform}
& &\!\!\!\!\!\!\!\!\!\sum_{n=1}^{\infty}f^{(*n)}(\tau) = f(\tau) + \int_{0}^{\tau} f(x_1) f(\tau-x_1) dx_1\nonumber \\
&&\!\!\!\!\!\!\!\!\! +\int_{0}^{\tau}\!\!\int_{0}^{x_1}\!\!f(\tau\!-\!x_1) f(x_1\!-\!x_2) f(x_2) dx_1 dx_2\nonumber \\
&&\!\!\!\!\!\!\!\!\! + \int_{0}^{\tau}\!\!\int_{0}^{x_1}\!\!\!\!\int_{0}^{x_2}\!\!f(\tau\!-\!x_1) f(x_1\!-\!x_2) f(x_2\!-\!x_3) f(x_3) dx_1 dx_2 dx_3\nonumber\\
&&\!\!\!\!\!\!\!\!\! + \cdots
\end{eqnarray}
Since all arguments $(\cdot)$ of the pdf $f$ in \eqref{eq:sum_n_convolution_uniform} are in the range $[0,\tau]$ we can replace all pdfs $f(\cdot)$ in \eqref{eq:sum_n_convolution_uniform} by $1/b$. Equation \eqref{eq:sum_n_convolution_uniform} evaluates then to the series expansion of the exponential function, i.e.,
\begin{eqnarray}
\label{eq:uniform_exponential_series}
\sum_{n=1}^{\infty}f^{(*n)}(\tau) & = &\frac{1}{b} + \frac{\tau}{b^2} + \frac{\tau^2}{2!b^3} + \frac{\tau^3}{3!b^4} + \cdots \nonumber\\
& = &\frac{1}{b} \sum_{n=0}^{\infty} \frac{(\tau/b)^n}{n!}\nonumber\\
& = &\frac{1}{b} e^{\tau/b}
\end{eqnarray}
Finally, we use \eqref{eq:uniform_exponential_series}, \eqref{eq:autocorrelation_general}, and that $\mu_A = 2/b$ to derive

\begin{eqnarray*}
\mathsf{E}\left[A(t)A(t+\tau)\right] = \frac{1}{2} \mu_A^2 e^{\tau/b}
\end{eqnarray*}
for $\tau \in [0,b]$. Since the process is mixing~\cite{baccelli:pasta}, we conclude for $\tau > b$ that the autocorrelation $\mathsf{E}\left[A(t)A(t+\tau)\right]$ converges quickly to $ \mu_A^2$.

\subsection{Distribution of the autocovariance of geometrically sampled iid Gaussian sequences:}
\label{subsec:distr_iid_gaussian_sampled_CI}
Given a sample path $w(t)$ of $W(t)$, that is described by \eqref{eq:process_sampling}, where $A(t)$ is a Bernoulli process and $Y(t)$ is a Gaussian iid process with mean $\mu_Y$ and variance $\sigma_Y^2$. The mean of the observations is known as $\mu_W = \mu_A \mu_Y$. The variance of $W(t)$ is given by $\sigma_W^2 = \sigma_A^2\mu_Y^2 +\sigma_Y^2\mu_A^2 +\sigma_A^2 \sigma_Y^2$ through independence of $A(t)$ and $Y(t)$. From the Bernoulli process $A(t)$ we know $\sigma_A^2 + \mu_A^2= \mu_A $ such that we can write $\sigma_W^2 = \sigma_A^2\mu_Y^2 +\sigma_Y^2\mu_A $. We consider a limited sample size $T$ such that $w(t)$ given for $t \in [1,T]$. An unbiased estimator of the autocovariance $\tilde{c}_W(\tau)$ is
\begin{eqnarray*}
\tilde{c}_W(\tau) = \frac{1}{T-\tau}\sum_{t=1}^{T-\tau}\left(w(t) - {\mu}_W\right)\left(w(t+\tau) - {\mu}_W\right) .
\end{eqnarray*}
After expansion of the product, for large $T-\tau$ we apply the central limit theorem to approximate the individual terms by normal random variables to find
\begin{eqnarray*}
\tilde{c}_W(\tau)\!\!\!\! & \approx &\!\!\!\!\mathcal{N}\left(\mu_W^2,\tfrac{\sigma_W^4+2\mu_W^2\sigma_W^2}{T-\tau}\right) - \mathcal{N}\left(2\mu_W^2,\tfrac{2\mu_W^2\sigma_W^2}{T-\tau}\right) + \mu_W^2\\
& = &\!\!\!\!\mathcal{N}\left(0,\tfrac{\sigma_W^4+ 4\mu_W^2\sigma_W^2}{T-\tau}\right).
\end{eqnarray*}
The confidence interval is directly obtained as $2$ times the standard deviation, i.e., $ \pm 2 \sqrt{\sigma_W^4+ 4\mu_W^2\sigma_W^2}/\sqrt{(T-\tau)}$. Finally, we insert $\sigma_W^2 = \sigma_A^2\mu_Y^2 +\sigma_Y^2\mu_A $ and assume $T \gg \tau$ to find the confidence interval $\pm 2 \sqrt{(\sigma_A^2\mu_Y^2 +\sigma_Y^2\mu_A )^2+ 4\mu_A^2\mu_Y^2(\sigma_A^2\mu_Y^2 +\sigma_Y^2\mu_A )}/\sqrt{T}$.
\subsection{Distribution of the autocovariance of the geometric sampling process:}
\label{subsec:distr_poisson_sampling_CI}
Given a sample path $a(t)$ of a  Bernoulli sampling process $A(t)$ with a limited sample size $T$ such that $A(t)$ is given for $t \in [1,T]$. The Bernoulli sampling process has geometrically distributed inter-sample times as in Tab.~\ref{tab:sampling_distributions}. Given the mean $\mu_A$ is known. An unbiased estimator of the autocovariance $\tilde{c}_A(\tau)$ is
\begin{eqnarray*}
\tilde{c}_A(\tau) = \frac{1}{T-\tau}\sum_{t=1}^{T-\tau}\left(a(t) - {\mu}_A\right)\left(a(t+\tau) - {\mu}_A\right) .
\end{eqnarray*}
After expansion of the product, for large $T-\tau$ we apply the central limit theorem to approximate the individual terms by normal random variables to find
\begin{eqnarray*}
\tilde{c}_A(\tau) \!\!\!\! & \approx &\!\!\!\! \mathcal{N}\left(\mu_A^2,\tfrac{\sigma_A^4+2\mu_A^2\sigma_A^2}{T-\tau}\right) - \mathcal{N}\left(2\mu_A^2,\tfrac{2\mu_A^2\sigma_A^2}{T-\tau}\right) + \mu_A^2\\
\!\!\!\! & = &\!\!\!\!\mathcal{N}\left(0,\tfrac{\sigma_A^4+ 4\mu_A^2\sigma_A^2}{T-\tau}\right).
\end{eqnarray*}
The confidence interval is directly obtained as $2$ times the standard deviation, i.e., $\pm 2 \sigma_A\sqrt{\sigma_A^2+4\mu_A^2}/\sqrt{T - \tau}$.
\subsection{Bias of the autocovariance estimator for $Y(t)$:}
\label{subsec:bias_autocov_y}
We derive the bias of the covariance estimator~\eqref{eq:cov_estimator} if applied to a sample path $y(t)$ of the LRD process $Y(t)$ and show that it is asymptotically unbiased as the sample duration tends to infinity $T \rightarrow \infty$. Given a sample path $y(t)$ with sample mean $\tilde{\mu}_{Y_0} = \tfrac{1}{(T-\tau)} \sum_{t=1}^{T-\tau} y(t)$ and $\tilde{\mu}_{Y_\tau} = \tfrac{1}{(T-\tau)} \sum_{t=1}^{T-\tau} y(t+\tau)$ a covariance estimator is
\begin{eqnarray}
\label{eq:cov_estimator}
\tilde{c}_Y(\tau) = \frac{1}{T-\tau}\sum_{t=1}^{T-\tau}\left(y(t) - \tilde{\mu}_{Y_0} \right) \left(y(t+\tau) - \tilde{\mu}_{Y_\tau} \right)
\end{eqnarray}
To estimate the bias, we derive the expected value $\mathsf{E}[\tilde{c}_Y(\tau)]$. To this end, we expand the product and compute the expected values of the individual terms $y(t)y(t+\tau)\!-\!y(t)\tilde{\mu}_{Y_\tau}-y(t\!~+~\!\tau)\tilde{\mu}_{Y_0} + \tilde{\mu}_{Y_0}\tilde{\mu}_{Y_\tau}$ as
\begin{eqnarray*}
\mathsf{E} \left[ \frac{1}{T-\tau} \sum_{t=1}^{T-\tau} y(t)y(t+\tau) \right] = c_Y(\tau) + \mu_Y^2 ,
\end{eqnarray*}
where $c_Y(\tau)$ and $\mu_Y$ are the population parameters, and
\begin{eqnarray*}
\mathsf{E} \left[ \frac{1}{T-\tau} \sum_{t=1}^{T-\tau} y(t)\tilde{\mu}_{Y_\tau} \right] = \mathsf{E}[\tilde{\mu}_{Y_0}\tilde{\mu}_{Y_\tau}] ,
\end{eqnarray*}
where we used that $\tfrac{1}{T-\tau} \sum_{t=1}^{T-\tau} y(t) = \tilde{\mu}_{Y_0}$. The same argument applies for the product $y(t+\tau)\tilde{\mu}_{Y_0}$. We estimate
\begin{eqnarray*}
\mathsf{E}\left[\tilde{\mu}_{Y_0}\tilde{\mu}_{Y_\tau}\right] & = & \text{Cov}\left[\tilde{\mu}_{Y_0}, \tilde{\mu}_{Y_\tau}\right] + \mathsf{E}\left[\tilde{\mu}_{Y_0}\right]\mathsf{E}\left[\tilde{\mu}_{Y_\tau}\right] \\ & \le & \text{Var}\left[ \tilde{\mu}_{Y_0} \right] + \mu_Y^2,
\end{eqnarray*}
with $\mathsf{E}\left[\tilde{\mu}_{Y_0}\right] = \mu_Y$, respectively, $\mathsf{E}\left[\tilde{\mu}_{Y_\tau}\right] = \mu_Y$ that are unbiased estimators of the population mean. Note that the samples that form $\tilde{\mu}_{Y_0}$ and $\tilde{\mu}_{Y_\tau}$ overlap by $T-2\tau$, such that for $\tau \ll T$ we have $\text{Cov}\left[\tilde{\mu}_{Y_0}, \tilde{\mu}_{Y_\tau}\right] \approx \text{Var}\left[ \tilde{\mu}_{Y_0} \right]$. Finally, we use that the variance of the mean of $T-\tau$ samples of an LRD process decays as $\sigma_Y^2/(T-\tau)^{2-2H}$ to derive
\begin{eqnarray*}
\mathsf{E}\left[\tilde{\mu}_{Y_0}\tilde{\mu}_{Y_\tau}\right] \approx \frac{\sigma_Y^2}{(T-\tau)^{2-2H}} + \mu_Y^2 .
\end{eqnarray*}
Putting all pieces together we obtain
\begin{eqnarray}
\label{eq:biascy} \mathsf{E}\left[\tilde{c}_Y(\tau)\right] \!\!\!& = &\!\!\! c_Y(\tau) - \text{Cov}[\tilde{\mu}_{Y_0},\tilde{\mu}_{Y_\tau}] \\
\!\!\!& \approx &\!\!\! c_Y(\tau) - \frac{\sigma_Y^2}{(T-\tau)^{2 -2H}},\nonumber
\end{eqnarray}
i.e., the estimator underestimates the covariance, where the bias diminishes if $T-\tau$ is large. We note, that the bias cannot be easily eliminated if the prefactor $1/(T-\tau-1)$ is used instead of $1/(T-\tau)$ in~\eqref{eq:cov_estimator}, as it is typically done if the covariance is estimated using the sample mean.

\subsection{Bias of the autocovariance estimator for $W(t)$:}
\label{subsec:bias_autocov_w}
%\noindent \textbf{Bias of the autocovariance estimator for $W(t)$:}
%
We derive the bias of the covariance estimator~\eqref{eq:cov_estimator} if applied to the observed process $W(t)$. We show that the estimator is asymptotically unbiased for large sample durations $T\rightarrow \infty$. Given a sample path $w(t)$ with sample mean $\tilde{\mu}_{W_0}$, respectively, $\tilde{\mu}_{W_\tau}$ defined in Sect.~\ref{subsec:bias_autocov_y}. We use the covariance estimator from~\eqref{eq:cov_estimator}. To estimate the bias we derive $\mathsf{E}[\tilde{c}_W(\tau)]$ from~\eqref{eq:biascy} as
\begin{eqnarray*}
\mathsf{E}\left[\tilde{c}_W(\tau)\right] = c_W(\tau) - \text{Cov}[\tilde{\mu}_{W_0},\tilde{\mu}_{W_\tau}],
\end{eqnarray*}
where $c_W(\tau)$ is the population parameter. As before we estimate $\text{Cov}[\tilde{\mu}_{W_0},\tilde{\mu}_{W_\tau}] \approx \text{Var}[\tilde{\mu}_{W_0}]$ and express $\tilde{\mu}_{W_0}$ as a sum to compute the variance $\text{Var}\left[\tilde{\mu}_{W_0}\right]$ as
\begin{eqnarray*}
\text{Var}\left[\tilde{\mu}_{W_0}\right] = \frac{1}{(T-\tau)^2} \sum_{i=1}^{T-\tau} \sum_{j=1}^{T-\tau} \text{Cov}[w(i),w(j)],
\end{eqnarray*}
where we use the identity
\begin{eqnarray}
\label{eq:identity_variance_of_sum}
\text{Var}\left[\sum_{i=1}^{n}X_i\right]=\sum_{i = 1}^{n}\sum_{j = 1}^{n}\text{Cov}\left[X_i,X_j\right]
\end{eqnarray}
for random variables $X_i$, $i \in [1,n]$. Rearranging the statement above yields
\begin{eqnarray*}
\text{Var}\left[\tilde{\mu}_{W_0}\right] = \frac{c_W(0)}{T-\tau} + \frac{2}{(T-\tau)^2} \sum_{t=1}^{T-\tau-1} (T-\tau-t) c_W(t),
\end{eqnarray*}
where we used the notation $\text{Cov}[w(t),w(t+\tau)] = c_W(\tau)$. The expected value of the sample covariance follows as
\begin{eqnarray}
\mathsf{E}[\tilde{c}_W(\tau)] \!\!\!& \approx & c_W(\tau) -\frac{c_W(0)}{T-\tau} \nonumber\\ \label{eq:biascw}
 & - & \frac{2}{(T-\tau)^2} \sum_{t=1}^{T-\tau-1} (T-\tau-t) c_W(t).
\end{eqnarray}

\subsection{Aggregate variance of a sampled process:}
\label{subsec:aggvar_of_sampled_process}
%\noindent \textbf{Aggregate variance of a sampled process:}
\begin{proof}[Proof of Lem.~\ref{lem:agg_var}]
The aggregated version of the process $W(t)$ on the aggregation level $M$ is defined for $k \in \mathbb{N}$ as
\begin{eqnarray*}
W^{(M)}(k) = \frac{1}{M} \sum_{t = 1+(k-1)M}^{kM} W(t) ,
\end{eqnarray*}
where $M$ is the block size that is used for averaging. The variance of $W^{(M)}$ is obtained using the identity \eqref{eq:identity_variance_of_sum} as
\begin{eqnarray*}
\text{Var}\left(W^{(M)}\right) = \frac{1}{M^2} \sum_{i=1}^M\sum_{j=1}^M \text{Cov}(W(i),W(j)).
\end{eqnarray*}
Using the notation $c_W(\tau) = \text{Cov}(W(i),W(i+\tau))$ we rearrange the previous statement as
\begin{eqnarray}
\label{eq:agg_var_var}
\text{Var}\left(W^{(M)}\right) = \frac{c_W(0)}{M} + \frac{2}{M^2} \sum_{\tau=1}^{M-1} (M-\tau) c_W(\tau).
\end{eqnarray}
The same expression can be formulated for $\text{Var}(A^{(M)})$ and $\text{Var}(Y^{(M)})$ by substituting $c_A(\tau)$ resp. $c_Y(\tau)$ for $c_W(\tau)$ in \eqref{eq:agg_var_var}. Next, we insert $c_W(\tau)$ from Lem.~\ref{lem:sample_covariance} into~\eqref{eq:agg_var_var} to relate $\text{Var}(W^{(M)})$ to $\text{Var}(A^{(M)})$, $\text{Var}(Y^{(M)})$, $c_A(\tau)$ and $c_Y(\tau)$. We obtain
\begin{multline*}
\text{Var}\left(W^{(M)}\right) = \frac{1}{M}c_Y(0)(c_A(0) + \mu_A^2) + \frac{1}{M}c_A(0) \mu_Y^2 \\ + \frac{2}{M^2} \sum_{\tau=1}^{M-1} (M-\tau)( c_Y(\tau)(c_A(\tau) + \mu_A^2) + c_A(\tau) \mu_Y^2).
\end{multline*}
After some reordering we arrive at
\begin{multline*}
\text{Var}\left(W^{(M)}\right) = \mu_Y^2 \frac{c_A(0)}{M} + \mu_Y^2 \frac{2}{M^2} \sum_{\tau=1}^{M-1} (M-\tau)  c_A(\tau) \\ + \mu_A^2 \frac{c_Y(0)}{M} + \mu_A^2 \frac{2}{M^2} \sum_{\tau=1}^{M-1} (M-\tau) c_Y(\tau) \\ + \frac{1}{M}c_Y(0)c_A(0) + \frac{2}{M^2} \sum_{\tau=1}^{M-1} (M-\tau) c_Y(\tau)c_A(\tau)
\end{multline*}
and by application of~\eqref{eq:agg_var_var} we obtain
\begin{multline*}
\text{Var}\left(W^{(M)}\right) = \mu_Y^2 \text{Var}\left(A^{(M)}\right) + \mu_A^2 \text{Var}\left(Y^{(M)}\right) \\ + \frac{1}{M}c_Y(0)c_A(0) + \frac{2}{M^2} \sum_{\tau=1}^{M-1} (M-\tau) c_Y(\tau)c_A(\tau) .
\end{multline*}
\end{proof}
\section{Data sets from Internet measurements}
\label{sec:add_internet_results}
We perform measurements over multiple weeks using \mbox{\emph{H-probe}} from our lab in Germany targeting a number of worldwide PlanetLab nodes. Next, we describe the measurement setup:
\begin{itemize}
  \item Discretization slot length $\delta = 1$ ms.
  \item Geometrically distributed inter-sample times with $p\!\!~=~\!\!0.1$.
  \item Number of probes collected $10^6$ ($\sim 3$ hours)
  \item ICMP probing packets of size $64$ Byte
  \item Probing rate $100$ pkt/s $\sim 70$ kbps ($24$ Byte layer $2$ overhead)
\end{itemize}

In the following we present a representative set of the measurement results, where the target is planetlab1.cis.upenn.edu. We show exemplary end-to-end covariance as well aggregate variance estimates at two different days. In addition, we show estimates of measurements starting at 10:45 UTC from 17-24.7.2012.

For the autocovariance as well as the aggregate variance method the slope of the curve is given by $2H-2$. Slope estimates are obtained through least square regression. The $H$ estimates from Internet measurements have a moderately higher variance compared to active probing results from Fig.~\ref{fig:boxplot_1hop}. Further, the $H$ estimates in Fig.~\ref{fig:upenn_H_estimates} show diurnal behavior.

\clearpage

\hspace{-20pt}
%\centering
\begin{minipage}{\linewidth}
\includegraphics[type=pdf,ext=.pdf,read=.pdf,width=1.00\columnwidth]{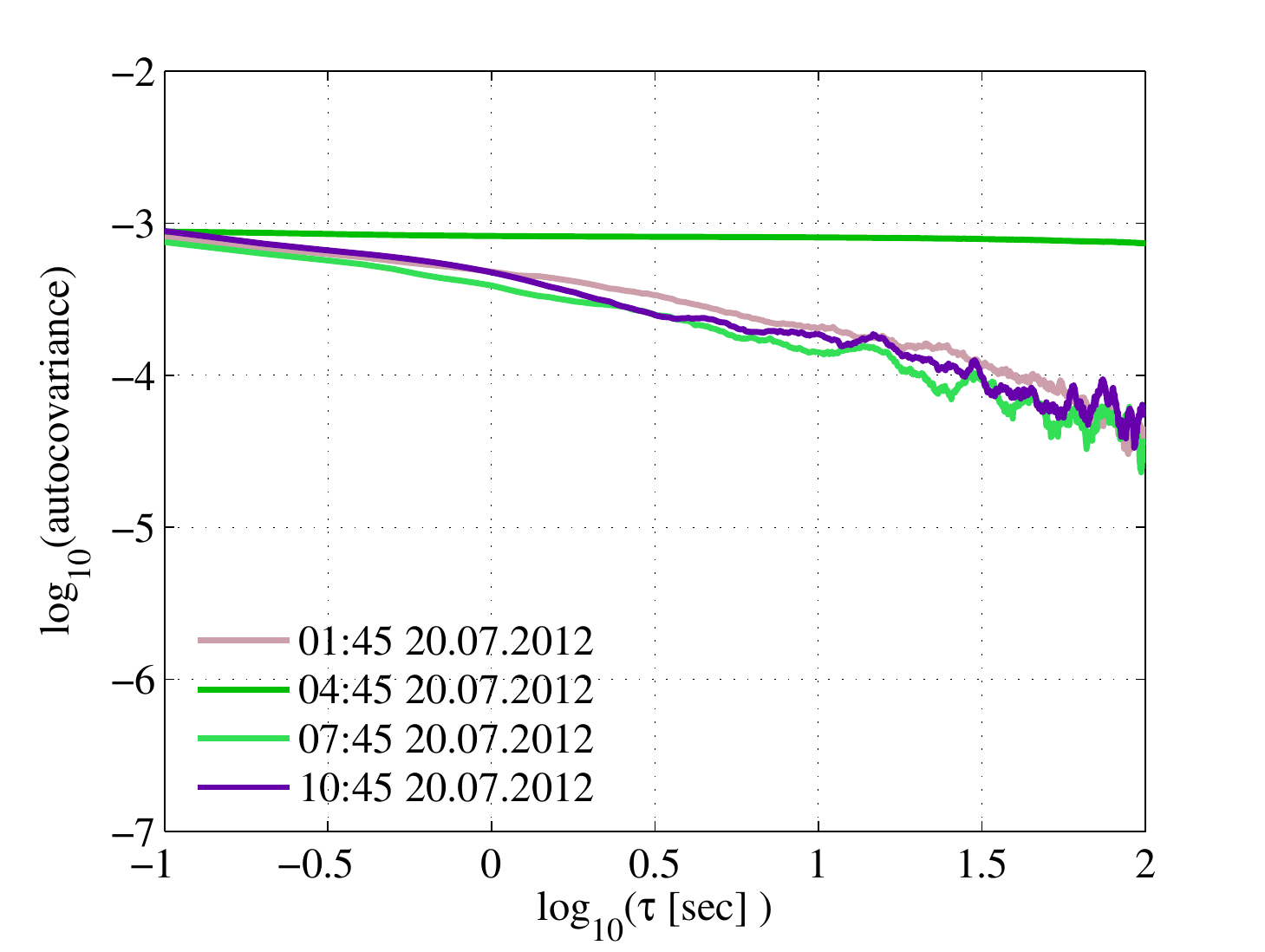}
\captionof{figure}{Covariance estimates 20.7.2012 UTC}
\end{minipage}
\hspace{-20pt}
\begin{minipage}{\linewidth}
\includegraphics[type=pdf,ext=.pdf,read=.pdf,width=1.00\columnwidth]{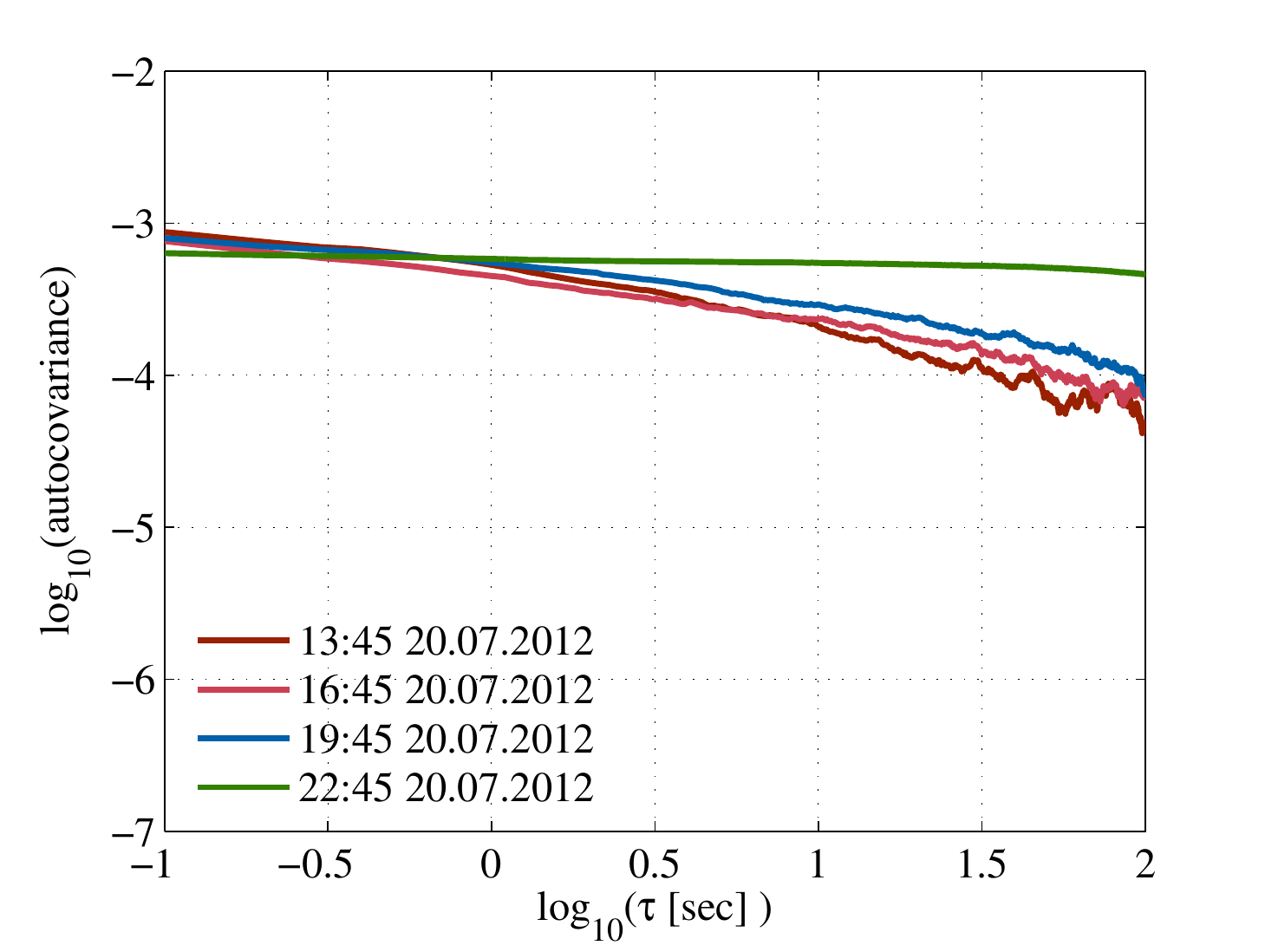}
\captionof{figure}{Covariance estimates 20.7.2012 UTC}
\end{minipage}
\begin{minipage}{\linewidth}
\includegraphics[type=pdf,ext=.pdf,read=.pdf,width=1.00\columnwidth]{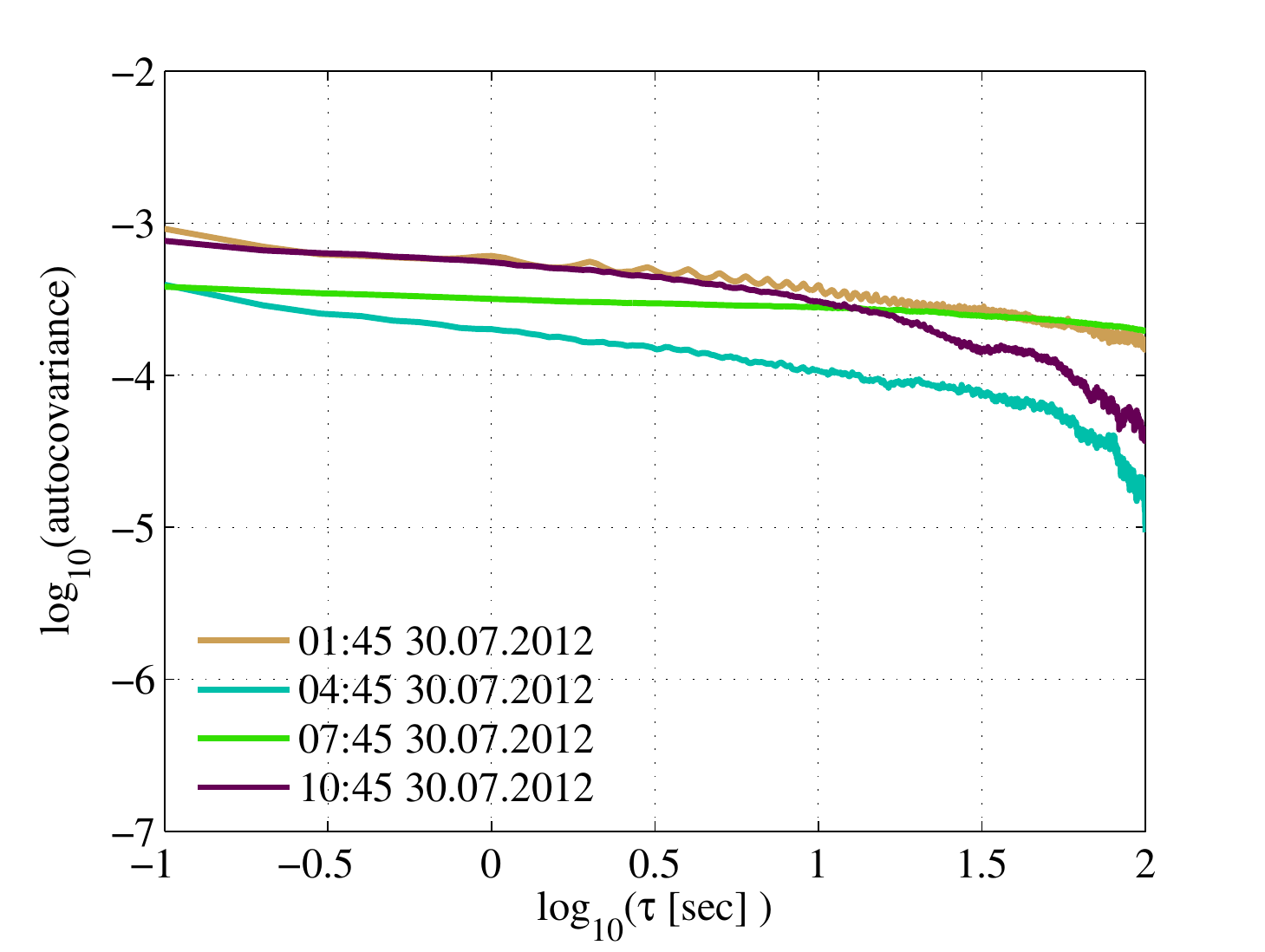}
\captionof{figure}{Covariance estimates 30.7.2012 UTC}
\end{minipage}
\begin{minipage}{\linewidth}
\includegraphics[type=pdf,ext=.pdf,read=.pdf,width=1.00\columnwidth]{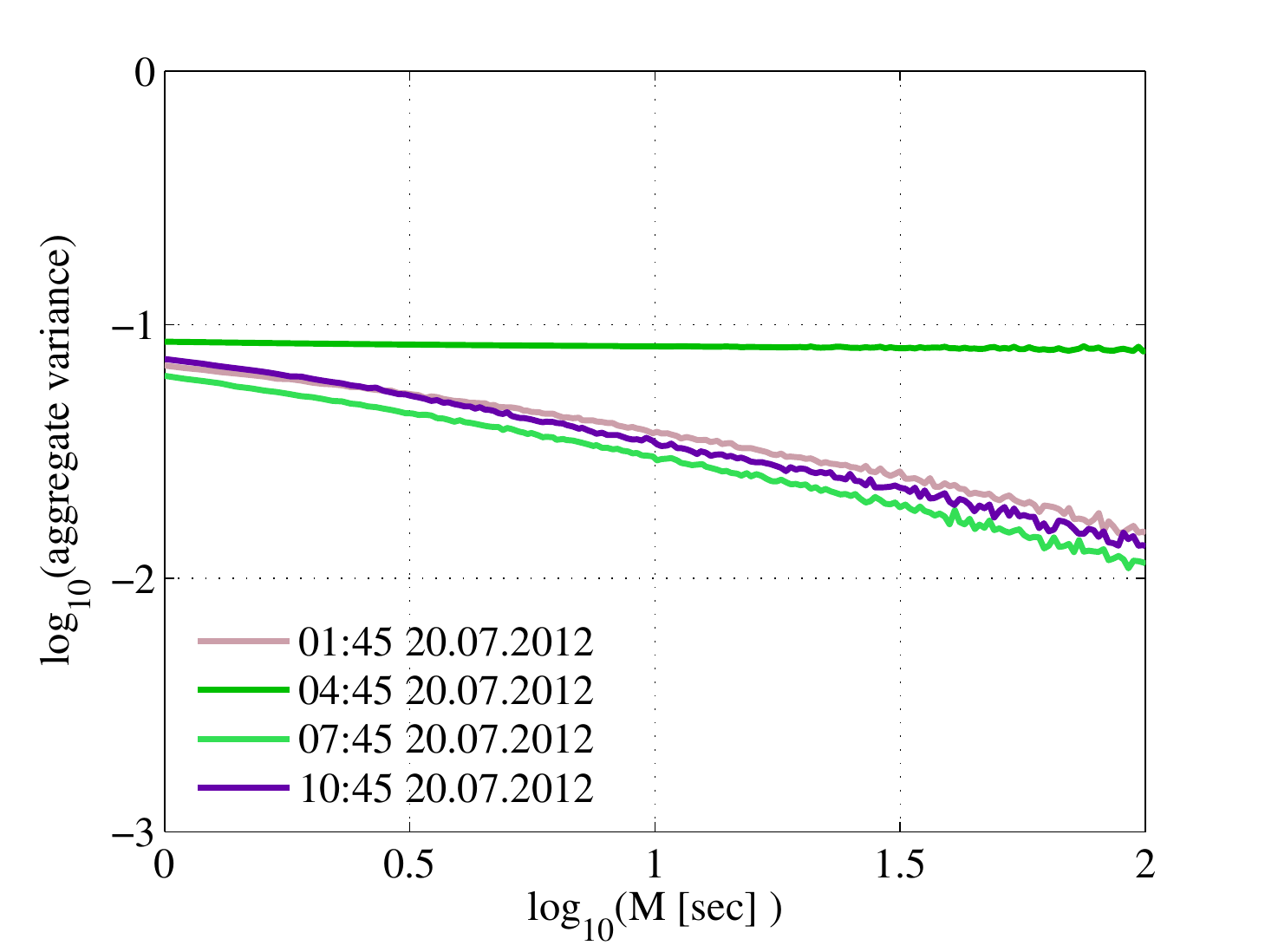}
\captionof{figure}{Aggregate variance estimates 20.7.2012 UTC}
\end{minipage}
\begin{minipage}{\linewidth}
\includegraphics[type=pdf,ext=.pdf,read=.pdf,width=1.00\columnwidth]{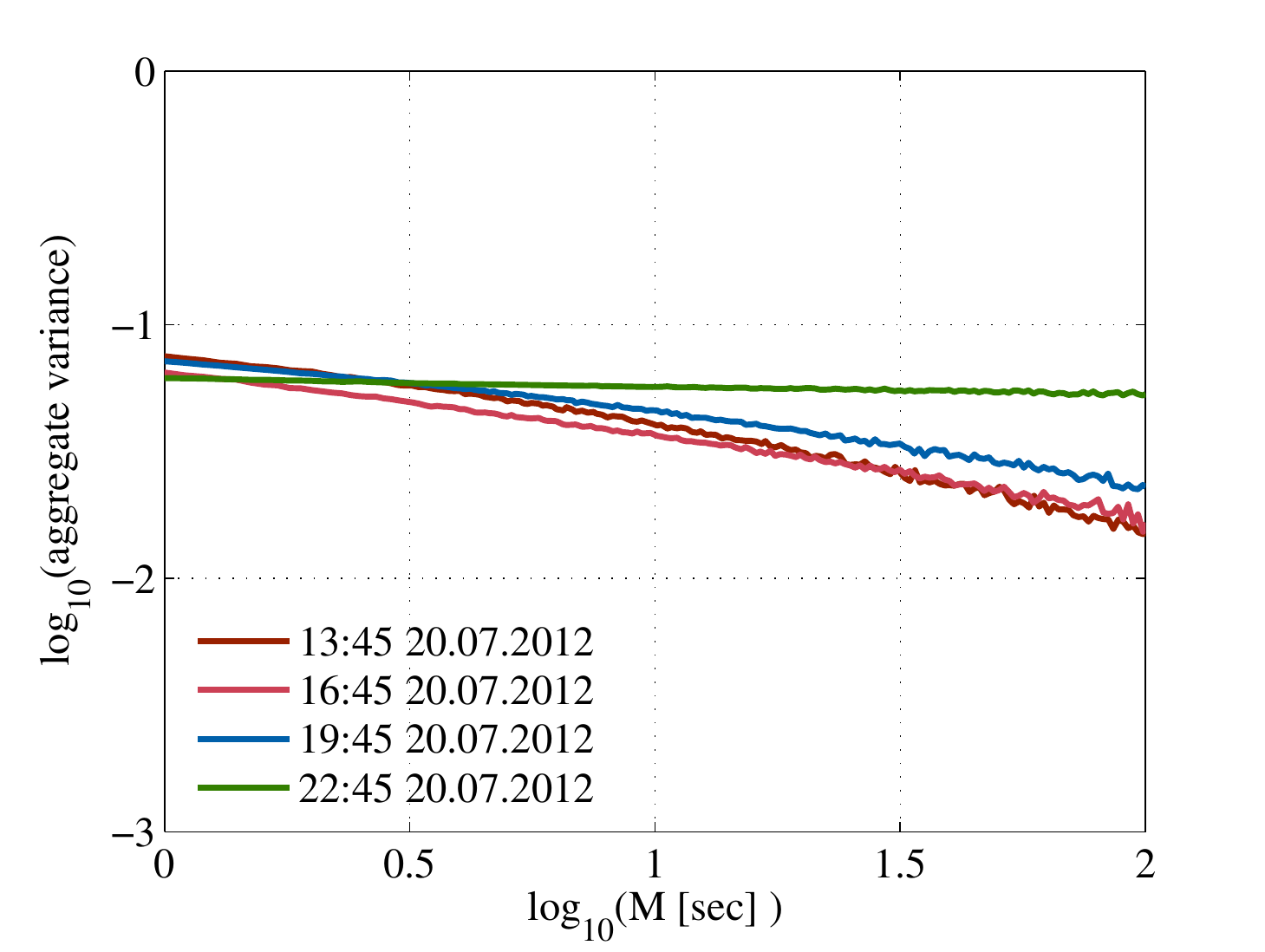}
\captionof{figure}{Aggregate variance estimates  20.7.2012 UTC}
\end{minipage}
\begin{minipage}{\linewidth}
\includegraphics[type=pdf,ext=.pdf,read=.pdf,width=1.00\columnwidth]{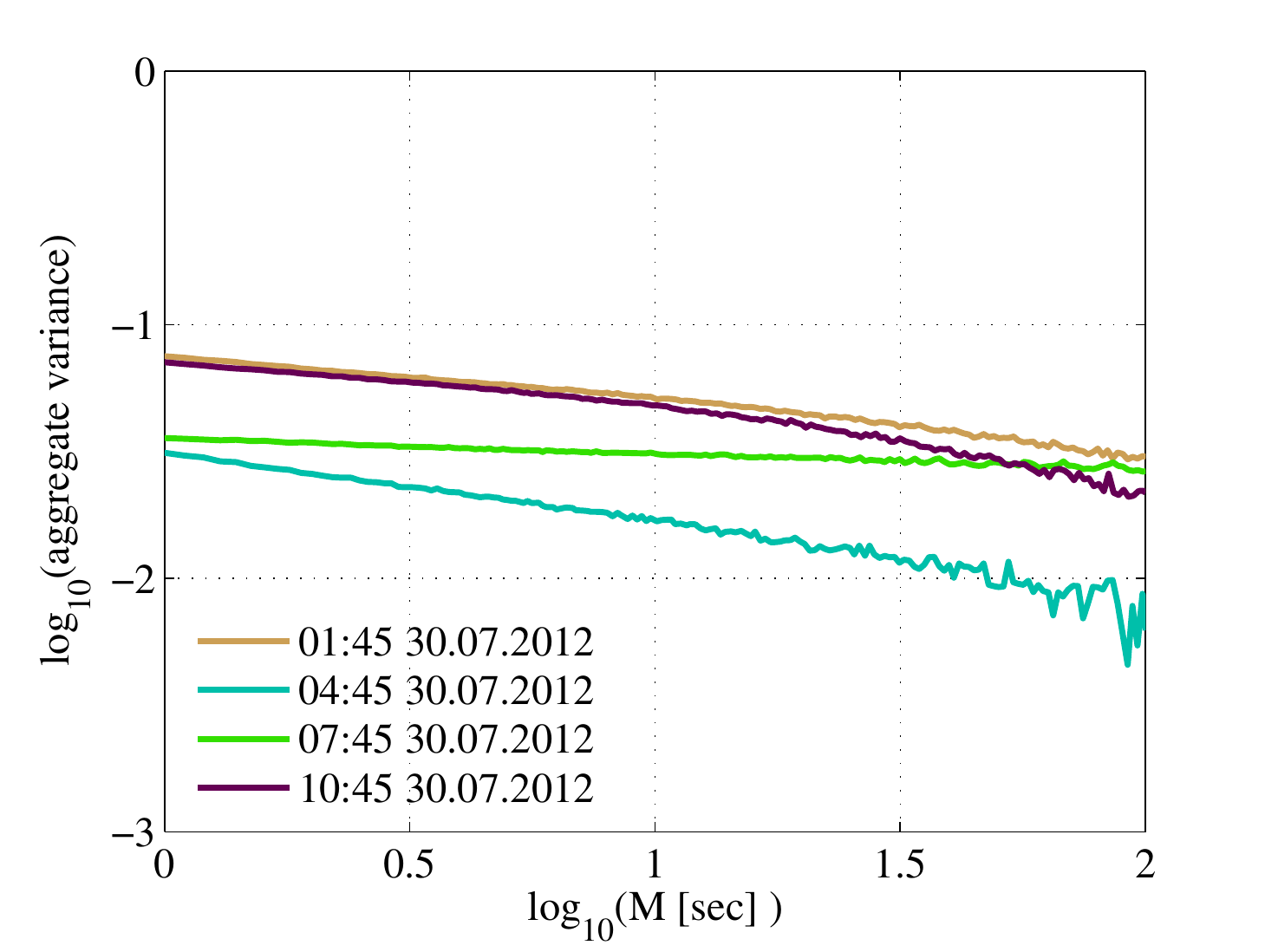}
\captionof{figure}{Aggregate variance estimates  30.7.2012 UTC}
\end{minipage}
\begin{minipage}{\linewidth}
\includegraphics[type=pdf,ext=.pdf,read=.pdf,width=1.00\columnwidth]{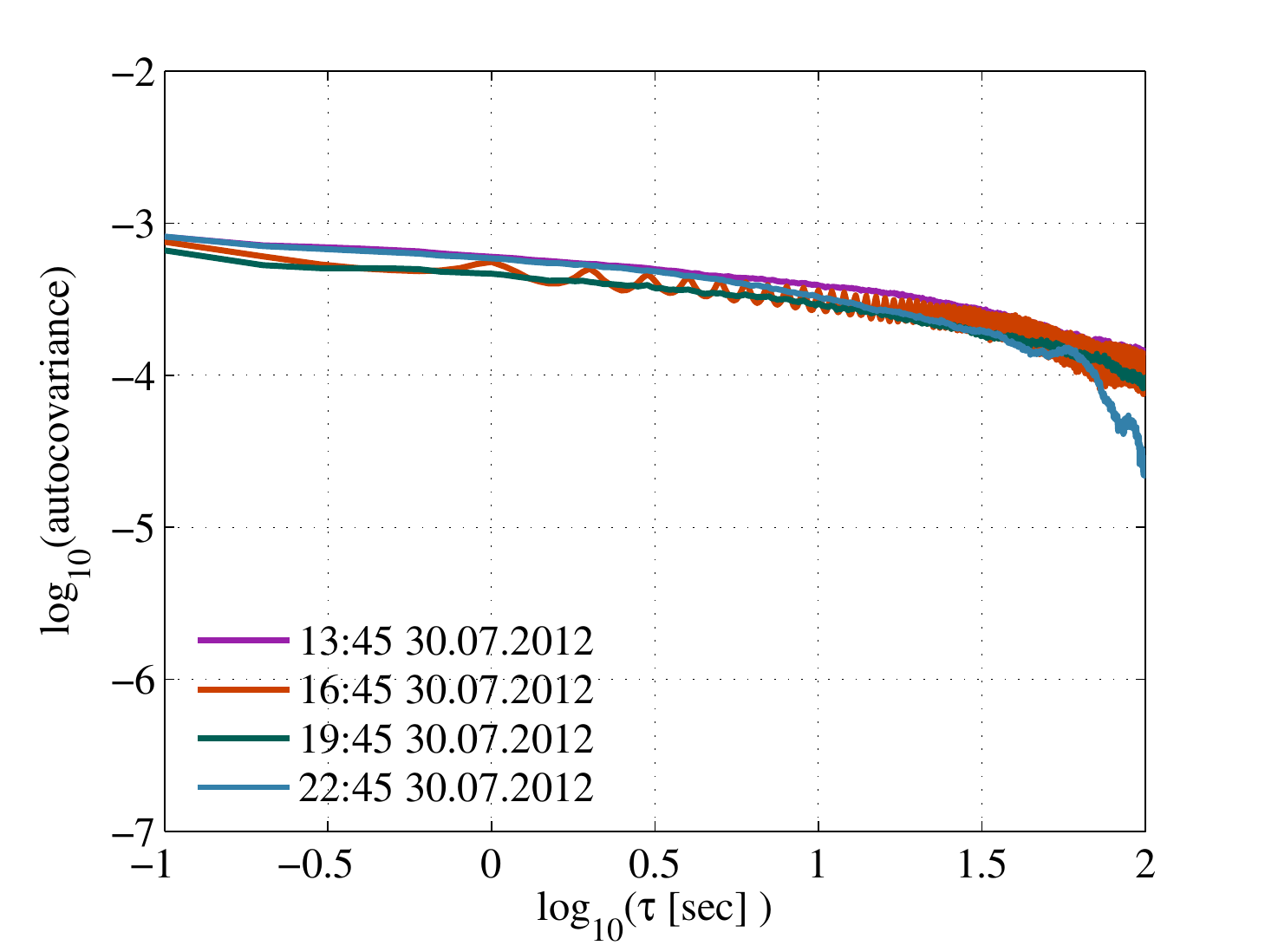}
\captionof{figure}{Covariance estimates 30.7.2012 UTC}
\end{minipage}
\begin{minipage}{\linewidth}
\includegraphics[type=pdf,ext=.pdf,read=.pdf,width=1.00\columnwidth]{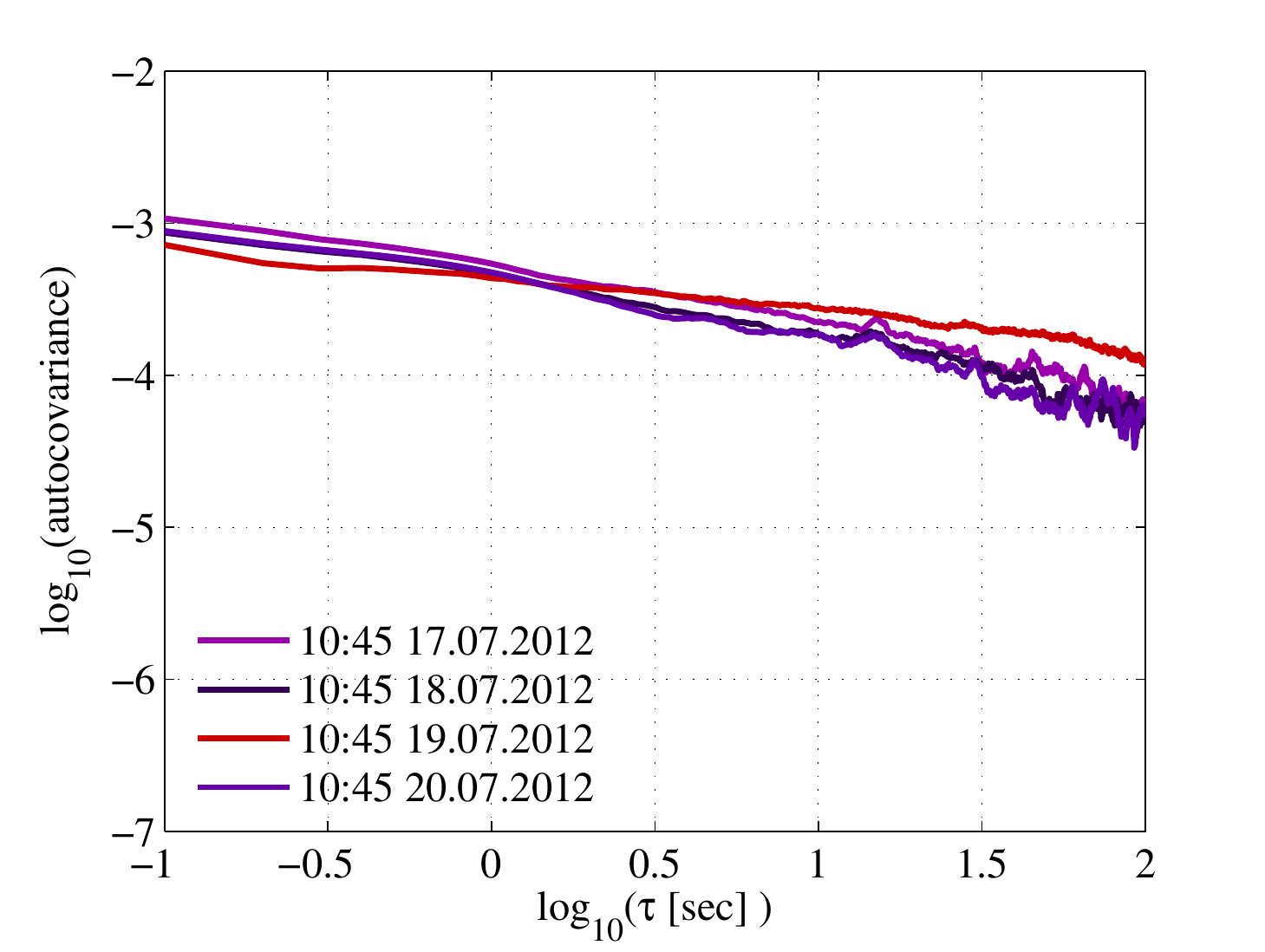}
\captionof{figure}{Covariance estimates at 10:45 UTC}
\end{minipage}
\begin{minipage}{\linewidth}
\includegraphics[type=pdf,ext=.pdf,read=.pdf,width=1.00\columnwidth]{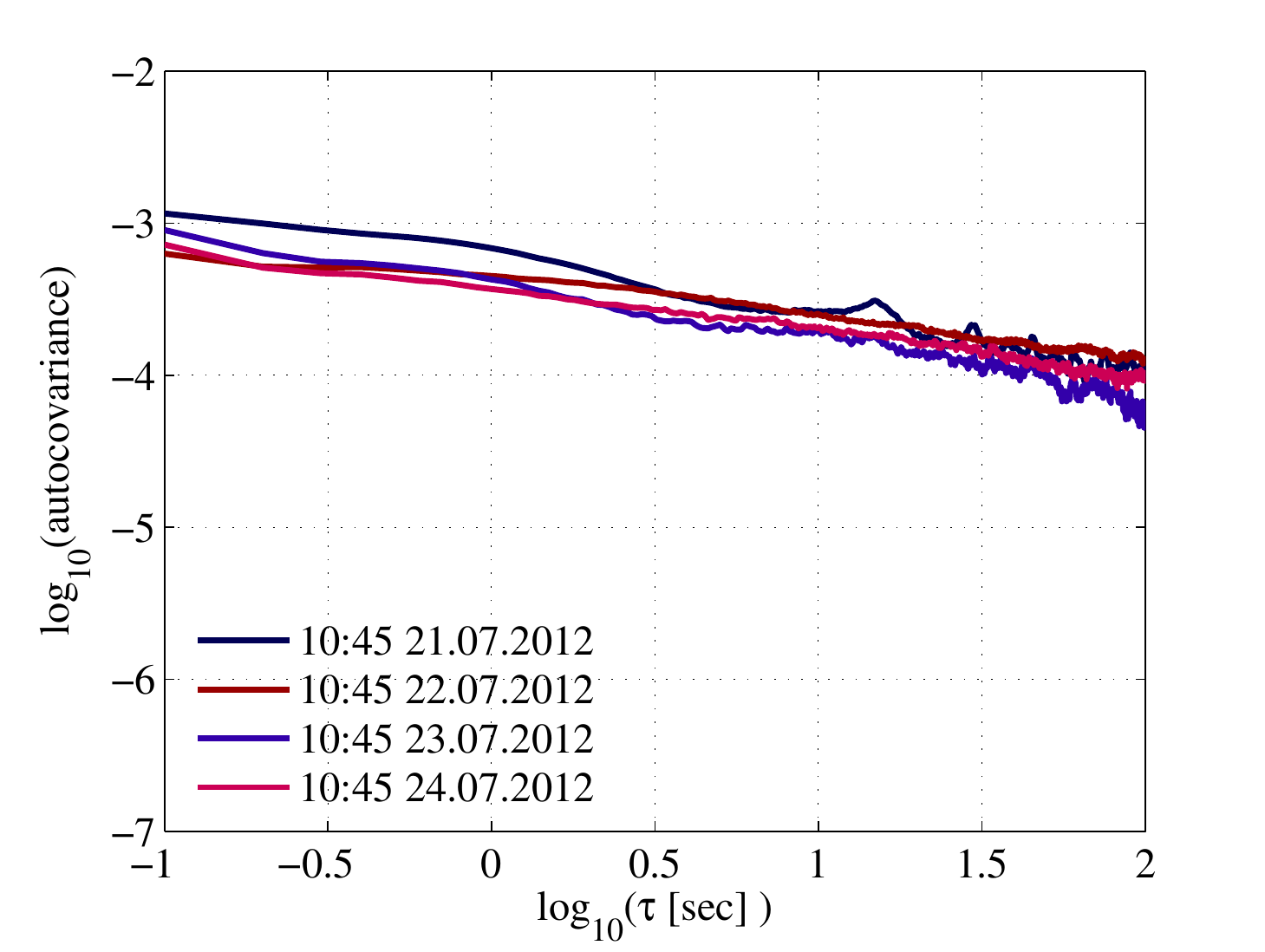}
\captionof{figure}{Covariance estimates at 10:45 UTC}
\end{minipage}
\begin{minipage}{\linewidth}
\includegraphics[type=pdf,ext=.pdf,read=.pdf,width=1.00\columnwidth]{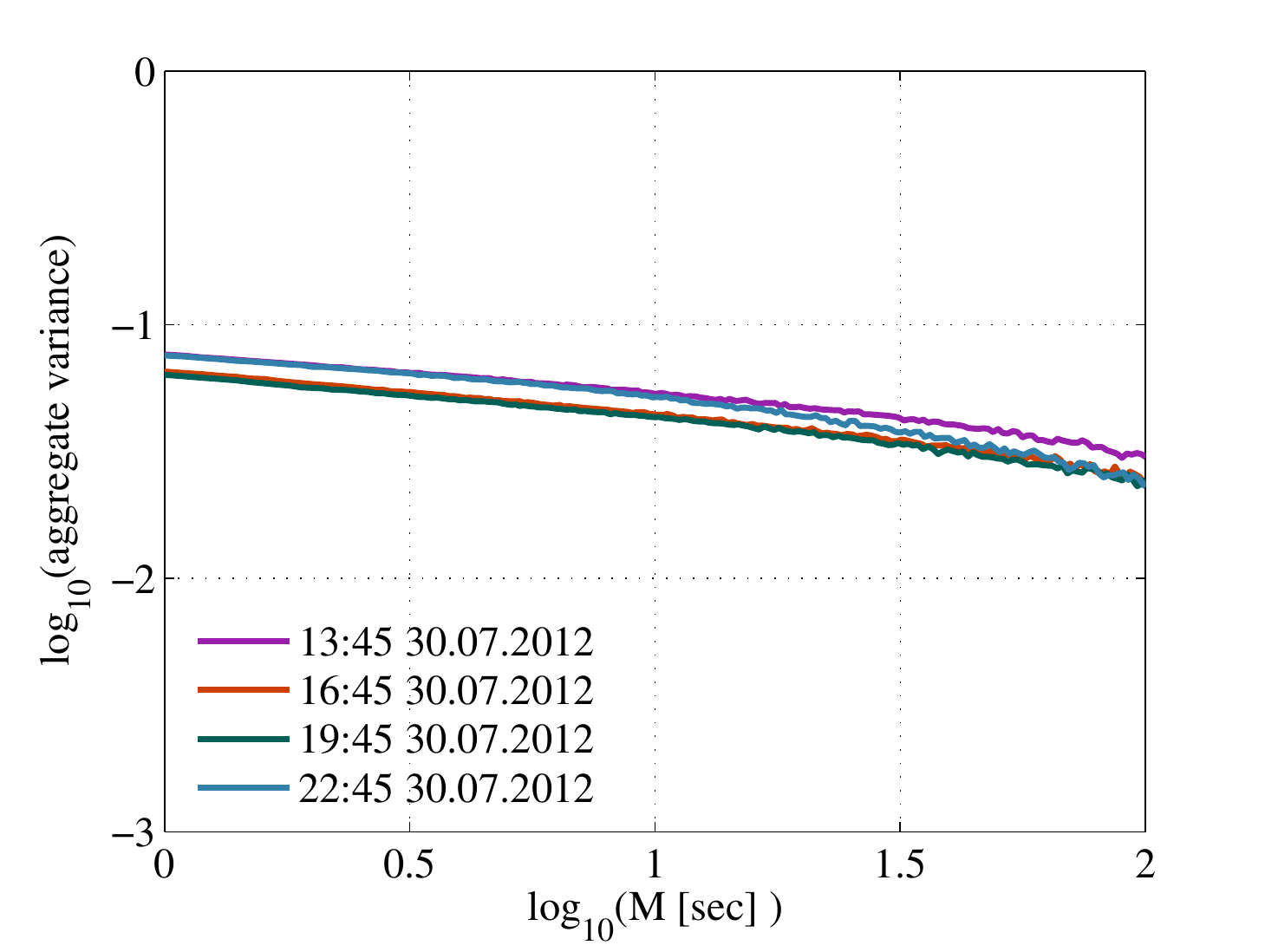}
\captionof{figure}{Aggregate variance estimates  30.7.2012 UTC}
\end{minipage}
\begin{minipage}{\linewidth}
\includegraphics[type=pdf,ext=.pdf,read=.pdf,width=1.00\columnwidth]{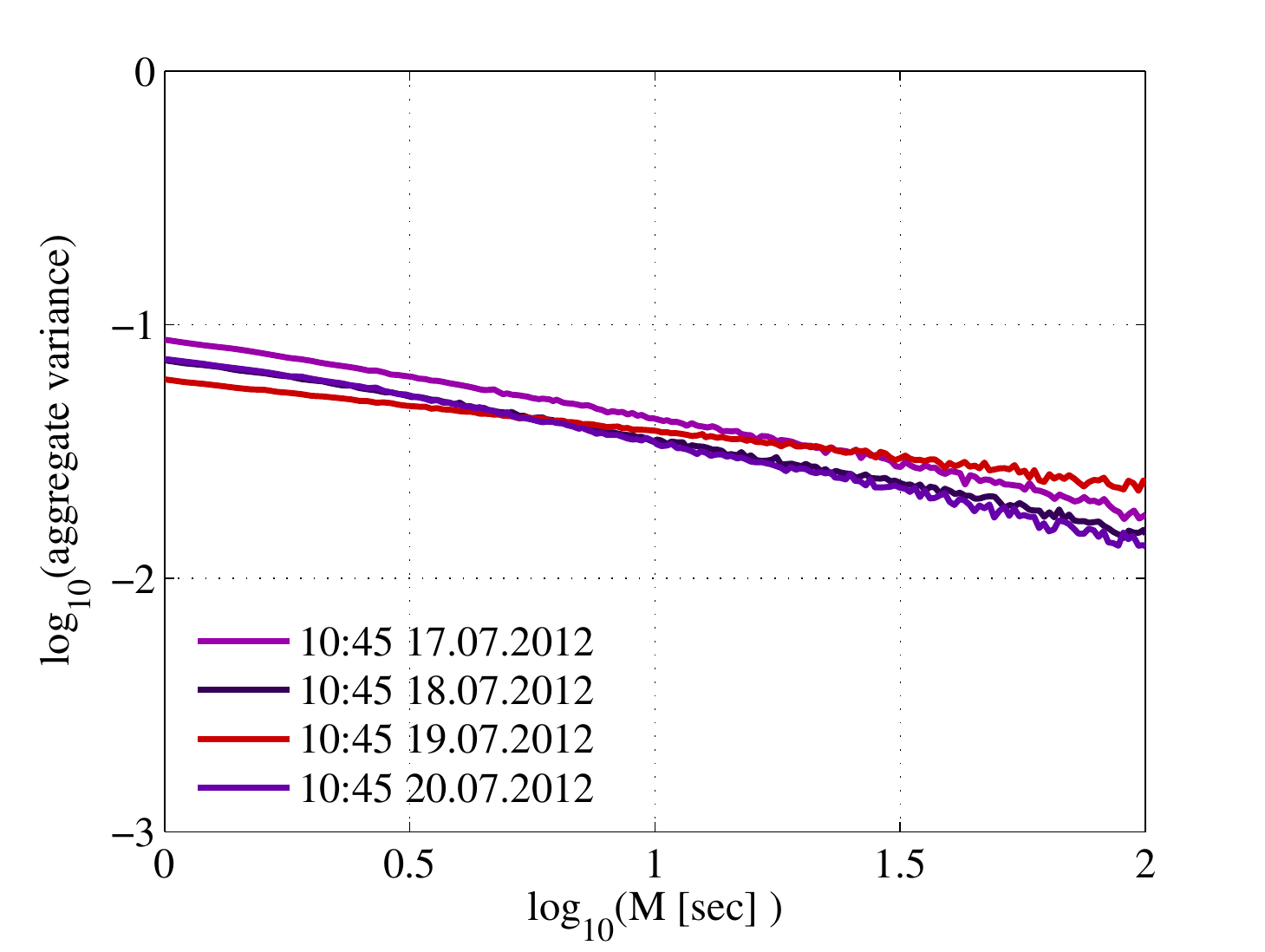}
\captionof{figure}{Aggregate variance estimates at 10:45 UTC}
\end{minipage}
\begin{minipage}{\linewidth}
\includegraphics[type=pdf,ext=.pdf,read=.pdf,width=1.00\columnwidth]{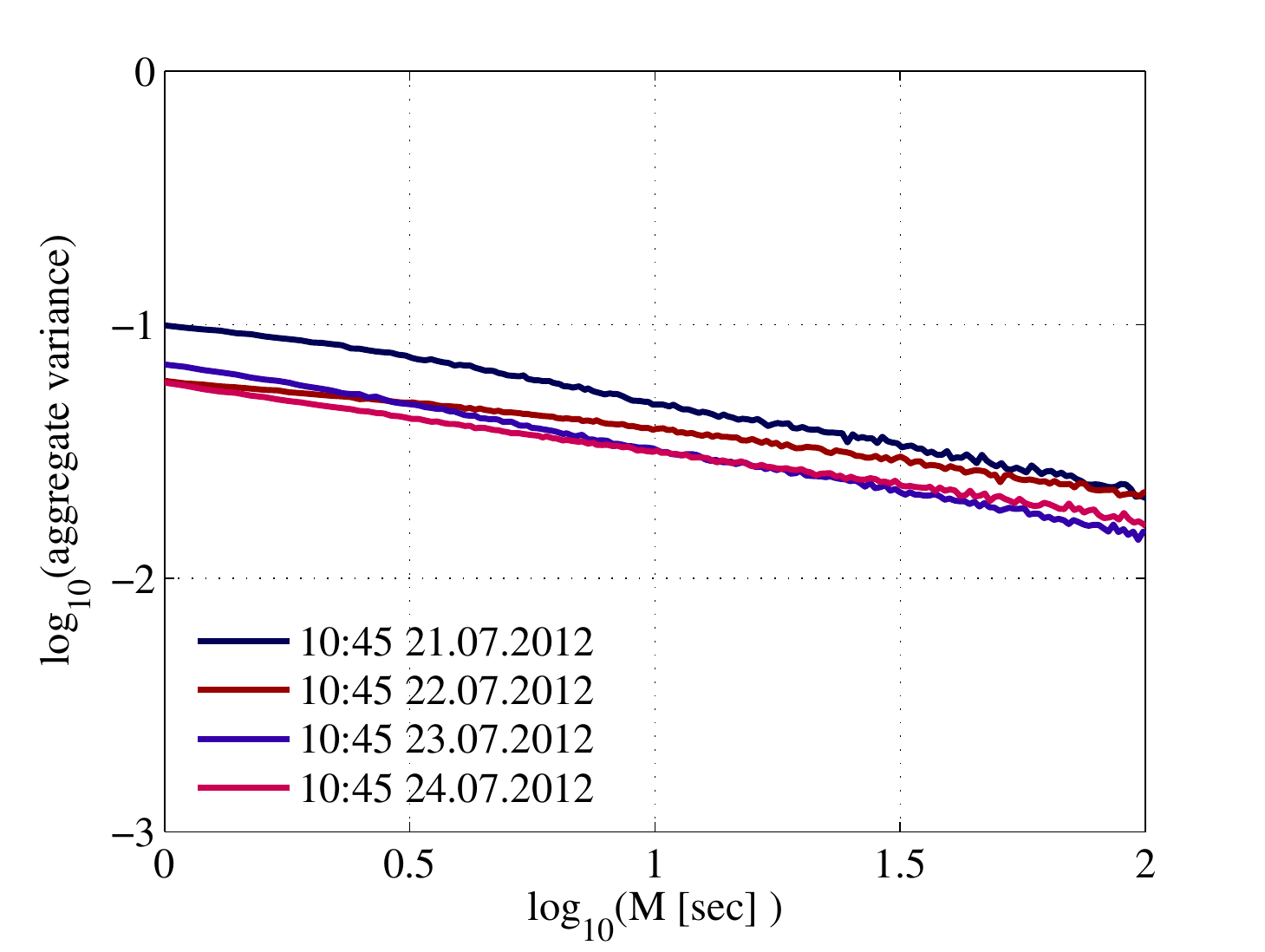}
\captionof{figure}{Aggregate variance estimates at 10:45 UTC}
\end{minipage}
%
%%
%\clearpage
\newpage
%

%\bibliographystyle{abbrv}
%\balance
%\bibliography{IEEEabrv,references_hprobe}

\begin{thebibliography}{10}

\bibitem{abramo65}
M.~Abramowitz and I.~Stegun.
\newblock {\em Handbook of Mathematical Functions}.
\newblock Dover, Dec. 1964.

\bibitem{baccelli:probing07}
F.~Baccelli, S.~Machiraju, D.~Veitch, and J.~Bolot.
\newblock On optimal probing for delay and loss measurement.
\newblock In {\em Proc. of IMC}, pages 291--302, 2007.

\bibitem{baccelli:pasta}
F.~Baccelli, S.~Machiraju, D.~Veitch, and J.~Bolot.
\newblock The role of {PASTA} in network measurement.
\newblock {\em {IEEE/ACM} Trans. Netw.}, 17(4):1340--1353, 2009.

\bibitem{beran94}
J.~Beran.
\newblock {\em Statistics for Long-Memory Processes}.
\newblock Chapman \& Hall/CRC, Oct. 1994.

\bibitem{hprobe:software}
Z.~Bozakov, A.~Rizk, and M.~Fidler.
\newblock H-probe software, 2012.
\newblock Available at: \url{http://www.ikt.uni-hannover.de/h-probe}.

\bibitem{broido:03}
A.~Broido, R.~King, E.~Nemeth, and K.~Claffy.
\newblock Radon spectroscopy of inter-packet delay.
\newblock In {\em Proc. of High Speed Networking Workshop}, 2003.

\bibitem{cox66}
D.~Cox and P.~Lewis.
\newblock {\em The statistical analysis of series of events}.
\newblock Methuen's Statistical Monographs, 1966.

\bibitem{crovella97}
M.~Crovella and A.~Bestavros.
\newblock Self-similarity in {W}orld {W}ide {W}eb traffic: evidence and
  possible causes.
\newblock {\em {IEEE/ACM} Trans. Netw.}, 5(6):835--846, Dec. 1997.

\bibitem{dovrolis:dispersion01}
C.~Dovrolis, P.~Ramanathan, and D.~Moore.
\newblock What do packet dispersion techniques measure?
\newblock In {\em Proc. of {INFOCOM}}, pages 905--914, 2001.

\bibitem{duffield94}
N.~Duffield and N.~O'Connell.
\newblock Large deviations and overflow probabilities for the general
  single-server queue, with applications.
\newblock {\em Math. Proc. Camb. Phil. Soc.}, 118(2):363--375, Sept. 1995.

\bibitem{erramilli96}
A.~Erramilli, O.~Narayan, and W.~Willinger.
\newblock Experimental queueing analysis with long-range dependent packet
  traffic.
\newblock {\em {IEEE/ACM} Trans. Netw.}, 4(2):209--223, 1996.

\bibitem{feldmann:iptraffic}
A.~Feldmann, A.~C. Gilbert, P.~Huang, and W.~Willinger.
\newblock Dynamics of {IP} traffic: A study of the role of variability and the
  impact of control.
\newblock In {\em Proc. of {SIGCOMM}}, pages 301--313, Aug. 1999.

\bibitem{ganesh:04}
A.~Ganesh, N.~O'Connell, and D.~Wischik.
\newblock {\em Big Queues}.
\newblock Springer, 2004.

\bibitem{stirzaker01}
G.~Grimmet and D.~Stirzaker.
\newblock {\em Probability and Random Processes}.
\newblock Oxford {U}niversity {P}ress, 2001.

\bibitem{gupta:09:coexistance}
H.~Gupta, A.~Mahanti, and V.~Ribeiro.
\newblock Revisiting coexistence of poissonity and self-similarity in internet
  traffic.
\newblock In {\em Proc. of MASCOTS}, pages 1--10, Sept. 2009.

\bibitem{he03}
G.~He and J.~Hou.
\newblock On exploiting long range dependency of network traffic in measuring
  cross traffic on an end-to-end basis.
\newblock In {\em Proc. of {INFOCOM}}, pages 1858--1868, 2003.

\bibitem{jacobson:97}
V.~Jacobson.
\newblock Pathchar: A tool to infer characteristics of internet paths, Apr.
  1997.

\bibitem{Jain:avbw:02}
M.~Jain and C.~Dovrolis.
\newblock End-to-end available bandwidth: measurement methodology, dynamics,
  and relation with {TCP} throughput.
\newblock {\em {IEEE/ACM} Trans. Netw.}, 11(4):537--549, Aug. 2003.

\bibitem{leland94}
W.~Leland, M.~Taqqu, W.~Willinger, and D.~Wilson.
\newblock On the self-similar nature of {E}thernet traffic.
\newblock {\em {IEEE/ACM} Trans. Netw.}, 2(1):1--15, Feb. 1994.

\bibitem{liebeherr12:heavytail}
J.~Liebeherr, A.~Burchard, and F.~Ciucu.
\newblock Delay bounds in communication networks with heavy-tailed and
  self-similar traffic.
\newblock {\em {IEEE} Trans. Inf. Theory}, 58(2):1010--1024, 2012.

\bibitem{liu05}
X.~Liu, K.~Ravindran, and D.~Loguinov.
\newblock What signals do packet-pair dispersions carry?
\newblock In {\em Proc. of {INFOCOM}}, pages 281--292, 2005.

\bibitem{liu07}
X.~Liu, K.~Ravindran, and D.~Loguinov.
\newblock A queueing-theoretic foundation of available bandwidth estimation:
  Single-hop analysis.
\newblock {\em {IEEE/ACM} Trans. Netw.}, 15(4):918--931, Aug. 2007.

\bibitem{loiseau:selfsimilarity10}
P.~Loiseau, P.~Goncalves, G.~Dewaele, P.~Borgnat, P.~Abry, and P.~Primet.
\newblock Investigating self-similarity and heavy-tailed distributions on a
  large-scale experimental facility.
\newblock {\em {IEEE/ACM} Trans. Netw.}, 18(4):1261--1274, Aug. 2010.

\bibitem{machiraju07}
S.~Machiraju, D.~Veitch, F.~Baccelli, and J.~Bolot.
\newblock Adding definition to active probing.
\newblock {\em Computer Communication Review}, 37(2):17--28, 2007.

\bibitem{mandjes07}
M.~Mandjes.
\newblock {\em Large Deviations for {G}aussian Queues}.
\newblock Wiley \& Sons, 2007.

\bibitem{massoulie99}
L.~Massoulié and A.~Simonian.
\newblock Large buffer asymptotics for the queue with {FBM} input.
\newblock {\em Applied Probability}, 36(3):894--906, Sept. 1999.

\bibitem{Melamed:asta90}
B.~Melamed and W.~Whitt.
\newblock On arrivals that see time averages.
\newblock {\em Oper. Res.}, 38(1):156--172, Feb. 1990.

\bibitem{norros95}
I.~Norros.
\newblock On the use of fractional {B}rownian motion in the theory of
  connectionless networks.
\newblock {\em {IEEE} J. Sel. Areas Commun.}, 13(6):953--962, Aug. 1995.

\bibitem{RFC2330}
V.~Paxson, G.~Almes, J.~Mahdavi, and M.~Mathis.
\newblock {RFC2330 - Framework for IP Performance Metrics}.
\newblock {\url{http://www.rfc-editor.org/rfc/rfc2330.txt}}, 1998.

\bibitem{paxon95}
V.~Paxson and S.~Floyd.
\newblock Wide-area traffic: The failure of {P}oisson modeling.
\newblock {\em {IEEE/ACM} Trans. Netw.}, 3(3):226--244, 1995.

\bibitem{viano:random_sampling}
A.~Philippe and M.Viano.
\newblock Random sampling of long-memory stationary processes.
\newblock {\em Journal of Statistical Planning and Inference},
  140(5):1110--1124, 2010.

\bibitem{ribeiro:cross_traffic_est00}
V.~Ribeiro, M.~Coates, R.~Riedi, S.~Sarvotham, B.~Hendricks, and R.~Baraniuk.
\newblock Multifractal cross-traffic estimation.
\newblock In {\em Proc. of ITC Conference on IP Traffic, Modeling and
  Management}, Sep. 2000.

\bibitem{ribeiro:sampling06}
V.~Ribeiro, R.~H. Riedi, and R.~Baraniuk.
\newblock Optimal sampling strategies for multiscale stochastic processes.
\newblock {\em IMS Lecture Notes - Monograph Series}, 49, Jan. 2006.

\bibitem{ribeiro:multiscalequeueing}
V.~J. Ribeiro, R.~H. Riedi, and R.~Baraniuk.
\newblock Multiscale queuing analysis.
\newblock {\em {IEEE/ACM} Trans. Netw.}, 14(5):1005--1018, 2006.

\bibitem{rizk12:e2efbm}
A.~Rizk and M.~Fidler.
\newblock Non-asymptotic end-to-end performance bounds for networks with long
  range dependent fbm cross traffic.
\newblock {\em Computer Networks}, 56(1):127--141, 2012.

\bibitem{roughan05}
M.~Roughan.
\newblock Fundamental bounds on the accuracy of network performance
  measurements.
\newblock In {\em Proc. of SIGMETRICS}, pages 253--264, 2005.

\bibitem{Roughan:comparison06}
M.~Roughan.
\newblock A comparison of {P}oisson and uniform sampling for active
  measurements.
\newblock {\em {IEEE} J. Sel. Areas Commun.}, 24(12):2299--2312, 2006.

\bibitem{spruce}
J.~Strauss, D.~Katabi, and F.~Kaashoek.
\newblock A measurement study of available bandwidth estimation tools.
\newblock In {\em Proc. of IMC}, pages 39--44, 2003.

\bibitem{taqqu:estimators}
M.~Taqqu, V.~Teverovsky, and W.~Willinger.
\newblock Estimators for long-range dependence: An empirical study.
\newblock {\em Fractals}, 3(4):785--798, 1995.

\bibitem{taqqu97}
M.~Taqqu, W.~Willinger, and R.~Sherman.
\newblock Proof of a fundamental result in self-similar traffic modeling.
\newblock {\em Comput. Commun. Rev.}, 27(2):5--23, Apr. 1997.

\bibitem{bintariq05}
M.~B. Tariq, A.~Dhamdhere, C.~Dovrolis, and M.~Ammar.
\newblock Poisson versus periodic path probing (or, does {PASTA} matter).
\newblock In {\em Proc. of IMC}, pages 119--124, 2005.

\bibitem{veitch99}
D.~Veitch and P.~Abry.
\newblock A wavelet-based joint estimator of the parameters of long-range
  dependence.
\newblock {\em {IEEE} Trans. Inf. Theory}, 45(2):878--897, Apr. 1999.

\bibitem{Veitch03}
D.~Veitch, P.~Abry, and M.~Taqqu.
\newblock On the automatic selection of the onset of scaling.
\newblock {\em Fractals}, 11(4):377--390, 2003.

\bibitem{veitch:05:multifractal}
D.~Veitch, N.~Hohn, and P.~Abry.
\newblock Multifractality in {TCP}/{IP} traffic: the case against.
\newblock {\em Computer Networks}, 48:293--313, 2005.

\bibitem{willinger97}
W.~Willinger, M.~Taqqu, R.~Sherman, and D.~Wilson.
\newblock Self-similarity through high-variability: statistical analysis of
  {E}thernet {LAN} traffic at the source level.
\newblock {\em {IEEE/ACM} Trans. Netw.}, 5(1):71--86, Feb. 1997.

\bibitem{wolff:pasta}
R.~Wolff.
\newblock Poisson arrivals see time averages.
\newblock {\em Operations Research}, 30(2):223--231, 1981.

\end{thebibliography}
%\input{appendix.tex}
%
%%%%-------------------------------------------------------------------------
%
% %%%Generated by IEEEtran.bst, version: 1.12 (2007/01/11)
%

%
%%%%------------------------------------------------------------------------
%
\end{document}